\documentclass[%
reprint, 
longbibliography,
superscriptaddress,
nofootinbib,
amsmath,amssymb,
aps,
prx
]{revtex4-1}

\usepackage{color}

\usepackage[utf8]{inputenc}
\usepackage{amsmath}
\usepackage{amsthm}
\usepackage{enumitem}
\usepackage{amssymb}
\usepackage{wasysym}
\usepackage{oplotsymbl}
\usepackage{bm} 
\usepackage{dsfont} 
\usepackage{mathrsfs} 
\usepackage{physics}
\usepackage{mathtools}

\usepackage{xcolor}
\usepackage{framed}
\definecolor{darkblue}{rgb}{0,0,.65}
\definecolor{darkgreen}{rgb}{0.28,0.41,0.19}
\definecolor{nicegreen}{rgb}{0.28,0.85,0.19}

\usepackage[colorlinks=true,linkcolor=blue,allcolors=blue]{hyperref}
\usepackage[capitalize,nameinlink]{cleveref}

\def\equationautorefname~#1\null{Eq. (#1)\null}

\renewcommand\vec{\bm}

\DeclareMathOperator\linspan{span}

\newcommand{\rhoG}{\ensuremath{\hat \rho_{\rm G}}}

\newcommand{\rhoGV}{\ensuremath{\hat \rho_{G,\mathcal{V}}}}
\newcommand{\Sconf}{\ensuremath{S_{\rm conf}}}
\newcommand{\Hcheck}{\ensuremath{H}}
\newcommand{\hamil}{\ensuremath{\mathcal H}}

\newtheorem{prototheorem}{Theorem}[section]

\newtheorem{protolemma}[prototheorem]{Lemma}

\newenvironment{theorem}{\colorlet{shadecolor}{gray!15}\begin{shaded}\begin{prototheorem}}
{\end{prototheorem}\end{shaded}}

\begin{document} 

\title{Topological Quantum Spin Glass Order and its realization in qLDPC codes} 

\author{Benedikt Placke}%
\affiliation{Rudolf Peierls Centre for Theoretical Physics, University of Oxford, Oxford OX1 3PU, United Kingdom}
\author{Tibor Rakovszky}
\affiliation{Department of Physics, Stanford University, Stanford, California 94305, USA},
\affiliation{Department of Theoretical Physics, Institute of Physics,
Budapest University of Technology and Economics, M\H{u}egyetem rkp. 3., H-1111 Budapest, Hungary}
\affiliation{HUN-REN-BME Quantum Error Correcting Codes and Non-equilibrium Phases Research Group,
Budapest University of Technology and Economics,
M\H{u}egyetem rkp. 3., H-1111 Budapest, Hungary}
\author{Nikolas P.\ Breuckmann}
\affiliation{School of Mathematics, University of Bristol, Bristol BS8 1UG, United Kingdom}
\author{Vedika Khemani}
\affiliation{Department of Physics, Stanford University, Stanford, California 94305, USA}

\begin{abstract}
Ordered phases of matter have close connections to computation. Two prominent examples are spin glass order, with wide-ranging applications in machine learning and optimization, and topological order, closely related to quantum error correction. Here, we introduce the concept of \emph{topological quantum spin glass} (TQSG) order which marries these two notions, exhibiting both the complex energy landscapes of spin glasses, and the quantum memory and long-range entanglement characteristic of topologically ordered systems. Using techniques from coding theory and a quantum generalization of Gibbs state decompositions, we show that TQSG order is the low-temperature phase of various quantum LDPC codes on expander graphs, including hypergraph and balanced product codes. Our work introduces a topological analog of spin glasses that preserves quantum information, opening new avenues for both statistical mechanics and quantum computer science.
\end{abstract}

\maketitle

\section{Introduction}

\begin{figure*}
    \centering
    \includegraphics{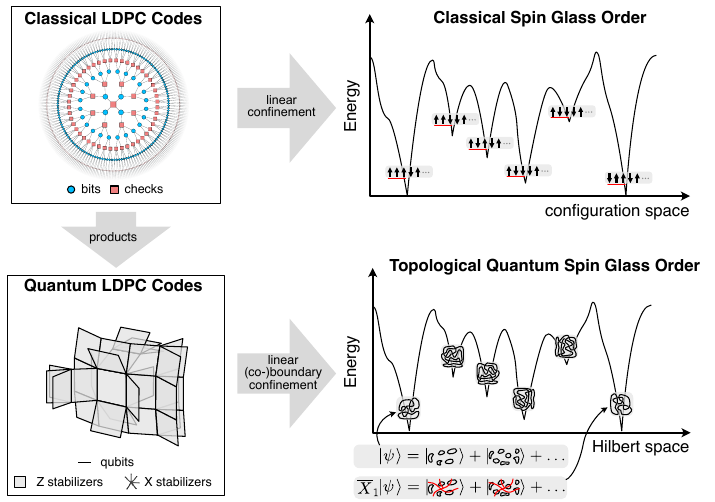}
    \caption{{\bf Classical and quantum LDPC codes and their low-temperature phases.} We consider classical and quantum LDPC codes defined on so-called expander graphs, where the boundary of a region scales proportionally to its volume. Classical expander codes are known to realize classical spin glass order, characterized by a complex (free) energy landscape with many global and local minima, which can be used to passively store classical information. Here, we show that quantum LDPC codes realize what we call \emph{topological quantum spin glass} order, characterized by a similarly complex landscape, but with local minima corresponding to topologically ordered quantum states, capable of passively preserving quantum information. 
    This follows from a property of the underlying codes, called \emph{linear confinement}, which also guarantees their efficient decodability. The quantum codes can be constructed out of their classical counterparts using hypergraph and balanced \emph{product constructions}.
}
    \label{fig:toc}
\end{figure*}

The traditional understanding of phases of matter is based on Landau's theory of spontaneous symmetry breaking (SSB). Many significant breakthroughs in many-body physics over the last fifty years, however, have been centered on two broad classes of systems which lie outside this paradigm:
(i) spin glasses~\cite{anderson1970localisation, edwards1975theory, sherrington1975solvable,parisi1979solution2,hopfield1982network,binder1986review,mezard2009information,stein2013spin} and (ii) topological order~\cite{wegner1972,kosterlitz1973ordering,klitzing1980qhe,tsui1982fqhe,laughlin1983wavefunction,haldane1988model,zeng2019quantum,moessner2021topological}. In this work, we formulate a new phase of matter which combines features of both classes---\emph{a topological quantum spin glass}---and show that this phase is explicitly realized in physical systems  implementing certain quantum error correcting codes.

A useful perspective on these different kinds of phases is in terms of their ability to serve as \emph{memories} for storing (classical or quantum) information. A thermally stable phase with spontaneous symmetry breaking can be used as a robust memory to store \emph{classical} information, which is the working principle behind a magnetic hard drive. In a ferromagnet below its critical temperature~$T_c$, there exist two stable equilibrium states with opposite magnetization, which can be used to store the `0' and `1' states of an encoded classical bit. 
Thermal stability of SSB means that the system can can act as a \emph{passive} memory: this bit is not corrupted, despite the presence of thermal fluctuations as long as the system is in an equilibrium state below $T_c$.

Spin glasses, on the other hand, furnish a different mechanism for classical memory that goes beyond the paradigm of SSB \cite{parisi1979solution2}. The term ``glassiness" has been used to refer to a broad range of physical phenomena associated with slow dynamics, in models ranging from structural glasses \cite{debenedetti_stillinger2001, tarjus2005frustration, bertier2011rmp} to frustrated magnets \cite{anderson1978concept,binder1986review,bramwell2001review}
and constrained systems \cite{newman_moore1999,garrahan2011kinetically,chamon2005quantum_glass,haah2011fractons,Castelnovo2012topo_glass}. While most of these problems are notoriously challenging to understand, certain models of spin glasses with geometrically non-local interactions are more amenable to analytic treatment \cite{sherrington1975solvable,parisi1979solution2,mezard2003two_solutions,dubois2002xorsat, montanari2006bethe}.  
In these cases, long-lived memory is attributed to a complex ``rugged" free-energy landscape with numerous global and local minima. These minima are not necessarily related to each other by symmetries, but are nevertheless associated with asymptotically long-lived equilibrium states below a glass ordering temperature $T_G$; they can be used to robustly store many classical bits (\cref{fig:toc}). 
While the applicability of these models in capturing the behavior of local, finite-dimensional spin glasses is still under debate~\cite{fisher_huse1986,Newman_2024}, by now they have transcended their original purpose and found numerous applications in complex systems research and computer science \cite{mezard2009information}, most notably in machine learning~\cite{hopfield1982network,amit1985sg_nn,stein2013spin,Krzakala2024machine_learning}, complexity theory \cite{mezard2009information,stein2013spin,monasson1999complexity_transitions,achlioptas2008algorithmic,mezard2002survey_propagation,krzakala2007gibbs,ricci_tersenghi2010xorsat}, and classical error correction~\cite{franz2022dynamic,di2004weight_enumerators,mezard2009information}.

A \emph{quantum memory}, on the other hand, requires the ability to robustly store a coherent quantum state, including arbitrary superpositions of the `0' and `1' states of the encoded `logical' qubits \cite{shor1995scheme,terhal2015quantum,brown2016quantum}. 
To protect this delicate ``Schr\"odinger's cat'' from the decohering effects of the environment,  the state of the encoded qubits needs to be stored in a collection of states that are all locally indistinguishable from each other, making the encoded information inaccessible to the environment~\cite{knill1997theory}. 
Such local indistinguishability is a defining feature of \emph{topological order}~\cite{bravyi2006lieb,chen2010local,bravyi2010topological}.
Topological order is often discussed in the context of zero-temperature phases, which exhibit multiple locally-indistinguishable long-range entangled ground states on topologically non-trivial manifolds. While these ground states can be used to encode logical quantum states, a robust passive memory requires thermally stable topological order, which is only known to exist in spatial dimensions four or higher~\cite{dennis2002topological,brown2016quantum, alicki2010thermal, hastings2011topological, yoshida2011feasibility, liu2024dissipative}.

Here we propose, and instantiate, a phase of matter which we call a \emph{topological quantum spin glass}, which combines the \emph{classical complexity} of the free-energy landscape of spin glasses with the  \emph{quantum complexity} and entanglement structure of topologically ordered states (\cref{fig:toc}).\footnote{We note that the phrase ``topological quantum glassiness" has been used in the context of fracton phases \cite{chamon2005quantum_glass,haah2011fractons,Castelnovo2012topo_glass} and disordered toric codes \cite{Castelnovo2012topo_glass,tsomokos2011interplay} However, there are crucial differences: these examples are about zero temperature physics, and do not display thermally stable topological order or the complex free energy landscapes that are characteristic of the TQSG phase we define.} 
We formally define the phase by generalizing the theory of extremal Gibbs state decompositions \cite{georgii2011gibbs,friedli2017statistical,mezard2009information} to quantum systems, which may be of independent interest.

Our examples are based on quantum low density parity check (qLDPC) codes in non-Euclidean, so-called \emph{expanding} geometries, with a property called \emph{linear confinement} \cite{gromov2010singularities,linial2006homological,bombin2015single_shot, leverrier2015quantum_expander}; these include recently discovered ``good'' qLDPC codes with optimal coding parameters~\cite{breuckmann2021balanced,panteleev2022qldpc,dinur2023qldpc}.
We establish that these models have energy landscapes with exponentially many (in system size)  local minima, each hosting a stable equilibrium state below some critical temperature $T_G$. 
While these include global ground states related by symmetries (of which there are already exponentially many), most local minima are at finite energy density and not related by symmetries. The complex energy landscape of these quantum codes is similar to that of their classical counterparts, which realize classical spin glass order~\cite{franz2001ferromagnet,franz2022dynamic,mezard2009information,placke2024classical_glass} (\cref{fig:toc}). Notably, these models need not have any randomness or frustration, properties that are often instrinsically associated with glassy behavior \cite{anderson1978concept, tarjus2005frustration}. 
However, in contrast to classical spin glasses, in TQSGs, the equilibrium (mixed) states associated with typical minima display a particularly strong form of topological order: they are supported on subspaces that contain no topologically trivial (short-range entangled) states at all~\cite{freedman2014nlts,anshu2022nlts}. This allows not just the set of ground states to encode a stable quantum memory -- rather, each finite energy-density minima is  \emph{also} typically part of an emergent stable quantum memory. 

Our results establish that the low temperature physics of Hamiltonians associated to these qLDPC codes \cite{rakovszky2023physics, rakovszky2023physics2} is described by TQSG order, similar to how topological order describes more conventional quantum error correcting codes in Euclidean geometries.
Our results demonstrate that non-Euclidean geometries host unique quantum phases of matter, at a time when rapid advances pave the way for their experimental implementation~\cite{kollar2019hyperbolic,avikar2021programmable,bluvstein2024logical,xu2024constant,ramette2022any_to_any,bravyi2024qldpc}. They also suggest potential applications of spin glass theory to quantum computer science, in analogy to its wide-ranging impact in classical computer science. 

\section{Topological Quantum Spin Glass order\label{sec:glass}}

We now present a more formal definition of TQSG order by formulating a generalization of the theory of Gibbs state decompositions to quantum systems. A central tenet of statistical mechanics is that the behavior of physical systems in thermal equilibrium  is described by the \emph{Gibbs state}, $\rhoG \propto e^{-\beta \hamil}$, where $\hamil$ is the system's Hamiltonian and $\beta$ the inverse temperature. One characteristic feature of many non-trivial finite-temperature phases is that the Gibbs state ceases to be unique in the thermodynamic limit. The arguably simplest example of this occurs in the classical two-dimensional Ising model below its critical temperature, which exhibits at least two different Gibbs states with different average magnetizations \cite{friedli2017statistical}. This idea of ``multiple Gibbs states'' can be formalized in different ways~\cite{georgii2011gibbs, friedli2017statistical, mezard2009information}; in this work, we take an approach that has proven useful in the context of classical spin glasses~\cite{mezard2009information} and we generalize it to quantum systems. 
Specifically, we ask whether the ``global" Gibbs state $\rhoG$ fragments into multiple different components separated by macroscopically divergent \emph{free energy barriers}. If it does, then each component can be thought of as a distinct Gibbs state in the thermodynamic limit. In the classical setting, this notion of Gibbs state components is closely related to the ``replicas'' appearing in spin glass theory~\cite{franchini2023replica,mezard2009information}.

Let $\mathcal{V}$ be a subspace of the many-body Hilbert space and $P_\mathcal{V}$ the projector onto this subspace. 
We define its free energy as $F(\mathcal{V}) = -\log(\tr(P_\mathcal{V}e^{-\beta \hamil})) / \beta$.  
We will need a notion of the ``boundary'', $\partial \mathcal{V}$, of $\mathcal{V}$, which we informally define as the subspace spanned by all states outside of $\mathcal V$ that can be obtained from states in $\mathcal{V}$ by local operations (see supplementary material~\cite{SM_quantum} for a more formal definition).
We then say that $\mathcal{V}$ is surrounded by a macroscopic free energy barrier if the following \emph{bottleneck condition} is satisfied\footnote{While this condition is appropriate to the stabilizer models considered in this paper, in general one needs a stricter condition that also involves a bound on the off-diagonal parts of $\rho_G$; see Refs. \onlinecite{rakovszky2024bottleneck, garmanik2024slow} and the supplementary material~\cite{SM_quantum}}:
\begin{equation}\label{eq:bottleneck_condition}
    \frac{\tr(P_{\partial\mathcal{V}}\rhoG)}{\tr(P_\mathcal{V}\rhoG)} = e^{-\beta(F(\partial\mathcal{V})-F(\mathcal{V}))} \equiv \Delta(\mathcal{V}) \to 0 \,\, \text{ as } n\to\infty,
\end{equation}
where $n$ is the number of degrees of freedom. The Gibbs state component corresponding to $\mathcal{V}$ is $\rhoGV \propto P_\mathcal{V} \rhoG P_\mathcal{V}$, with weight $\tr(P_\mathcal{V}\rhoG)$. The component is \emph{extremal} if~$\mathcal{V}$ does not contain multiple orthogonal subspaces which separately satisfy \cref{eq:bottleneck_condition}. 

Intuitively, $\mathcal{V}$ corresponds to states near some energy minimum. Importantly, the definition allows for cases where this is only a \emph{local} minimum [\cref{fig:expansion}(b)]. Unlike the global ground states, most local minima are not related by symmetries. This especially includes cases where the subspace $\mathcal{V}$ only contains states at finite energy density, in which case we call the corresponding Gibbs state \emph{incongruent}, adapting the terminology in Ref. \onlinecite{huse_fisher1987incongruent} to indicate that these states (statistically) look locally different from any ground state. 

The importance of \cref{eq:bottleneck_condition} is underlain by the \emph{bottleneck theorem} proven in Ref. \onlinecite{rakovszky2024bottleneck}, generalizing a well-known result from classical Markov chains~\cite{levin2017markov} to quantum systems (see also \cite{garmanik2024slow}). 
In the present context, it states that if $\mathcal{M}$ is any local quantum channel that has $\rhoG$ as its steady state (e.g.\ by virtue of obeying quantum detailed balance~\cite{frigerio1977quantum,fagnola2007generators,temme2010chi,chen2311efficient,gilyen2024quantum,ding2024efficient}), then the state $\rhoGV$ will also be an approximate steady state, remaining stationary up to times $t_{\mathcal{V}} \propto 1/\Delta(\mathcal{V})$. This establishes a close connection between thermodynamic properties, encoded in the free energy landscape, and dynamical ones. In particular, it shows that the different Gibbs state components can be used as a \emph{passive memory}: information encoded in them will be preserved forever in the thermodynamic limit. 

The amount and nature (classical vs.\ quantum) of this memory is encoded in the space of all approximate steady states $\rhoGV$, which form a convex set. 
In a classical memory, all extremal components $\mathcal{V}_i$ are orthogonal and any equilibrium state can be written as their classical mixture, $\sum_{i=1}^\mathcal{N} p_i \hat \rho_{\text{G},\mathcal{V}_i}$. 
This can be true even if $\hamil$ itself involves non-commuting terms; e.g.\ when a weak transverse field is added to a classical Hamiltonian. 
Thus, the system preserves an $\mathcal{N}-$dimensional probability vector, corresponding to $\log_2(\mathcal{N})$ bits of classical memory. 

A quantum memory, on the other hand, preserves the state of a general $D$-dimensional qudit (or, equivalently, $k=\log_2(D)$ qubits). In this case, each steady state is a \emph{density matrix} over the extremal components, which can involve quantum coherences; specifically, it has the form $\mu \otimes \sigma$, where $\mu$ is a $D\times D$ density matrix containing the encoded information, while $\sigma$ is the same for all steady states~\cite{blume2010information, baumgartner2012structures, SM_quantum}. Different pure states $\mu$ correspond to mutually non-orthogonal extremal Gibbs states. 

In the most general case, relevant for the models discussed here, the system may act as a ``hybrid memory"~\cite{dauphinais2024stabilizerformalism} and preserve a combination of classical and quantum information:  its steady states can be written as $\sum_i p_i \mu_i \otimes \sigma_i$, see ~\cite{blume2010information, baumgartner2012structures, SM_quantum}. This defines an effective density matrix over extremal states, $\rho^{\rm eff}\equiv \bigoplus_i p_i \mu_i$. The von-Neumann entropy of $\rho^{\rm eff}$, $S_{\rm conf} \equiv s_{\rm conf} n$, defines the \emph{configurational entropy} which quantifies the number of relevant Gibbs components. 
This generalizes a related measure from classical spin glasses, which quantifies the \emph{classical complexity} of the free-energy landscape~\cite{mezard2009information,franz2022dynamic}.  

Finally, defining topological order in mixed states is a subtle issue~\cite{hastings2011topological, chen2024separability}.
A pure state is topologically ordered if it is long-range entangled (LRE), which means it cannot be (approximately) obtained by evolving an uncorrelated product state for a finite amount of time with a local Hamiltonian or, relatedly, by a local unitary circuit.  A state is ``trivial" or short-ranged entangled (SRE) otherwise. The circuit complexity, i.e.\ the circuit depth needed to approximate a state, defines a notion of \emph{quantum complexity}.  
A natural generalization of these concepts defines a mixed state $\hat\rho$ as topologically ordered or LRE if it cannot be expressed as a convex sum of SRE pure states, meaning that it cannot be written as $\hat\rho =\sum_i p_i |\psi_i\rangle \langle\psi_i|$ for each $|\psi_i\rangle$ being SRE \cite{hastings2011topological}. 
For example, at low enough temperatures in the 4D toric code, the global Gibbs state $\rhoG$ is topologically ordered in this sense~\cite{alicki2010thermal}. A stronger topological property is the no-low energy trivial states (NLTS) condition~\cite{freedman2014nlts,anshu2022nlts}: the \emph{full} set of states below some sufficiently small cutoff in energy density\footnote{These states need not be eigenstates or equilibrium states.} contains no trivial states \emph{at all}. In the models we study, we will show that the subspaces $\mathcal{V}$ supporting typical Gibbs state components are topologically ordered in this stronger sense. 

With this preparation, we now state our main result:
\begin{theorem}[Existence of TQSGs, informal]\label{thm:main}
There exist families of sparse Hamiltonians 
$\hamil$ which realize all of the following properties for temperatures $T < T_G$:
\begin{itemize}
    \item \textbf{Shattering:} There are exponentially many distinct Gibbs state components, each with exponentially small weight, with lifetimes diverging as a (stretched) exponential of system size.
    \item \textbf{Incongruence:} The configurational entropy density $s_{\rm conf}$ grows with $T$; most Gibbs state components $\rhoGV$ are associated with subspaces $\mathcal{V}$ that only contain finite energy-density states.
    \item \textbf{Topological order:} A typical Gibbs state component $\rhoGV$ is LRE; in fact,  the associated subspace $\mathcal{V}$ contains no trivial SRE states. 
    \item \textbf{Emergent Quantum Memories:}
     The system acts as a hybrid memory, passively encoding a combination of classical and quantum information: the collection of all Gibbs state components, including those containing no ground states, can be grouped into sets, with each set defining a stable passive quantum memory encoding ${k \propto n}$ qubits. 
\end{itemize}
\label{thm:TQSG}
\end{theorem} 

Since the typical $\rhoGV$ are LRE and these carry almost all the weight in $\rhoG$ at low enough temperatures, the global Gibbs state $\rhoG$ should also be LRE below $T_G$.
Since $\rhoG$ is known to be trivial at sufficiently high temperatures \cite{bakashi2024high}, this implies a ``separability transition" \cite{chen2024separability} in $\rhoG$  at $T_G$.  

\section{Topological Spin Glass Order in qLDPC Codes\label{sec:expansion}}

We now elaborate on the different criteria in the definition of TQSG in \cref{thm:TQSG}, and discuss how they are realized in quantum LDPC codes. 

\begin{figure*}
    \centering
    \includegraphics{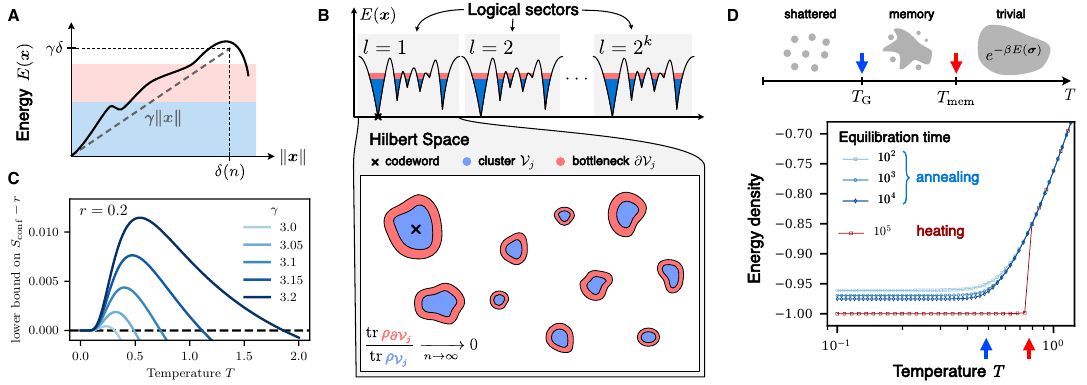}
    \caption{{\bf From linear confinement to spin glass order.} (a) Linear confinement: any excitation of size less than $\delta(n)$ costs an energy that is linear in its size.
    (b) This implies a complex energy landscape, and in particular \emph{shattering} of the low-temperature Gibbs state into components $\mathcal V_j$ separated by bottleneck regions. (c) The configurational entropy $\Sconf$, measuring the number of relevant components at a given temperature, equals code rate $r$ at $T=0$ and increases with $T$ initially. This indicates that the Gibbs state in this regime is dominated by \emph{local} minima, incongruent with the global ground states.
    (d) Two numerical experiments, where temperature is increased starting from zero (heating) and decreased from infinity (annealing) on a balanced product code with $n=218880$ qubits, suggest two distinct phase transitions. While ground states are thermally stable up to temperature $T_{\rm mem}$, the global Gibbs state is dominated by one large component until a lower temperature $T_G$.
    }
    
    \label{fig:expansion}
\end{figure*}

\vspace{-0.3cm}
\subsection{Quantum codes with linear confinement}

We study Hamiltonians associated with families of quantum error correcting codes. These are $n$-qubit commuting stabilizer Hamiltonians that can be thought of as generalizations of the well-known toric code model~\cite{kitaev2003}. They take the form
\begin{equation}\label{eq:hamil}
    \hamil = -\sum_{i=1}^{m_X} A_i - \sum_{j=1}^{m_Z} B_j, 
\end{equation}
where the stabilizers (also known as checks) $A_i$ ($B_j$) are products of single-qubit Pauli-$X$ (Pauli-$Z$) operators. These Hamiltonians satisfy a low-density-parity-check (LDPC) condition,  which enforces that every qubit only interacts finitely many others (i.e.\ every check only involves finitely many qubits, and every qubit is only part of finitely many checks). This allows for Hamiltonians with a generalized notion of locality on a graph, which may not be local on any finite-dimensional Euclidean lattice. Indeed, the focus of this paper will be on models defined on \emph{expander graphs}, which have a finite boundary-to-bulk ratio~\cite{hoory2006expander} (\autoref{fig:toc}). These can realize properties impossible in finite dimensions~\cite{bpt2010}, which we will use extensively; for example, the ground state degeneracy scales \emph{exponentially} in $n$, corresponding to an extensive ground state entropy,  and there are strongly diverging energy barriers separating different equilibrium states.      

The ground states of $\hamil$ are the simultaneous $+1$ eigenstates of the stabilizers, and the ground state degeneracy is $2^k$. In the models we consider, $k \propto n$. All ground states are \emph{locally indistinguishable}, i.e.\ the reduced density matrices of all ground states agree on subsystems of size smaller than the \emph{code distance} $d$, which diverges with $n$. The local indistinguishability property is associated with topological order in condensed matter physics \cite{bravyi2006lieb,chen2010local,bravyi2010topological}, and with the Knill-Laflamme conditions in error correction~\cite{knill1997theory}. This allows one to identify the ground state subspace with a stabilizer code of \emph{rate} $r=k/n$, encoding $k$ logical qubits. Logical operators $\{\overline{X}_l, \overline{Z}_l\}$ for $l = 1,\dotsc, k$  act as Pauli $X/Z$ operators on the $l$-th logical qubit. The logical operators are defined by exact (non-commuting) symmetries of the Hamiltonian which cannot be written as products of stabilizers, and which thus act non-trivially on the ground-state subspace. 
(In the toric code, these are the non-contractible $X/Z$ Wilson loops). The $2^k$ different symmetry sectors or `logical sectors' may be specified by the eigenvalues of the $\{\overline{Z}_l\}$ operators, while the $\{\overline{X_l}\}$ operators transform between symmetry sectors. The logical $\overline{X}/\overline{Z}$ operators 
can be  deformed into equivalent operators by multiplying with products of $A/B$ stabilizers, and the smallest logical has weight $d$. 

A feature that we will rely on in some of our arguments is a \emph{lack of redundancies}, meaning that there is is no product of the checks $A_i$, $B_j$ that equals the identity, so each check may be violated independently of the others.\footnote{We expect our argument to straightforwardly generalize to the case where the number of redundancies $n_R$ scales sub-extensively with $n$, so that $n_R/n \rightarrow 0$ as $n\rightarrow \infty$.} This lack of redundancies has a direct consequence for the thermodynamics of these Hamiltonians: under an appropriate (non-local) change of variables, $\hamil$ becomes equivalent to a trivial paramagnet of decoupled spins, making its partition function analytic at all temperatures~\cite{eggarter1974cayley, yoshida2011feasibility, montanari2006bethe, freedman2014nlts, weinstein2019universality, rakovszky2023physics, Hong2024quantum_memory}. In what follows, we will prove that these models can nevertheless realize a thermodynamically non-trivial phase below a critical temperature $T_G$, as characterized by the structure of Gibbs states decompositions in the sense of \autoref{sec:glass}. 

A crucial aspect of these models is their energy landscape, i.e. the energy required for creating and moving excitations above the ground states of $\hamil$. The excitations are point-like, in the sense that each check in $\hamil$ can be independently excited due to the lack of redundancies. In local Euclidean settings, topological order in models with point-like excitations is usually associated with \emph{deconfinement}. For example, in the 2D toric code, one can create a pair of anyonic excitations and move them arbitrarily far apart without additional energy cost. This implies only a finite $O(1)$ energy barrier to go from one ground state to another, which is why the topological order in the 2D toric code is thermally unstable. The mobility of excitations is more restricted in 3D fracton models, but they can still be separated arbitrarily far with only a weak (logarithmic) energy barrier, and these models are also thermally unstable~\cite{bravyi2013quantum_self_correction,prem2017glassy,brown2016quantum}. 

This is in stark contrast to the situation in the models we study, which exhibit a feature called \emph{linear confinement} \cite{gromov2010singularities,linial2006homological,bombin2015single_shot, leverrier2015quantum_expander}. 
Colloquially, this says that excitations cannot be moved without creating more excitations, which causes the energy to increase \emph{linearly} with the number of local moves made, up to a diverging threshold $\delta(n)$. 
Thus, even if two states have the same energy, say each with one charge but in different locations, it can cost a diverging energy barrier to go from one state to the other via a sequence of local moves. This confinement is a feature unique to expanding geometries, and it produces an energy landscape with strong energy barriers which, as we will show, allows for many stable Gibbs state components and (topological) spin glass order.      

We now define the desired features of the models in more technical terms. A succinct way of writing the $A$-type checks in \cref{eq:hamil} is by defining a binary ``parity check" matrix $\Hcheck_X \in \mathbb{F}_2^{m_x\times n}$ where $(\Hcheck_X)_{j,l} = 1$ if qubit $l\in \operatorname{supp}(A_j)$, and zero otherwise.  
Similarly for $\Hcheck_Z$ and $B$-checks. The LDPC condition enforces that $\Hcheck_Z$ and $\Hcheck_X$ are sparse; the no-redundancy condition means that they are full rank; while the commutation condition on checks amounts to $\Hcheck_X \cdot \Hcheck_Z^T = 0$ where here, and in the following, arithmetic is mod 2. The kernel of these matrices defines the code space, which is a linear vector space. The code rate is set by the number of linearly independent checks via $k = n - \rank(\Hcheck_X) - \rank(\Hcheck_Z)$.

Since $\hamil$ is a stabilizer Hamiltonian, we can write all eigenstates of $\hamil$ in the form $\ket{\vec x, \vec z} = \hat X^{\vec x} \hat Z^{\vec z}\ket{\psi}$ for some ground state $\ket{\psi}$, where $\vec x, \vec z \in \mathbb{F}_2^n$ are binary vectors, and we introduced the shorthand $\hat O^{\vec x} = \prod_i O_i^{x_i}$. The energy of the state $\ket{\vec x, \vec z}$ is simply given by $    E(\vec x, \vec z) = \abs{\Hcheck_Z \vec x} + \abs{\Hcheck_X \vec z}$,
where $\abs{\bullet}$ denotes the number of $1$s in a binary vector, called its Hamming weight; here $\Hcheck_Z \vec x \in \mathbb{F}_2^{m_Z}$ is the `$Z$ syndrome' vector, indicating the set of $Z$-checks that are violated by acting with $X$-Paulis (``errors") on the qubits selected by $\vec x$.  If two operators $\hat{X}^{\vec x_1}$ and $\hat{X}^{\vec x_2}$ can be deformed into each other by multiplying with the checks $A_i$,  they define the same state. This motivates defining the \emph{reduced weight} $||\vec x||$ as the smallest Hamming weight amongst all such vectors, which can be viewed as the minimum size of the $X$-error needed to produce a state with a given $Z$ syndrome. The discussion is analogous for $Z$ errors and $X$ syndromes.

The Hamiltonian $\hamil$ has \emph{linear confinement} if 
\begin{equation}
    \norm{\vec x} \leq \delta(n)~\Rightarrow~\abs{\Hcheck_Z\vec x} \geq \gamma \norm{\vec x}
    \label{eq:expansion}
\end{equation}
and similarly for $\Hcheck_Z$. Thus, the energy of $\ket{\vec x,\vec z}$ grows linearly with the reduced weights of $\vec x, \vec z$ with slope $\gamma >0$, up to some threshold $\delta(n)$ that diverges super-logarithmically with $n$ (i.e.\ $\delta(n)/\log(n)$ diverges as $n\to\infty$). 
Since the linear growth persists up to the threshold $\delta(n)$, it also implies that every ground state of $\hamil$ is surrounded by an energy barrier of size $\geq \gamma\delta(n)$, as sketched in \cref{fig:expansion}(a). Naturally, this also implies a lower bound on the code distance: $d > \delta(n)$.

We can now summarize the properties we require from $\hamil$, defined in \autoref{eq:hamil}, in order for it to provably satisfy \cref{thm:main}: (i) strong enough linear confinement with $\gamma > \gamma^*$ and $\delta(n)$ diverging super-logarithmically with $n$;  (ii) the absence of redundancies; and (iii) a constant code-rate, $k \propto n$. 
The constant $\gamma^*$ in (i) depends on the rate, e.g. we have $\gamma^* \approx 3$ for $r=1/15$.

One way of obtaining quantum Hamiltonians satisfying (i)-(iii) is by applying the so-called \emph{hypergraph product} (HGP) construction~\cite{tillich2009hgp} to a pair of suitable classical Hamiltonians on expander graphs, corresponding to good classical expander codes. Such HGP codes provably have all the desired properties (with $\delta(n) \propto \sqrt{n}$) and therefore realize TQSG order~\cite{SM_quantum}. 
Another family, good codes~\cite{panteleev2022qldpc,dinur2023qldpc} obtained via \emph{balanced products} (BP)~\cite{breuckmann2021balanced}, on the other hand, can realize optimal scaling for linear confinement, $\delta(n) \propto n$, which makes their analysis more straightforward; however the provable bounds on $\gamma$ and the number of redundancies in this case are weaker than what we require. Nevertheless, these provable bounds, as well as our requirements, are likely to be loose, and so we expect TQSG order to be realized in BP codes as well.
This expectation is born out by large scale numerics (\autoref{sec:shattering}). For similar reasons, we also expect TQSG order to be realized in other families of qLDPC codes on expander graphs.   

\subsection{Typical low-temperature states are surrounded by bottlenecks}
\label{sec:bottlenecks}

Linear confinement [\cref{eq:expansion}] is a strong feature unique to expanding geometries. One important physical implication is that it gives rise to the possibility of stable Gibbs state components at \emph{finite energy density}.  We will now use  linear confinement  to argue for the existence of large free-energy barriers, as defined in \cref{eq:bottleneck_condition}, around low energy-density states. 

In particular, we prove a \emph{probabilistic} statement. Let us choose $\ket{\vec x_0, \vec z_0}$ randomly according to its Gibbs weight. 
At sufficiently low temperature, we can define a subspace $\mathcal V(\vec x_0, \vec z_0)$ around this state that, with probability at least $1-\order{n\cdot e^{-\delta(n)}}$, satisfies the bottleneck condition in \cref{eq:bottleneck_condition} with $\Delta = e^{-\Omega(\delta(n))}$.

The argument is simplest in the case where $\delta(n)$ scales optimally with $n$, i.e. $\delta(n) = \delta\cdot n$ for some constant $\delta$, as in good codes. Consider first a ground state of such a Hamiltonian, and errors (fluctuations) acting on this state denoted by $\vec{x}, \vec{z}$. As long as $||\vec x||, ||\vec z|| < \delta n$, the energy grows linearly with the error-size due to \cref{eq:expansion}, so that the
ground state is surrounded by an \emph{extensive} $\order{n}$ energy barrier. This, in turn, implies an extensive free-energy barrier at low enough temperatures. 
Indeed, we can define the subspaces $\mathcal{V}$ and $\partial\mathcal{V}$ in such a way that \emph{all} states in~$\partial\mathcal{V}$ have energies $\order{n}$ above the ground state. Since the entropic contribution to the free energy $F(\partial\mathcal{V})$ is also at most extensive, at sufficiently low $T < T_{\rm mem}$ this implies a free energy barrier $\Delta = e^{-\Theta(n)}$. Thus, there is a stable Gibbs state component growing out of each ground state, which implies that the ground states encode a passive quantum memory. This is closely related to the fact that these models are known to have local decoders; in fact, these will correct exactly the errors that span $\mathcal{V}$~\cite{leverrier2015quantum_expander,dinur2023qldpc}. The fact that the ground states can form a passive memory below $T_{\rm mem}$, even as the partition function is analytic at all temperatures, is well-known known for classical expander codes~\cite{montanari2006bethe,mezard2009information}, and was recently discussed for qLDPC codes in~\cite{Hong2024quantum_memory}. If the code also has finite rate, $k \propto n$, then this already implies the existence of exponentially many distinct Gibbs state components at low temperatures. However, there might be many more, associated with finite energy density subspaces. 

Let $\ket{\vec x_0,\vec z_0}$ be a reference eigenstate with energy density $E(\vec x_0,\vec z_0) / n = \varepsilon >0$. We wish to study the energy difference of configurations  ``near" this reference, say obtained by acting with $X^{\vec x}$ with $\norm{\vec x} \leq \delta n$. Defining $\Delta E \equiv E(\vec x_0+ \vec x,\vec z_0) - E(\vec x_0,\vec z_0) = \abs{H_Z (\vec x_0 + \vec x)} - \abs{H_Z \vec x_0}$ and using \cref{eq:expansion} and the triangle inequality on $\abs{H(\vec x_0 + \vec x) + H\vec x_0}$, we write
\begin{align}
    \Delta E \geq \gamma \norm*{\vec x} - 2 \abs{H_Z \vec x_0} \geq \gamma \norm*{\vec x} - 2 \varepsilon n,
    \label{eq:barriers}
\end{align}
and similarly for $\vec z$ errors. Assuming that $\varepsilon < \delta \gamma / 2$, this means that $\ket{\vec x_0,\vec z_0}$ is  ``surrounded'' by an extensive energy barrier separating it from other states whose energy is also $< \delta \gamma n /2$, but which are separated from $\ket{\vec x_0,\vec z_0}$ by distance (reduced error weight) $> \delta n$.    
We can then define a subspace
\begin{equation}
 \mathcal V(\vec x_0, \vec z_0) = \linspan \left\{ X^{\vec x} Z^{\vec z} \ket{\vec x_0, \vec z_0}; \norm{\vec x}, \norm{\vec z}  \leq 2\epsilon n/\gamma \right\}
 \label{eq:subspace}
\end{equation}
with $ \delta \gamma/2 > \epsilon > \varepsilon$ and define the boundary $\partial\mathcal{V}(\vec x_0, \vec z_0)$ similarly to be states outside of $ \mathcal V(\vec x_0, \vec z_0)$ but within distance of $\delta\cdot n$ of $\ket{\vec x_0,\vec z_0}$ (in terms of $||\vec x||, ||\vec z||$). We again find that \cref{eq:bottleneck_condition} is satisfied with $\Delta(\mathcal{V}(\vec x_0, \vec z_0)) = e^{-\Theta(n)}$. 
Thus, there is a stable Gibbs state (passive memory) also associated with $\ket{\vec x_0,\vec z_0}$, 
 with the subspace $\mathcal{V}(\vec x_0, \vec z_0)$ spanned by correctable errors and only containing states at finite energy density. Because of the reduction in barrier size by $2\epsilon n$ in \cref{eq:barriers}, these Gibbs state components are expected to less stable than the ones associated with ground states.  

Let us now consider the more general case where $\delta(n)$ is a sub-linear function of $n$, as in the HGP code with $\delta(n) =\delta \sqrt{n}$. In this case, the argument from linear confinement to free energy barriers is more involved. The key idea is that \cref{eq:expansion} also extends to states with reduced weights \emph{greater than} $\delta(n)$, as long as the largest \emph{connected}\footnote{We say that two qubits in ${\vec x}$ are connected if there exists a check in $\hamil$ that acts on both.} \emph{cluster} of errors in $\vec{x}$ has a size below $\delta(n)$, since the energies are additive between the different disconnected clusters.  Typical thermal fluctuations at low temperatures only create ``non-percolating errors'' with small disconnected clusters; the \emph{largest} cluster typically has size $O(\log{n})$. We thus define define the analog of $ \mathcal V(\vec x_0, \vec z_0)$ in such a way that it includes these typical fluctuations but allows no large connected clusters of size $> \delta(n)/2$ in $\vec{x},\vec{z}$.
We can count the energy of each large connected component in $\partial \mathcal{V}$ separately, which ensures that the energetic and entropic contributions are proportional, and the former dominate at low temperatures. In order to make this heuristic argument precise, we also need to deal with the case when many small connected components are close to each other and form a large ``almost connected'' cluster, which we do~\cite{SM_quantum} by building on ideas developed in \cite{bombin2015single_shot,Fawzi2017efficient,Fawzi2018constant,quintavalle2021single_shot}.

In summary, as long as $\delta(n)/\log n \to \infty$ as $n\to\infty$, \emph{typical} low-temperature eigenstates are surrounded by free energy barriers. We now count the number of Gibbs state components to argue that two randomly drawn low energy states will, with high probability, lie in \emph{distinct} Gibbs state components separated by bottlenecks,  
so the landscape takes the form sketched in \cref{fig:expansion}(b). In particular, we will show that most Gibbs state components are associated with subspaces that do \emph{not} contain a ground state and that there are exponentially many (in $n$) Gibbs states even \emph{within} each symmetry (logical) sector.  

\subsection{Shattering and Incongruence\label{sec:shattering}}
\vspace{-0.3cm}

To establish the existence of spin glass order as defined in \autoref{sec:glass}, we now argue that no single extremal Gibbs state $\rhoGV$ dominates the thermal state $\rhoG$, and that the number of relevant components grows with temperature. 

To this end, we lower bound the configuration entropy $S_{\rm conf}$ using the weight of the most probable subspace $\mathcal{V}_i$, $\Sconf \geq \log(1/w_{\rm max})$ where $w_{\rm max} = \max_i(\text{Tr}(P_{\mathcal{V}_i}\rhoG))$ \cite{SM_quantum}. 

We thus have to upper bound $w_{\rm max}$. To this end, we introduce an energy shell, $\Xi$, of width $\sqrt{\xi n}$ around the mean energy at temperature $\beta$ (more precisely, $\Xi$ is the subspace spanned by all eigenstates within this window). For an appropriate choice of the constant $\xi$, almost all the weight of $\rhoG$ is contained in this shell, and hence, using Hoeffding's inequality~\cite{SM_quantum}
\begin{equation}
    w_{\rm max} \leq Z^{-1} \left( \dim(\mathcal V \cap \Xi) e^{-\beta \expval{E} + \beta \sqrt{\xi n}} + e^{-\Theta(n)} \right),
\end{equation}
where $Z = \tr(e^{-\beta\hamil})$ is the partition function.

We now use the absence of redundancies. In this case, $Z = 2^{rn} \left(1+e^{-\beta}\right)^{(1-r)n}$, and we can use a simple counting argument to upper bound the dimension of $\mathcal V \cap \Xi$. This counting argument follows from Eqs.~\eqref{eq:barriers}, \eqref{eq:subspace}: if $\norm{\vec x}+\norm{\vec z}$ is large, then $\ket{\vec{\vec x_0 + \vec x, \vec z_0 + \vec z}} \notin \Xi$, so the dimension of $\mathcal V \cap \Xi$ can be upper bounded by counting all the vectors with small (reduced) Hamming weight. The upshot is a lower bound on the configurational entropy of the form
\begin{equation}
    s_{\rm conf} = \Sconf/n \geq  r + f(T) + \order{1/\sqrt n},
\end{equation}
where $f$ is an increasing function of the temperature $T$ for sufficiently low temperature and large enough $\gamma$. For specific parameters, the lower bound is shown in \autoref{fig:expansion}(c).
Remarkably, even though our result is only a lower bound, it captures the qualitative behavior seen also e.g. in replica calculations \cite{franz2022dynamic} for classical LDPC codes. 
At zero temperature, the configurational entropy density is equal to the code-rate $r$, stemming from the exponentially large ground state degeneracy.  Then $s_{\rm conf}$ initially increases with temperature so that the physics at finite temperature is dominated by local minima at finite energy densities, which are \emph{incongruent} with the global ground states.

Our arguments establish that at sufficiently low temperatures, $\rhoG$ decomposes into exponentially many components, each with exponentially small weight. While this is sufficient to prove the claimed spin glass behavior, we supplement it by numerical experiments that help elaborate the more detailed phase diagram at intermediate temperatures. In particular, we study a family of balanced-product codes analogous to known constructions of good qLDPC codes~\cite{SM_quantum}. In \cref{fig:expansion} (d), we show the energy as a function of temperature in two different numerical experiments. We use stabilizer Monte-Carlo dynamics initialized in a specific state and measure the energy after a fixed, long time to ensure local equilibration. 
In the first experiment, we initialize the system in a ground state and increase the temperature (heating); in this case, the energy shows a sharp, first-order like transition from close-to-zero energy density to the equilibrium value, at a temperature that we denote by $T_{\rm mem}$. 
This temperature however is invisible to the Gibbs sampler in the second experiment, where initializing in a randomly chosen high-temperature state and slowly lower the temperature (annealing). In this case, the energy plateaus, signaling the system falling out of equilibrium, only at a lower temperature $T_{\rm G} < T_{\rm mem}$.

Our proposed physical origin of this behavior is illustrated at the top of \cref{fig:expansion} (d), and posits that shattering happens in two steps as the temperature is decreased. At high temperatures, the Gibbs state is unique and the configuration space explored by equilibrium dynamics is connected.
As the temperature is lowered, the system first enters an intermediate ``memory'' regime below a temperature $T_{\rm mem}$. While there are exponentially many extremal Gibbs states in this regime, a single component carries almost all the weight; since the most stable Gibbs state components are expected to be the ones containing the ground states (which are surrounded by the largest barriers), $T_{\rm mem}$ corresponds to the passive quantum memory threshold for the ground states of $\hamil$. 
The Gibbs state then shatters into exponentially many low-weight components at an even lower temperature $T_{\rm G}$. The fact that annealing fails to reach lower energies below $T_{\rm G}$ is a consequence of the configurational entropy growing with temperature: typical local minima are located just below the equilibrium energy density of $T_{\rm G}$ and hence the Gibbs sampler gets stuck in one of these minima, failing to cool to temperatures below it. 
This behavior mirrors established structural transitions in classical mean-field glass formers such as diluted $p$-spin glasses and random satisfiability problems \cite{franz2001ferromagnet,krzakala2007gibbs,placke2024classical_glass}

\vspace{-0.3cm}
\subsection{Topological Order and Emergent Quantum Memories}
\vspace{-0.3cm}

\begin{figure}
    \centering
    \includegraphics[width=0.7\columnwidth]{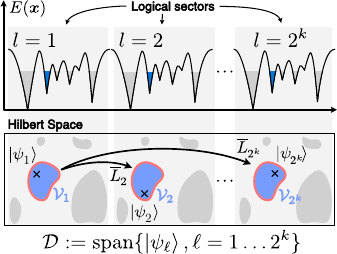}
    \vspace{-0.1cm}
    \caption{{\bf Emergent quantum memories.} We can use arbitrary states from subspaces $\mathcal V_j$ to define thermally stable high-rate and distance codes $\mathcal D$. Since no code state of $\mathcal D$ can be topological trivial, neither can the extremal Gibbs states of $\hamil$, which are supported on the $\mathcal V_j$. 
    }
    \label{fig:complexity}
\end{figure}

Thus far we have used linear confinement to \emph{prove} that, at low enough temperatures, the system displays a complex free energy landscape with shattering and incongruence. We now turn to proving the uniquely quantum (topological) nature of the Gibbs state components.

First, note that the ground states of a finite-rate and macroscopic distance $d$ code are LRE, with a circuit complexity of at least $\log d$~\cite{bravyi2024entanglement}. This is also naturally true of all stabilizer eigenstates $\ket{\vec x_0, \vec z_0}$ obtained by acting on ground states with products of single-site Pauli operators. We now show that typical extremal Gibbs state components $\rhoGV$ at low temperatures are also topologically ordered or long-range entangled, i.e.\ they cannot be approximated by a classical mixture of SRE states. 
In fact, we prove an even stronger statement: at low temperatures, typical $\rhoGV$ are supported on subspaces $\mathcal{V}$ which contain \emph{no} SRE states at all. This may be regarded as a ``local minimum" version of NLTS \cite{freedman2014nlts,anshu2022nlts}. 
 
We prove this by showing that, in the models we study, it is not just the ground states that form a quantum error correcting code; instead we can create a finite-rate and macroscopic distance quantum code from \emph{any} state drawn from a subspace $\mathcal{V}$ (associated with $\rhoGV$) and its symmetry related partners in different logical sectors (see \cref{fig:complexity}). We then leverage the results of Ref. \onlinecite{bravyi2024entanglement} which show that such codestates cannot be SRE.

Explicitly, consider a typical low-energy subspace with a bottleneck, $\mathcal V_1 := \mathcal V(\vec x_0 , \vec z_0)$, constructed as in \autoref{sec:bottlenecks}. Also consider an arbitrary (not necessarily stabilizer) state $\ket{\psi_1} = \sum_{\alpha} c_\alpha \hat E_\alpha \ket{1} \in \mathcal V_1$, where $\ket{1} = \ket{\vec x_0,\vec z_0}$ is the reference state used to construct $\mathcal{V}_1$ and the $\hat E_\alpha$ are possible patterns of fluctuations (errors) on top of this reference state.
We can obtain $2^k - 1$ other subspaces $\mathcal{V}_\ell$ by applying logical operators $\hat L_{\ell}$ to elements of $\mathcal{V}_1$. By construction, these subspaces are spanned by states that have large (relative) reduced distance. One can then show that suitably chosen states from each subspace, $\ket{\psi_\ell} = \sum_{\alpha} c_\alpha \hat E_\alpha \ket{\ell} \in \mathcal{V}_\ell$, where $\ket{\ell} = \hat L_\ell \ket{1}$,  are all  indistinguishable on any subset of size less than $\sim \delta(n)$~\cite{SM_quantum}. Consequently, the code $\mathcal D = \linspan\{\ket{\psi_\ell}, \ell = 1\dots 2^k\}$ is a constant rate code with distance $d\sim\delta(n)$. The code states of any quantum error correcting code with constant rate ($k\sim n$) and distance $d\geq 2$ has circuit complexity of at least $\log d$ \cite{bravyi2024entanglement}. Since the choice of subspace $\mathcal V_1$ and $\ket{\psi_1}\in \mathcal V_1$ here was arbitrary (other than assuming that $\mathcal{V}_1$ is surrounded by a free energy barrier to ensure its thermal stability), this establishes the non triviality of \emph{any} state in any low-temperature subspace $\mathcal V_j$ which is surrounded by a bottleneck, which also implies that the associated $\rhoGV$ is topologically ordered. 

Note that we have constructed the codes $\mathcal D$ from \emph{symmetry-related} subspaces. Importantly, we did not have the choice to use an arbitrary subset of subspaces, even though all subspaces $\mathcal{V}_i$ have large relative distance.
This is because quantum error correction  also requires that the states in the codespace are \emph{locally indistinguishable} while with different error syndromes (i.e. different patterns of violated checks) are clearly locally distinguishable.\footnote{In classical glasses, by contrast, all subspaces together encode an emergent classical code with rate larger than the ground state codespace, $\tilde{k}> k$ (\autoref{fig:toc}).}

Finally, we note that each set of symmetry related subspaces defines an emergent quantum memory, each encoding $k$ qubits. 
Typical sets are associated with subspaces that do not contain codewords of the original code at all, but only states at finite energy density.
TQSGs hence implement a hybrid mixture of classical and quantum memory, and any equilbrium state, as detailed in \autoref{sec:glass},  takes the form $\hat\rho = \sum_i p_i \mu_i \otimes \sigma_i$. The sum is over all emergent quantum codes $i$, the states $\mu_i$ represents the quantum information encoded in the $k$ logical qubits of code $i$, and the collection $\{p_i\}$ represents the classical information encoded using the fact that there are many emergent quantum codes.\footnote{To connect to an operational definition of memory such as the one in Ref. \onlinecite{schumacher1996quantum}, one could imagine encoding a quantum state drawn from a particular code, $\ket{\phi_i} \in \linspan\{\hat L_\ell \ket{\vec x_i, \vec z_i}, \ell = 1\dots 2^k\}$; this encoding can be formulated by entangling the system with some external reference (and the choice of code $i$ represents classical information). We can then run detailed-balance obeying local dynamics to thermalize the system to the appropriate Gibbs states and compute the \emph{coherent information} between the reference and the resulting state, which quantifies whether all the encoded quantum information is in principle recoverable~\cite{brown2016quantum}.}

\vspace{-0.4cm}
\section{Summary and Outlook}
\vspace{-0.2cm}

We have established the existence of a new phase of matter, a topological quantum spin glass. To define and characterize this phase, we developed a generalization of the extremal Gibbs state decomposition to quantum systems, which may be of independent interest. Like classical spin glasses, TQSGs display a complex free-energy landscape with many stable Gibbs state components associated with both global and local energy minima. 
In contrast to conventional spin glasses, however, in TQSGs the individual Gibbs state components are long-range entangled and robustly store quantum information. We explicitly showed that TQSGs are realized as the low-temperature phase of qLDPC codes with linear confinement, including hypergraph and  balanced products of classical expander codes. 

Our work raises many interesting questions for further characterizing these phases.  
For example, nontrivial phases of matter are often characterized in terms of their long-range correlations. These correlations take an unusual form in both topologically ordered phases (e.g., Wilson loops \cite{zeng2019quantum}) as well as in classical spin glasses (such as point-to-set correlators or Edwards-Anderson type order parameters \cite{mezard2009information}). It would be interesting to see whether some combination of these ideas can also characterize TQSGs. Relatedly, classical extremal Gibbs states are characterized by a clustering of correlations, which is in tension with the long-range entanglement of their quantum counterparts uncovered here. Another question is the stability of TQSGs to perturbations of the Hamiltonian; while the stability of topological order in ground states of qLDPC codes in non-Euclidean geometries was recently proven~\cite{de_roeck2024ldpc_stability, yin2024ldpc_stability}, the more general result requires proving stability for the Gibbs states associated with \emph{local} minima.
While this can indeed be proven for the case of extensive energy barriers \cite{rakovszky2024bottleneck,garmanik2024slow}, it remains to be shown in the more general case, encompassing e.g. hypergraph product codes, along with the robustness of long-range entanglement. Finally, our results on the circuit complexity of states in local minima are remarkably close to that of the NLTS theorem for good quantum codes, and hence raise the question of whether  the NLTS property applies to a much broader class of codes with linear confinement but subextensive energy barriers. 

More broadly, our results could have wide-ranging applications via the introduction of methods from spin glass theory to \emph{quantum} computer science. 
For example, our results on the free-energy landscape of qLDPC codes may yield novel methods for the solution of uniquely quantum satisfiability problems, mirroring the success of conventional spin glass methods in understanding classical satisfiability problems~\cite{mezard2002survey_propagation}.
Second, our results for classifying the complexity of equilibrium states connect to central questions in quantum complexity theory, including the NLTS theorem \cite{anshu2022cnlts,anshu2022nlts} and the quantum PCP conjecture~\cite{aharonov2021qpcp}.
Finally, the fact that TQSGs are able to store quantum information in their local minima suggest exciting possible applications in quantum machine learning \cite{hopfield1982network, arya2015associative}.

\vspace{0.0em}
\centerline{\bf Acknowledgements.} 
\vspace{0.2em}

We are particularly grateful to Grace Sommers for collaboration on related work, to Shivaji Sondhi and David Huse for many useful discussions, and to Shivaji Sondhi for sparking this collaboration by bringing the four of us together.
We also thank Andrea Montanari for pointing out useful references, and we thank Anurag Anshu, Claudio Chamon, Daniel Fisher, Silvio Franz, Louis Golowich, Sarang Gopalakrishnan, Tarun Grover, Curt von Keyserlingk, Steve Kivelson, Chris Laumann, Ethan Lake, Andrew Lucas, Roderich Moessner, Akshat Pandey, and Gilles Tarjus for useful discussions.
B.P. acknowledges funding through a Leverhulme-Peierls Fellowship at the University of Oxford and the Alexander von Humboldt foundation through a Feodor-Lynen fellowship.
T.R. was supported in part by Stanford Q-FARM Bloch Postdoctoral Fellowship in Quantum Science and Engineering and by the HUN-REN Welcome Home and Foreign Researcher Recruitment Programme 2023.
This work was done in part while N.P.B. was visiting the Simons Institute for the Theory of Computing, supported by DOE QSA grant \#FP00010905.
V.K. acknowledges support from the Packard Foundation through a Packard Fellowship in Science and Engineering and the Office of Naval Research Young Investigator Program (ONR YIP) under Award Number N00014-24-1-2098. 

\bibliography{references.bib}

\begin{thebibliography}{110}%
\makeatletter
\providecommand \@ifxundefined [1]{%
 \@ifx{#1\undefined}
}%
\providecommand \@ifnum [1]{%
 \ifnum #1\expandafter \@firstoftwo
 \else \expandafter \@secondoftwo
 \fi
}%
\providecommand \@ifx [1]{%
 \ifx #1\expandafter \@firstoftwo
 \else \expandafter \@secondoftwo
 \fi
}%
\providecommand \natexlab [1]{#1}%
\providecommand \enquote  [1]{``#1''}%
\providecommand \bibnamefont  [1]{#1}%
\providecommand \bibfnamefont [1]{#1}%
\providecommand \citenamefont [1]{#1}%
\providecommand \href@noop [0]{\@secondoftwo}%
\providecommand \href [0]{\begingroup \@sanitize@url \@href}%
\providecommand \@href[1]{\@@startlink{#1}\@@href}%
\providecommand \@@href[1]{\endgroup#1\@@endlink}%
\providecommand \@sanitize@url [0]{\catcode `\\12\catcode `\$12\catcode `\&12\catcode `\#12\catcode `\^12\catcode `\_12\catcode `\%12\relax}%
\providecommand \@@startlink[1]{}%
\providecommand \@@endlink[0]{}%
\providecommand \url  [0]{\begingroup\@sanitize@url \@url }%
\providecommand \@url [1]{\endgroup\@href {#1}{\urlprefix }}%
\providecommand \urlprefix  [0]{URL }%
\providecommand \Eprint [0]{\href }%
\providecommand \doibase [0]{http://dx.doi.org/}%
\providecommand \selectlanguage [0]{\@gobble}%
\providecommand \bibinfo  [0]{\@secondoftwo}%
\providecommand \bibfield  [0]{\@secondoftwo}%
\providecommand \translation [1]{[#1]}%
\providecommand \BibitemOpen [0]{}%
\providecommand \bibitemStop [0]{}%
\providecommand \bibitemNoStop [0]{.\EOS\space}%
\providecommand \EOS [0]{\spacefactor3000\relax}%
\providecommand \BibitemShut  [1]{\csname bibitem#1\endcsname}%
\let\auto@bib@innerbib\@empty
\bibitem [{\citenamefont {Anderson}(1970)}]{anderson1970localisation}%
  \BibitemOpen
  \bibfield  {author} {\bibinfo {author} {\bibfnamefont {Philip~Warren}\ \bibnamefont {Anderson}},\ }\bibfield  {title} {\enquote {\bibinfo {title} {Localisation theory and the {Cu Mn} problem: Spin glasses},}\ }\href@noop {} {\bibfield  {journal} {\bibinfo  {journal} {Materials Research Bulletin}\ }\textbf {\bibinfo {volume} {5}},\ \bibinfo {pages} {549--554} (\bibinfo {year} {1970})}\BibitemShut {NoStop}%
\bibitem [{\citenamefont {Edwards}\ and\ \citenamefont {Anderson}(1975)}]{edwards1975theory}%
  \BibitemOpen
  \bibfield  {author} {\bibinfo {author} {\bibfnamefont {Samuel~Frederick}\ \bibnamefont {Edwards}}\ and\ \bibinfo {author} {\bibfnamefont {Phil~W}\ \bibnamefont {Anderson}},\ }\bibfield  {title} {\enquote {\bibinfo {title} {Theory of spin glasses},}\ }\href@noop {} {\bibfield  {journal} {\bibinfo  {journal} {Journal of Physics F: Metal Physics}\ }\textbf {\bibinfo {volume} {5}},\ \bibinfo {pages} {965} (\bibinfo {year} {1975})}\BibitemShut {NoStop}%
\bibitem [{\citenamefont {Sherrington}\ and\ \citenamefont {Kirkpatrick}(1975)}]{sherrington1975solvable}%
  \BibitemOpen
  \bibfield  {author} {\bibinfo {author} {\bibfnamefont {David}\ \bibnamefont {Sherrington}}\ and\ \bibinfo {author} {\bibfnamefont {Scott}\ \bibnamefont {Kirkpatrick}},\ }\bibfield  {title} {\enquote {\bibinfo {title} {Solvable model of a spin-glass},}\ }\href@noop {} {\bibfield  {journal} {\bibinfo  {journal} {Physical review letters}\ }\textbf {\bibinfo {volume} {35}},\ \bibinfo {pages} {1792} (\bibinfo {year} {1975})}\BibitemShut {NoStop}%
\bibitem [{\citenamefont {Parisi}(1979)}]{parisi1979solution2}%
  \BibitemOpen
  \bibfield  {author} {\bibinfo {author} {\bibfnamefont {G.}~\bibnamefont {Parisi}},\ }\bibfield  {title} {\enquote {\bibinfo {title} {Infinite number of order parameters for spin-glasses},}\ }\href {\doibase 10.1103/PhysRevLett.43.1754} {\bibfield  {journal} {\bibinfo  {journal} {Phys. Rev. Lett.}\ }\textbf {\bibinfo {volume} {43}},\ \bibinfo {pages} {1754--1756} (\bibinfo {year} {1979})}\BibitemShut {NoStop}%
\bibitem [{\citenamefont {Hopfield}(1982)}]{hopfield1982network}%
  \BibitemOpen
  \bibfield  {author} {\bibinfo {author} {\bibfnamefont {J~J}\ \bibnamefont {Hopfield}},\ }\bibfield  {title} {\enquote {\bibinfo {title} {Neural networks and physical systems with emergent collective computational abilities.}}\ }\href {\doibase 10.1073/pnas.79.8.2554} {\bibfield  {journal} {\bibinfo  {journal} {Proceedings of the National Academy of Sciences}\ }\textbf {\bibinfo {volume} {79}},\ \bibinfo {pages} {2554--2558} (\bibinfo {year} {1982})}\BibitemShut {NoStop}%
\bibitem [{\citenamefont {Binder}\ and\ \citenamefont {Young}(1986)}]{binder1986review}%
  \BibitemOpen
  \bibfield  {author} {\bibinfo {author} {\bibfnamefont {K.}~\bibnamefont {Binder}}\ and\ \bibinfo {author} {\bibfnamefont {A.~P.}\ \bibnamefont {Young}},\ }\bibfield  {title} {\enquote {\bibinfo {title} {Spin glasses: Experimental facts, theoretical concepts, and open questions},}\ }\href {\doibase 10.1103/RevModPhys.58.801} {\bibfield  {journal} {\bibinfo  {journal} {Rev. Mod. Phys.}\ }\textbf {\bibinfo {volume} {58}},\ \bibinfo {pages} {801--976} (\bibinfo {year} {1986})}\BibitemShut {NoStop}%
\bibitem [{\citenamefont {Mezard}\ and\ \citenamefont {Montanari}(2009)}]{mezard2009information}%
  \BibitemOpen
  \bibfield  {author} {\bibinfo {author} {\bibfnamefont {Marc}\ \bibnamefont {Mezard}}\ and\ \bibinfo {author} {\bibfnamefont {Andrea}\ \bibnamefont {Montanari}},\ }\href@noop {} {\emph {\bibinfo {title} {Information, physics, and computation}}}\ (\bibinfo  {publisher} {Oxford University Press},\ \bibinfo {year} {2009})\BibitemShut {NoStop}%
\bibitem [{\citenamefont {Stein}\ and\ \citenamefont {Newman}(2013)}]{stein2013spin}%
  \BibitemOpen
  \bibfield  {author} {\bibinfo {author} {\bibfnamefont {Daniel~L}\ \bibnamefont {Stein}}\ and\ \bibinfo {author} {\bibfnamefont {Charles~M}\ \bibnamefont {Newman}},\ }\href@noop {} {\emph {\bibinfo {title} {Spin glasses and complexity}}},\ Vol.~\bibinfo {volume} {4}\ (\bibinfo  {publisher} {Princeton University Press},\ \bibinfo {year} {2013})\BibitemShut {NoStop}%
\bibitem [{\citenamefont {Wegner}(1971)}]{wegner1972}%
  \BibitemOpen
  \bibfield  {author} {\bibinfo {author} {\bibfnamefont {Franz~J.}\ \bibnamefont {Wegner}},\ }\bibfield  {title} {\enquote {\bibinfo {title} {Duality in generalized {Ising} models and phase transitions without local order parameters},}\ }\href {\doibase 10.1063/1.1665530} {\bibfield  {journal} {\bibinfo  {journal} {Journal of Mathematical Physics}\ }\textbf {\bibinfo {volume} {12}},\ \bibinfo {pages} {2259--2272} (\bibinfo {year} {1971})}\BibitemShut {NoStop}%
\bibitem [{\citenamefont {Kosterlitz}\ and\ \citenamefont {Thouless}(1973)}]{kosterlitz1973ordering}%
  \BibitemOpen
  \bibfield  {author} {\bibinfo {author} {\bibfnamefont {J~M}\ \bibnamefont {Kosterlitz}}\ and\ \bibinfo {author} {\bibfnamefont {D~J}\ \bibnamefont {Thouless}},\ }\bibfield  {title} {\enquote {\bibinfo {title} {Ordering, metastability and phase transitions in two-dimensional systems},}\ }\href {\doibase 10.1088/0022-3719/6/7/010} {\bibfield  {journal} {\bibinfo  {journal} {Journal of Physics C: Solid State Physics}\ }\textbf {\bibinfo {volume} {6}},\ \bibinfo {pages} {1181} (\bibinfo {year} {1973})}\BibitemShut {NoStop}%
\bibitem [{\citenamefont {Klitzing}\ \emph {et~al.}(1980)\citenamefont {Klitzing}, \citenamefont {Dorda},\ and\ \citenamefont {Pepper}}]{klitzing1980qhe}%
  \BibitemOpen
  \bibfield  {author} {\bibinfo {author} {\bibfnamefont {K.~v.}\ \bibnamefont {Klitzing}}, \bibinfo {author} {\bibfnamefont {G.}~\bibnamefont {Dorda}}, \ and\ \bibinfo {author} {\bibfnamefont {M.}~\bibnamefont {Pepper}},\ }\bibfield  {title} {\enquote {\bibinfo {title} {New method for high-accuracy determination of the fine-structure constant based on quantized {Hall} resistance},}\ }\href {\doibase 10.1103/PhysRevLett.45.494} {\bibfield  {journal} {\bibinfo  {journal} {Phys. Rev. Lett.}\ }\textbf {\bibinfo {volume} {45}},\ \bibinfo {pages} {494--497} (\bibinfo {year} {1980})}\BibitemShut {NoStop}%
\bibitem [{\citenamefont {Tsui}\ \emph {et~al.}(1982)\citenamefont {Tsui}, \citenamefont {Stormer},\ and\ \citenamefont {Gossard}}]{tsui1982fqhe}%
  \BibitemOpen
  \bibfield  {author} {\bibinfo {author} {\bibfnamefont {D.~C.}\ \bibnamefont {Tsui}}, \bibinfo {author} {\bibfnamefont {H.~L.}\ \bibnamefont {Stormer}}, \ and\ \bibinfo {author} {\bibfnamefont {A.~C.}\ \bibnamefont {Gossard}},\ }\bibfield  {title} {\enquote {\bibinfo {title} {Two-dimensional magnetotransport in the extreme quantum limit},}\ }\href {\doibase 10.1103/PhysRevLett.48.1559} {\bibfield  {journal} {\bibinfo  {journal} {Phys. Rev. Lett.}\ }\textbf {\bibinfo {volume} {48}},\ \bibinfo {pages} {1559--1562} (\bibinfo {year} {1982})}\BibitemShut {NoStop}%
\bibitem [{\citenamefont {Laughlin}(1983)}]{laughlin1983wavefunction}%
  \BibitemOpen
  \bibfield  {author} {\bibinfo {author} {\bibfnamefont {R.~B.}\ \bibnamefont {Laughlin}},\ }\bibfield  {title} {\enquote {\bibinfo {title} {Anomalous quantum hall effect: An incompressible quantum fluid with fractionally charged excitations},}\ }\href {\doibase 10.1103/PhysRevLett.50.1395} {\bibfield  {journal} {\bibinfo  {journal} {Phys. Rev. Lett.}\ }\textbf {\bibinfo {volume} {50}},\ \bibinfo {pages} {1395--1398} (\bibinfo {year} {1983})}\BibitemShut {NoStop}%
\bibitem [{\citenamefont {Haldane}(1988)}]{haldane1988model}%
  \BibitemOpen
  \bibfield  {author} {\bibinfo {author} {\bibfnamefont {F.~D.~M.}\ \bibnamefont {Haldane}},\ }\bibfield  {title} {\enquote {\bibinfo {title} {Model for a quantum hall effect without landau levels: Condensed-matter realization of the ``parity anomaly"},}\ }\href {\doibase 10.1103/PhysRevLett.61.2015} {\bibfield  {journal} {\bibinfo  {journal} {Phys. Rev. Lett.}\ }\textbf {\bibinfo {volume} {61}},\ \bibinfo {pages} {2015--2018} (\bibinfo {year} {1988})}\BibitemShut {NoStop}%
\bibitem [{\citenamefont {Zeng}\ \emph {et~al.}(2019)\citenamefont {Zeng}, \citenamefont {Chen}, \citenamefont {Zhou}, \citenamefont {Wen} \emph {et~al.}}]{zeng2019quantum}%
  \BibitemOpen
  \bibfield  {author} {\bibinfo {author} {\bibfnamefont {Bei}\ \bibnamefont {Zeng}}, \bibinfo {author} {\bibfnamefont {Xie}\ \bibnamefont {Chen}}, \bibinfo {author} {\bibfnamefont {Duan-Lu}\ \bibnamefont {Zhou}}, \bibinfo {author} {\bibfnamefont {Xiao-Gang}\ \bibnamefont {Wen}},  \emph {et~al.},\ }\href@noop {} {\emph {\bibinfo {title} {Quantum information meets quantum matter}}}\ (\bibinfo  {publisher} {Springer},\ \bibinfo {year} {2019})\BibitemShut {NoStop}%
\bibitem [{\citenamefont {Moessner}\ and\ \citenamefont {Moore}(2021)}]{moessner2021topological}%
  \BibitemOpen
  \bibfield  {author} {\bibinfo {author} {\bibfnamefont {Roderich}\ \bibnamefont {Moessner}}\ and\ \bibinfo {author} {\bibfnamefont {Joel~E}\ \bibnamefont {Moore}},\ }\href@noop {} {\emph {\bibinfo {title} {Topological phases of matter}}}\ (\bibinfo  {publisher} {Cambridge University Press},\ \bibinfo {year} {2021})\BibitemShut {NoStop}%
\bibitem [{\citenamefont {Debenedetti}\ and\ \citenamefont {Stillinger}(2001)}]{debenedetti_stillinger2001}%
  \BibitemOpen
  \bibfield  {author} {\bibinfo {author} {\bibfnamefont {Pablo~G.}\ \bibnamefont {Debenedetti}}\ and\ \bibinfo {author} {\bibfnamefont {Frank~H.}\ \bibnamefont {Stillinger}},\ }\bibfield  {title} {\enquote {\bibinfo {title} {Supercooled liquids and the glass transition},}\ }\href {\doibase 10.1038/35065704} {\bibfield  {journal} {\bibinfo  {journal} {Nature}\ }\textbf {\bibinfo {volume} {410}},\ \bibinfo {pages} {259--267} (\bibinfo {year} {2001})}\BibitemShut {NoStop}%
\bibitem [{\citenamefont {Tarjus}\ \emph {et~al.}(2005)\citenamefont {Tarjus}, \citenamefont {Kivelson}, \citenamefont {Nussinov},\ and\ \citenamefont {Viot}}]{tarjus2005frustration}%
  \BibitemOpen
  \bibfield  {author} {\bibinfo {author} {\bibfnamefont {G}~\bibnamefont {Tarjus}}, \bibinfo {author} {\bibfnamefont {S~A}\ \bibnamefont {Kivelson}}, \bibinfo {author} {\bibfnamefont {Z}~\bibnamefont {Nussinov}}, \ and\ \bibinfo {author} {\bibfnamefont {P}~\bibnamefont {Viot}},\ }\bibfield  {title} {\enquote {\bibinfo {title} {The frustration-based approach of supercooled liquids and the glass transition: a review and critical assessment},}\ }\href {\doibase 10.1088/0953-8984/17/50/R01} {\bibfield  {journal} {\bibinfo  {journal} {Journal of Physics: Condensed Matter}\ }\textbf {\bibinfo {volume} {17}},\ \bibinfo {pages} {R1143} (\bibinfo {year} {2005})}\BibitemShut {NoStop}%
\bibitem [{\citenamefont {Berthier}\ and\ \citenamefont {Biroli}(2011)}]{bertier2011rmp}%
  \BibitemOpen
  \bibfield  {author} {\bibinfo {author} {\bibfnamefont {Ludovic}\ \bibnamefont {Berthier}}\ and\ \bibinfo {author} {\bibfnamefont {Giulio}\ \bibnamefont {Biroli}},\ }\bibfield  {title} {\enquote {\bibinfo {title} {Theoretical perspective on the glass transition and amorphous materials},}\ }\href {\doibase 10.1103/RevModPhys.83.587} {\bibfield  {journal} {\bibinfo  {journal} {Rev. Mod. Phys.}\ }\textbf {\bibinfo {volume} {83}},\ \bibinfo {pages} {587--645} (\bibinfo {year} {2011})}\BibitemShut {NoStop}%
\bibitem [{\citenamefont {Anderson}(1978)}]{anderson1978concept}%
  \BibitemOpen
  \bibfield  {author} {\bibinfo {author} {\bibfnamefont {PW}~\bibnamefont {Anderson}},\ }\bibfield  {title} {\enquote {\bibinfo {title} {The concept of frustration in spin glasses},}\ }\href@noop {} {\bibfield  {journal} {\bibinfo  {journal} {Journal of the Less Common Metals}\ }\textbf {\bibinfo {volume} {62}},\ \bibinfo {pages} {291--294} (\bibinfo {year} {1978})}\BibitemShut {NoStop}%
\bibitem [{\citenamefont {Bramwell}\ and\ \citenamefont {Gingras}(2001)}]{bramwell2001review}%
  \BibitemOpen
  \bibfield  {author} {\bibinfo {author} {\bibfnamefont {Steven~T.}\ \bibnamefont {Bramwell}}\ and\ \bibinfo {author} {\bibfnamefont {Michel J.~P.}\ \bibnamefont {Gingras}},\ }\bibfield  {title} {\enquote {\bibinfo {title} {Spin ice state in frustrated magnetic pyrochlore materials},}\ }\href {\doibase 10.1126/science.1064761} {\bibfield  {journal} {\bibinfo  {journal} {Science}\ }\textbf {\bibinfo {volume} {294}},\ \bibinfo {pages} {1495--1501} (\bibinfo {year} {2001})},\ \Eprint {http://arxiv.org/abs/https://www.science.org/doi/pdf/10.1126/science.1064761} {https://www.science.org/doi/pdf/10.1126/science.1064761} \BibitemShut {NoStop}%
\bibitem [{\citenamefont {Newman}\ and\ \citenamefont {Moore}(1999)}]{newman_moore1999}%
  \BibitemOpen
  \bibfield  {author} {\bibinfo {author} {\bibfnamefont {M.~E.~J.}\ \bibnamefont {Newman}}\ and\ \bibinfo {author} {\bibfnamefont {Cristopher}\ \bibnamefont {Moore}},\ }\bibfield  {title} {\enquote {\bibinfo {title} {Glassy dynamics and aging in an exactly solvable spin model},}\ }\href {\doibase 10.1103/PhysRevE.60.5068} {\bibfield  {journal} {\bibinfo  {journal} {Phys. Rev. E}\ }\textbf {\bibinfo {volume} {60}},\ \bibinfo {pages} {5068--5072} (\bibinfo {year} {1999})}\BibitemShut {NoStop}%
\bibitem [{\citenamefont {Garrahan}\ \emph {et~al.}(2011)\citenamefont {Garrahan}, \citenamefont {Sollich},\ and\ \citenamefont {Toninelli}}]{garrahan2011kinetically}%
  \BibitemOpen
  \bibfield  {author} {\bibinfo {author} {\bibfnamefont {Juan~P}\ \bibnamefont {Garrahan}}, \bibinfo {author} {\bibfnamefont {Peter}\ \bibnamefont {Sollich}}, \ and\ \bibinfo {author} {\bibfnamefont {Cristina}\ \bibnamefont {Toninelli}},\ }\bibfield  {title} {\enquote {\bibinfo {title} {Kinetically constrained models},}\ }\href@noop {} {\bibfield  {journal} {\bibinfo  {journal} {Dynamical heterogeneities in glasses, colloids, and granular media}\ }\textbf {\bibinfo {volume} {150}},\ \bibinfo {pages} {111--137} (\bibinfo {year} {2011})}\BibitemShut {NoStop}%
\bibitem [{\citenamefont {Chamon}(2005)}]{chamon2005quantum_glass}%
  \BibitemOpen
  \bibfield  {author} {\bibinfo {author} {\bibfnamefont {Claudio}\ \bibnamefont {Chamon}},\ }\bibfield  {title} {\enquote {\bibinfo {title} {Quantum glassiness in strongly correlated clean systems: An example of topological overprotection},}\ }\href {\doibase 10.1103/PhysRevLett.94.040402} {\bibfield  {journal} {\bibinfo  {journal} {Phys. Rev. Lett.}\ }\textbf {\bibinfo {volume} {94}},\ \bibinfo {pages} {040402} (\bibinfo {year} {2005})}\BibitemShut {NoStop}%
\bibitem [{\citenamefont {Haah}(2011)}]{haah2011fractons}%
  \BibitemOpen
  \bibfield  {author} {\bibinfo {author} {\bibfnamefont {Jeongwan}\ \bibnamefont {Haah}},\ }\bibfield  {title} {\enquote {\bibinfo {title} {Local stabilizer codes in three dimensions without string logical operators},}\ }\href {\doibase 10.1103/PhysRevA.83.042330} {\bibfield  {journal} {\bibinfo  {journal} {Phys. Rev. A}\ }\textbf {\bibinfo {volume} {83}},\ \bibinfo {pages} {042330} (\bibinfo {year} {2011})}\BibitemShut {NoStop}%
\bibitem [{\citenamefont {Castelnovo}\ and\ \citenamefont {Chamon}(2012)}]{Castelnovo2012topo_glass}%
  \BibitemOpen
  \bibfield  {author} {\bibinfo {author} {\bibfnamefont {Claudio}\ \bibnamefont {Castelnovo}}\ and\ \bibinfo {author} {\bibfnamefont {Claudio}\ \bibnamefont {Chamon}},\ }\bibfield  {title} {\enquote {\bibinfo {title} {Topological quantum glassiness},}\ }\href {\doibase 10.1080/14786435.2011.609152} {\bibfield  {journal} {\bibinfo  {journal} {Philosophical Magazine}\ }\textbf {\bibinfo {volume} {92}},\ \bibinfo {pages} {304--323} (\bibinfo {year} {2012})}\BibitemShut {NoStop}%
\bibitem [{\citenamefont {M{\'e}zard}\ \emph {et~al.}(2003)\citenamefont {M{\'e}zard}, \citenamefont {Ricci-Tersenghi},\ and\ \citenamefont {Zecchina}}]{mezard2003two_solutions}%
  \BibitemOpen
  \bibfield  {author} {\bibinfo {author} {\bibfnamefont {M.}~\bibnamefont {M{\'e}zard}}, \bibinfo {author} {\bibfnamefont {F.}~\bibnamefont {Ricci-Tersenghi}}, \ and\ \bibinfo {author} {\bibfnamefont {R.}~\bibnamefont {Zecchina}},\ }\bibfield  {title} {\enquote {\bibinfo {title} {Two solutions to diluted p-spin models and xorsat problems},}\ }\href {\doibase 10.1023/A:1022886412117} {\bibfield  {journal} {\bibinfo  {journal} {Journal of Statistical Physics}\ }\textbf {\bibinfo {volume} {111}},\ \bibinfo {pages} {505--533} (\bibinfo {year} {2003})}\BibitemShut {NoStop}%
\bibitem [{\citenamefont {Dubois}\ and\ \citenamefont {Mandler}(2002)}]{dubois2002xorsat}%
  \BibitemOpen
  \bibfield  {author} {\bibinfo {author} {\bibfnamefont {Olivier}\ \bibnamefont {Dubois}}\ and\ \bibinfo {author} {\bibfnamefont {Jacques}\ \bibnamefont {Mandler}},\ }\bibfield  {title} {\enquote {\bibinfo {title} {The 3-xorsat threshold},}\ }\href {\doibase https://doi.org/10.1016/S1631-073X(02)02563-3} {\bibfield  {journal} {\bibinfo  {journal} {Comptes Rendus Mathematique}\ }\textbf {\bibinfo {volume} {335}},\ \bibinfo {pages} {963--966} (\bibinfo {year} {2002})}\BibitemShut {NoStop}%
\bibitem [{\citenamefont {Montanari}\ and\ \citenamefont {Semerjian}(2006)}]{montanari2006bethe}%
  \BibitemOpen
  \bibfield  {author} {\bibinfo {author} {\bibfnamefont {Andrea}\ \bibnamefont {Montanari}}\ and\ \bibinfo {author} {\bibfnamefont {Guilhem}\ \bibnamefont {Semerjian}},\ }\bibfield  {title} {\enquote {\bibinfo {title} {On the dynamics of the glass transition on bethe lattices},}\ }\href {\doibase 10.1007/s10955-006-9103-1} {\bibfield  {journal} {\bibinfo  {journal} {Journal of Statistical Physics}\ }\textbf {\bibinfo {volume} {124}},\ \bibinfo {pages} {103--189} (\bibinfo {year} {2006})}\BibitemShut {NoStop}%
\bibitem [{\citenamefont {Fisher}\ and\ \citenamefont {Huse}(1986)}]{fisher_huse1986}%
  \BibitemOpen
  \bibfield  {author} {\bibinfo {author} {\bibfnamefont {Daniel~S.}\ \bibnamefont {Fisher}}\ and\ \bibinfo {author} {\bibfnamefont {David~A.}\ \bibnamefont {Huse}},\ }\bibfield  {title} {\enquote {\bibinfo {title} {Ordered phase of short-range ising spin-glasses},}\ }\href {\doibase 10.1103/PhysRevLett.56.1601} {\bibfield  {journal} {\bibinfo  {journal} {Phys. Rev. Lett.}\ }\textbf {\bibinfo {volume} {56}},\ \bibinfo {pages} {1601--1604} (\bibinfo {year} {1986})}\BibitemShut {NoStop}%
\bibitem [{\citenamefont {Newman}\ and\ \citenamefont {Stein}(2024)}]{Newman_2024}%
  \BibitemOpen
  \bibfield  {author} {\bibinfo {author} {\bibfnamefont {C.~M.}\ \bibnamefont {Newman}}\ and\ \bibinfo {author} {\bibfnamefont {D.~L.}\ \bibnamefont {Stein}},\ }\bibfield  {title} {\enquote {\bibinfo {title} {Free energy difference fluctuations in short-range spin glasses},}\ }\href {\doibase 10.1007/s10955-024-03334-4} {\bibfield  {journal} {\bibinfo  {journal} {Journal of Statistical Physics}\ }\textbf {\bibinfo {volume} {191}} (\bibinfo {year} {2024}),\ 10.1007/s10955-024-03334-4}\BibitemShut {NoStop}%
\bibitem [{\citenamefont {Amit}\ \emph {et~al.}(1985)\citenamefont {Amit}, \citenamefont {Gutfreund},\ and\ \citenamefont {Sompolinsky}}]{amit1985sg_nn}%
  \BibitemOpen
  \bibfield  {author} {\bibinfo {author} {\bibfnamefont {Daniel~J.}\ \bibnamefont {Amit}}, \bibinfo {author} {\bibfnamefont {Hanoch}\ \bibnamefont {Gutfreund}}, \ and\ \bibinfo {author} {\bibfnamefont {H.}~\bibnamefont {Sompolinsky}},\ }\bibfield  {title} {\enquote {\bibinfo {title} {Spin-glass models of neural networks},}\ }\href {\doibase 10.1103/PhysRevA.32.1007} {\bibfield  {journal} {\bibinfo  {journal} {Phys. Rev. A}\ }\textbf {\bibinfo {volume} {32}},\ \bibinfo {pages} {1007--1018} (\bibinfo {year} {1985})}\BibitemShut {NoStop}%
\bibitem [{\citenamefont {Krzakala}\ and\ \citenamefont {Zdeborová}(2024)}]{Krzakala2024machine_learning}%
  \BibitemOpen
  \bibfield  {author} {\bibinfo {author} {\bibfnamefont {Florent}\ \bibnamefont {Krzakala}}\ and\ \bibinfo {author} {\bibfnamefont {Lenka}\ \bibnamefont {Zdeborová}},\ }\bibfield  {title} {\enquote {\bibinfo {title} {Les houches 2022 special issue: Editorial},}\ }\href {\doibase 10.1088/1742-5468/ad4e2a} {\bibfield  {journal} {\bibinfo  {journal} {Journal of Statistical Mechanics: Theory and Experiment}\ }\textbf {\bibinfo {volume} {2024}},\ \bibinfo {pages} {101001} (\bibinfo {year} {2024})}\BibitemShut {NoStop}%
\bibitem [{\citenamefont {Monasson}\ \emph {et~al.}(1999)\citenamefont {Monasson}, \citenamefont {Zecchina}, \citenamefont {Kirkpatrick}, \citenamefont {Selman},\ and\ \citenamefont {Troyansky}}]{monasson1999complexity_transitions}%
  \BibitemOpen
  \bibfield  {author} {\bibinfo {author} {\bibfnamefont {R{\'e}mi}\ \bibnamefont {Monasson}}, \bibinfo {author} {\bibfnamefont {Riccardo}\ \bibnamefont {Zecchina}}, \bibinfo {author} {\bibfnamefont {Scott}\ \bibnamefont {Kirkpatrick}}, \bibinfo {author} {\bibfnamefont {Bart}\ \bibnamefont {Selman}}, \ and\ \bibinfo {author} {\bibfnamefont {Lidror}\ \bibnamefont {Troyansky}},\ }\bibfield  {title} {\enquote {\bibinfo {title} {Determining computational complexity from characteristic `phase transitions'},}\ }\href {\doibase 10.1038/22055} {\bibfield  {journal} {\bibinfo  {journal} {Nature}\ }\textbf {\bibinfo {volume} {400}},\ \bibinfo {pages} {133--137} (\bibinfo {year} {1999})}\BibitemShut {NoStop}%
\bibitem [{\citenamefont {Achlioptas}\ and\ \citenamefont {Coja-Oghlan}(2008)}]{achlioptas2008algorithmic}%
  \BibitemOpen
  \bibfield  {author} {\bibinfo {author} {\bibfnamefont {Dimitris}\ \bibnamefont {Achlioptas}}\ and\ \bibinfo {author} {\bibfnamefont {Amin}\ \bibnamefont {Coja-Oghlan}},\ }\bibfield  {title} {\enquote {\bibinfo {title} {Algorithmic barriers from phase transitions},}\ }in\ \href {\doibase 10.1109/FOCS.2008.11} {\emph {\bibinfo {booktitle} {2008 49th Annual IEEE Symposium on Foundations of Computer Science}}}\ (\bibinfo {year} {2008})\ pp.\ \bibinfo {pages} {793--802}\BibitemShut {NoStop}%
\bibitem [{\citenamefont {Mézard}\ \emph {et~al.}(2002)\citenamefont {Mézard}, \citenamefont {Parisi},\ and\ \citenamefont {Zecchina}}]{mezard2002survey_propagation}%
  \BibitemOpen
  \bibfield  {author} {\bibinfo {author} {\bibfnamefont {M.}~\bibnamefont {Mézard}}, \bibinfo {author} {\bibfnamefont {G.}~\bibnamefont {Parisi}}, \ and\ \bibinfo {author} {\bibfnamefont {R.}~\bibnamefont {Zecchina}},\ }\bibfield  {title} {\enquote {\bibinfo {title} {Analytic and algorithmic solution of random satisfiability problems},}\ }\href {\doibase 10.1126/science.1073287} {\bibfield  {journal} {\bibinfo  {journal} {Science}\ }\textbf {\bibinfo {volume} {297}},\ \bibinfo {pages} {812--815} (\bibinfo {year} {2002})}\BibitemShut {NoStop}%
\bibitem [{\citenamefont {Krzakała}\ \emph {et~al.}(2007)\citenamefont {Krzakała}, \citenamefont {Montanari}, \citenamefont {Ricci-Tersenghi}, \citenamefont {Semerjian},\ and\ \citenamefont {Zdeborov\'a}}]{krzakala2007gibbs}%
  \BibitemOpen
  \bibfield  {author} {\bibinfo {author} {\bibfnamefont {Florent}\ \bibnamefont {Krzakała}}, \bibinfo {author} {\bibfnamefont {Andrea}\ \bibnamefont {Montanari}}, \bibinfo {author} {\bibfnamefont {Federico}\ \bibnamefont {Ricci-Tersenghi}}, \bibinfo {author} {\bibfnamefont {Guilhem}\ \bibnamefont {Semerjian}}, \ and\ \bibinfo {author} {\bibfnamefont {Lenka}\ \bibnamefont {Zdeborov\'a}},\ }\bibfield  {title} {\enquote {\bibinfo {title} {Gibbs states and the set of solutions of random constraint satisfaction problems},}\ }\href {\doibase 10.1073/pnas.0703685104} {\bibfield  {journal} {\bibinfo  {journal} {Proceedings of the National Academy of Sciences}\ }\textbf {\bibinfo {volume} {104}},\ \bibinfo {pages} {10318--10323} (\bibinfo {year} {2007})}\BibitemShut {NoStop}%
\bibitem [{\citenamefont {Ricci-Tersenghi}(2010)}]{ricci_tersenghi2010xorsat}%
  \BibitemOpen
  \bibfield  {author} {\bibinfo {author} {\bibfnamefont {Federico}\ \bibnamefont {Ricci-Tersenghi}},\ }\bibfield  {title} {\enquote {\bibinfo {title} {Being glassy without being hard to solve},}\ }\href {\doibase 10.1126/science.1189804} {\bibfield  {journal} {\bibinfo  {journal} {Science}\ }\textbf {\bibinfo {volume} {330}},\ \bibinfo {pages} {1639--1640} (\bibinfo {year} {2010})}\BibitemShut {NoStop}%
\bibitem [{\citenamefont {Franz}\ \emph {et~al.}(2002)\citenamefont {Franz}, \citenamefont {Leone}, \citenamefont {Montanari},\ and\ \citenamefont {Ricci-Tersenghi}}]{franz2022dynamic}%
  \BibitemOpen
  \bibfield  {author} {\bibinfo {author} {\bibnamefont {Franz}}, \bibinfo {author} {\bibfnamefont {Michele}\ \bibnamefont {Leone}}, \bibinfo {author} {\bibfnamefont {Andrea}\ \bibnamefont {Montanari}}, \ and\ \bibinfo {author} {\bibfnamefont {Federico}\ \bibnamefont {Ricci-Tersenghi}},\ }\bibfield  {title} {\enquote {\bibinfo {title} {Dynamic phase transition for decoding algorithms},}\ }\href {\doibase 10.1103/PhysRevE.66.046120} {\bibfield  {journal} {\bibinfo  {journal} {Phys. Rev. E}\ }\textbf {\bibinfo {volume} {66}},\ \bibinfo {pages} {046120} (\bibinfo {year} {2002})}\BibitemShut {NoStop}%
\bibitem [{\citenamefont {Di}\ \emph {et~al.}(2004)\citenamefont {Di}, \citenamefont {Montanari},\ and\ \citenamefont {Urbanke}}]{di2004weight_enumerators}%
  \BibitemOpen
  \bibfield  {author} {\bibinfo {author} {\bibfnamefont {C.}~\bibnamefont {Di}}, \bibinfo {author} {\bibfnamefont {A.}~\bibnamefont {Montanari}}, \ and\ \bibinfo {author} {\bibfnamefont {R.}~\bibnamefont {Urbanke}},\ }\bibfield  {title} {\enquote {\bibinfo {title} {Weight distributions of ldpc code ensembles: combinatorics meets statistical physics},}\ }in\ \href {\doibase 10.1109/ISIT.2004.1365139} {\emph {\bibinfo {booktitle} {International Symposium onInformation Theory, 2004. ISIT 2004. Proceedings.}}}\ (\bibinfo {year} {2004})\ pp.\ \bibinfo {pages} {102--}\BibitemShut {NoStop}%
\bibitem [{\citenamefont {Shor}(1995)}]{shor1995scheme}%
  \BibitemOpen
  \bibfield  {author} {\bibinfo {author} {\bibfnamefont {Peter~W}\ \bibnamefont {Shor}},\ }\bibfield  {title} {\enquote {\bibinfo {title} {Scheme for reducing decoherence in quantum computer memory},}\ }\href@noop {} {\bibfield  {journal} {\bibinfo  {journal} {Physical review A}\ }\textbf {\bibinfo {volume} {52}},\ \bibinfo {pages} {R2493} (\bibinfo {year} {1995})}\BibitemShut {NoStop}%
\bibitem [{\citenamefont {Terhal}(2015)}]{terhal2015quantum}%
  \BibitemOpen
  \bibfield  {author} {\bibinfo {author} {\bibfnamefont {Barbara~M}\ \bibnamefont {Terhal}},\ }\bibfield  {title} {\enquote {\bibinfo {title} {Quantum error correction for quantum memories},}\ }\href@noop {} {\bibfield  {journal} {\bibinfo  {journal} {Reviews of Modern Physics}\ }\textbf {\bibinfo {volume} {87}},\ \bibinfo {pages} {307--346} (\bibinfo {year} {2015})}\BibitemShut {NoStop}%
\bibitem [{\citenamefont {Brown}\ \emph {et~al.}(2016)\citenamefont {Brown}, \citenamefont {Loss}, \citenamefont {Pachos}, \citenamefont {Self},\ and\ \citenamefont {Wootton}}]{brown2016quantum}%
  \BibitemOpen
  \bibfield  {author} {\bibinfo {author} {\bibfnamefont {Benjamin~J}\ \bibnamefont {Brown}}, \bibinfo {author} {\bibfnamefont {Daniel}\ \bibnamefont {Loss}}, \bibinfo {author} {\bibfnamefont {Jiannis~K}\ \bibnamefont {Pachos}}, \bibinfo {author} {\bibfnamefont {Chris~N}\ \bibnamefont {Self}}, \ and\ \bibinfo {author} {\bibfnamefont {James~R}\ \bibnamefont {Wootton}},\ }\bibfield  {title} {\enquote {\bibinfo {title} {Quantum memories at finite temperature},}\ }\href@noop {} {\bibfield  {journal} {\bibinfo  {journal} {Reviews of Modern Physics}\ }\textbf {\bibinfo {volume} {88}},\ \bibinfo {pages} {045005} (\bibinfo {year} {2016})}\BibitemShut {NoStop}%
\bibitem [{\citenamefont {Knill}\ and\ \citenamefont {Laflamme}(1997)}]{knill1997theory}%
  \BibitemOpen
  \bibfield  {author} {\bibinfo {author} {\bibfnamefont {Emanuel}\ \bibnamefont {Knill}}\ and\ \bibinfo {author} {\bibfnamefont {Raymond}\ \bibnamefont {Laflamme}},\ }\bibfield  {title} {\enquote {\bibinfo {title} {Theory of quantum error-correcting codes},}\ }\href@noop {} {\bibfield  {journal} {\bibinfo  {journal} {Physical Review A}\ }\textbf {\bibinfo {volume} {55}},\ \bibinfo {pages} {900} (\bibinfo {year} {1997})}\BibitemShut {NoStop}%
\bibitem [{\citenamefont {Bravyi}\ \emph {et~al.}(2006)\citenamefont {Bravyi}, \citenamefont {Hastings},\ and\ \citenamefont {Verstraete}}]{bravyi2006lieb}%
  \BibitemOpen
  \bibfield  {author} {\bibinfo {author} {\bibfnamefont {Sergey}\ \bibnamefont {Bravyi}}, \bibinfo {author} {\bibfnamefont {Matthew~B}\ \bibnamefont {Hastings}}, \ and\ \bibinfo {author} {\bibfnamefont {Frank}\ \bibnamefont {Verstraete}},\ }\bibfield  {title} {\enquote {\bibinfo {title} {Lieb-robinson bounds and the generation of correlations and topological quantum order},}\ }\href@noop {} {\bibfield  {journal} {\bibinfo  {journal} {Physical review letters}\ }\textbf {\bibinfo {volume} {97}},\ \bibinfo {pages} {050401} (\bibinfo {year} {2006})}\BibitemShut {NoStop}%
\bibitem [{\citenamefont {Chen}\ \emph {et~al.}(2010)\citenamefont {Chen}, \citenamefont {Gu},\ and\ \citenamefont {Wen}}]{chen2010local}%
  \BibitemOpen
  \bibfield  {author} {\bibinfo {author} {\bibfnamefont {Xie}\ \bibnamefont {Chen}}, \bibinfo {author} {\bibfnamefont {Zheng-Cheng}\ \bibnamefont {Gu}}, \ and\ \bibinfo {author} {\bibfnamefont {Xiao-Gang}\ \bibnamefont {Wen}},\ }\bibfield  {title} {\enquote {\bibinfo {title} {Local unitary transformation, long-range quantum entanglement, wave function renormalization, and topological order},}\ }\href@noop {} {\bibfield  {journal} {\bibinfo  {journal} {Physical review b}\ }\textbf {\bibinfo {volume} {82}},\ \bibinfo {pages} {155138} (\bibinfo {year} {2010})}\BibitemShut {NoStop}%
\bibitem [{\citenamefont {Bravyi}\ \emph {et~al.}(2010{\natexlab{a}})\citenamefont {Bravyi}, \citenamefont {Hastings},\ and\ \citenamefont {Michalakis}}]{bravyi2010topological}%
  \BibitemOpen
  \bibfield  {author} {\bibinfo {author} {\bibfnamefont {Sergey}\ \bibnamefont {Bravyi}}, \bibinfo {author} {\bibfnamefont {Matthew~B}\ \bibnamefont {Hastings}}, \ and\ \bibinfo {author} {\bibfnamefont {Spyridon}\ \bibnamefont {Michalakis}},\ }\bibfield  {title} {\enquote {\bibinfo {title} {Topological quantum order: stability under local perturbations},}\ }\href@noop {} {\bibfield  {journal} {\bibinfo  {journal} {Journal of mathematical physics}\ }\textbf {\bibinfo {volume} {51}} (\bibinfo {year} {2010}{\natexlab{a}})}\BibitemShut {NoStop}%
\bibitem [{\citenamefont {Dennis}\ \emph {et~al.}(2002)\citenamefont {Dennis}, \citenamefont {Kitaev}, \citenamefont {Landahl},\ and\ \citenamefont {Preskill}}]{dennis2002topological}%
  \BibitemOpen
  \bibfield  {author} {\bibinfo {author} {\bibfnamefont {Eric}\ \bibnamefont {Dennis}}, \bibinfo {author} {\bibfnamefont {Alexei}\ \bibnamefont {Kitaev}}, \bibinfo {author} {\bibfnamefont {Andrew}\ \bibnamefont {Landahl}}, \ and\ \bibinfo {author} {\bibfnamefont {John}\ \bibnamefont {Preskill}},\ }\bibfield  {title} {\enquote {\bibinfo {title} {Topological quantum memory},}\ }\href@noop {} {\bibfield  {journal} {\bibinfo  {journal} {Journal of Mathematical Physics}\ }\textbf {\bibinfo {volume} {43}},\ \bibinfo {pages} {4452--4505} (\bibinfo {year} {2002})}\BibitemShut {NoStop}%
\bibitem [{\citenamefont {Alicki}\ \emph {et~al.}(2010)\citenamefont {Alicki}, \citenamefont {Horodecki}, \citenamefont {Horodecki},\ and\ \citenamefont {Horodecki}}]{alicki2010thermal}%
  \BibitemOpen
  \bibfield  {author} {\bibinfo {author} {\bibfnamefont {Robert}\ \bibnamefont {Alicki}}, \bibinfo {author} {\bibfnamefont {Michal}\ \bibnamefont {Horodecki}}, \bibinfo {author} {\bibfnamefont {Pawel}\ \bibnamefont {Horodecki}}, \ and\ \bibinfo {author} {\bibfnamefont {Ryszard}\ \bibnamefont {Horodecki}},\ }\bibfield  {title} {\enquote {\bibinfo {title} {On thermal stability of topological qubit in kitaev's 4d model},}\ }\href@noop {} {\bibfield  {journal} {\bibinfo  {journal} {Open Systems \& Information Dynamics}\ }\textbf {\bibinfo {volume} {17}},\ \bibinfo {pages} {1--20} (\bibinfo {year} {2010})}\BibitemShut {NoStop}%
\bibitem [{\citenamefont {Hastings}(2011)}]{hastings2011topological}%
  \BibitemOpen
  \bibfield  {author} {\bibinfo {author} {\bibfnamefont {Matthew~B}\ \bibnamefont {Hastings}},\ }\bibfield  {title} {\enquote {\bibinfo {title} {Topological order at nonzero temperature},}\ }\href@noop {} {\bibfield  {journal} {\bibinfo  {journal} {Physical review letters}\ }\textbf {\bibinfo {volume} {107}},\ \bibinfo {pages} {210501} (\bibinfo {year} {2011})}\BibitemShut {NoStop}%
\bibitem [{\citenamefont {Yoshida}(2011)}]{yoshida2011feasibility}%
  \BibitemOpen
  \bibfield  {author} {\bibinfo {author} {\bibfnamefont {Beni}\ \bibnamefont {Yoshida}},\ }\bibfield  {title} {\enquote {\bibinfo {title} {Feasibility of self-correcting quantum memory and thermal stability of topological order},}\ }\href {\doibase 10.1016/j.aop.2011.06.001} {\bibfield  {journal} {\bibinfo  {journal} {Annals of Physics}\ }\textbf {\bibinfo {volume} {326}},\ \bibinfo {pages} {2566--2633} (\bibinfo {year} {2011})}\BibitemShut {NoStop}%
\bibitem [{\citenamefont {Liu}\ and\ \citenamefont {Lieu}(2024)}]{liu2024dissipative}%
  \BibitemOpen
  \bibfield  {author} {\bibinfo {author} {\bibfnamefont {Yu-Jie}\ \bibnamefont {Liu}}\ and\ \bibinfo {author} {\bibfnamefont {Simon}\ \bibnamefont {Lieu}},\ }\bibfield  {title} {\enquote {\bibinfo {title} {Dissipative phase transitions and passive error correction},}\ }\href {\doibase 10.1103/PhysRevA.109.022422} {\bibfield  {journal} {\bibinfo  {journal} {Phys. Rev. A}\ }\textbf {\bibinfo {volume} {109}},\ \bibinfo {pages} {022422} (\bibinfo {year} {2024})}\BibitemShut {NoStop}%
\bibitem [{\citenamefont {Tsomokos}\ \emph {et~al.}(2011)\citenamefont {Tsomokos}, \citenamefont {Osborne},\ and\ \citenamefont {Castelnovo}}]{tsomokos2011interplay}%
  \BibitemOpen
  \bibfield  {author} {\bibinfo {author} {\bibfnamefont {Dimitris~I}\ \bibnamefont {Tsomokos}}, \bibinfo {author} {\bibfnamefont {Tobias~J}\ \bibnamefont {Osborne}}, \ and\ \bibinfo {author} {\bibfnamefont {Claudio}\ \bibnamefont {Castelnovo}},\ }\bibfield  {title} {\enquote {\bibinfo {title} {Interplay of topological order and spin glassiness in the toric code under random magnetic fields},}\ }\href@noop {} {\bibfield  {journal} {\bibinfo  {journal} {Physical Review B—Condensed Matter and Materials Physics}\ }\textbf {\bibinfo {volume} {83}},\ \bibinfo {pages} {075124} (\bibinfo {year} {2011})}\BibitemShut {NoStop}%
\bibitem [{\citenamefont {Georgii}(2011)}]{georgii2011gibbs}%
  \BibitemOpen
  \bibfield  {author} {\bibinfo {author} {\bibfnamefont {Hans-Otto}\ \bibnamefont {Georgii}},\ }\href@noop {} {\emph {\bibinfo {title} {Gibbs measures and phase transitions}}}\ (\bibinfo  {publisher} {Walter de Gruyter GmbH \& Co. KG, Berlin},\ \bibinfo {year} {2011})\BibitemShut {NoStop}%
\bibitem [{\citenamefont {Friedli}\ and\ \citenamefont {Velenik}(2017)}]{friedli2017statistical}%
  \BibitemOpen
  \bibfield  {author} {\bibinfo {author} {\bibfnamefont {Sacha}\ \bibnamefont {Friedli}}\ and\ \bibinfo {author} {\bibfnamefont {Yvan}\ \bibnamefont {Velenik}},\ }\href@noop {} {\emph {\bibinfo {title} {Statistical mechanics of lattice systems: a concrete mathematical introduction}}}\ (\bibinfo  {publisher} {Cambridge University Press},\ \bibinfo {year} {2017})\BibitemShut {NoStop}%
\bibitem [{\citenamefont {Gromov}(2010)}]{gromov2010singularities}%
  \BibitemOpen
  \bibfield  {author} {\bibinfo {author} {\bibfnamefont {Mikhail}\ \bibnamefont {Gromov}},\ }\bibfield  {title} {\enquote {\bibinfo {title} {Singularities, expanders and topology of maps. part 2: From combinatorics to topology via algebraic isoperimetry},}\ }\href@noop {} {\bibfield  {journal} {\bibinfo  {journal} {Geometric and Functional Analysis}\ }\textbf {\bibinfo {volume} {20}},\ \bibinfo {pages} {416--526} (\bibinfo {year} {2010})}\BibitemShut {NoStop}%
\bibitem [{\citenamefont {Linial*}\ and\ \citenamefont {Meshulam*}(2006)}]{linial2006homological}%
  \BibitemOpen
  \bibfield  {author} {\bibinfo {author} {\bibfnamefont {Nathan}\ \bibnamefont {Linial*}}\ and\ \bibinfo {author} {\bibfnamefont {Roy}\ \bibnamefont {Meshulam*}},\ }\bibfield  {title} {\enquote {\bibinfo {title} {Homological connectivity of random 2-complexes},}\ }\href@noop {} {\bibfield  {journal} {\bibinfo  {journal} {Combinatorica}\ }\textbf {\bibinfo {volume} {26}},\ \bibinfo {pages} {475--487} (\bibinfo {year} {2006})}\BibitemShut {NoStop}%
\bibitem [{\citenamefont {Bomb\'{\i}n}(2015)}]{bombin2015single_shot}%
  \BibitemOpen
  \bibfield  {author} {\bibinfo {author} {\bibfnamefont {H\'ector}\ \bibnamefont {Bomb\'{\i}n}},\ }\bibfield  {title} {\enquote {\bibinfo {title} {Single-shot fault-tolerant quantum error correction},}\ }\href {\doibase 10.1103/PhysRevX.5.031043} {\bibfield  {journal} {\bibinfo  {journal} {Phys. Rev. X}\ }\textbf {\bibinfo {volume} {5}},\ \bibinfo {pages} {031043} (\bibinfo {year} {2015})}\BibitemShut {NoStop}%
\bibitem [{\citenamefont {Leverrier}\ \emph {et~al.}(2015)\citenamefont {Leverrier}, \citenamefont {Tillich},\ and\ \citenamefont {Zémor}}]{leverrier2015quantum_expander}%
  \BibitemOpen
  \bibfield  {author} {\bibinfo {author} {\bibfnamefont {Anthony}\ \bibnamefont {Leverrier}}, \bibinfo {author} {\bibfnamefont {Jean-Pierre}\ \bibnamefont {Tillich}}, \ and\ \bibinfo {author} {\bibfnamefont {Gilles}\ \bibnamefont {Zémor}},\ }\bibfield  {title} {\enquote {\bibinfo {title} {Quantum expander codes},}\ }in\ \href {\doibase 10.1109/FOCS.2015.55} {\emph {\bibinfo {booktitle} {2015 IEEE 56th Annual Symposium on Foundations of Computer Science}}}\ (\bibinfo {year} {2015})\ pp.\ \bibinfo {pages} {810--824}\BibitemShut {NoStop}%
\bibitem [{\citenamefont {Breuckmann}\ and\ \citenamefont {Eberhardt}(2021)}]{breuckmann2021balanced}%
  \BibitemOpen
  \bibfield  {author} {\bibinfo {author} {\bibfnamefont {Nikolas~P.}\ \bibnamefont {Breuckmann}}\ and\ \bibinfo {author} {\bibfnamefont {Jens~N.}\ \bibnamefont {Eberhardt}},\ }\bibfield  {title} {\enquote {\bibinfo {title} {Balanced product quantum codes},}\ }\href {\doibase 10.1109/TIT.2021.3097347} {\bibfield  {journal} {\bibinfo  {journal} {IEEE Transactions on Information Theory}\ }\textbf {\bibinfo {volume} {67}},\ \bibinfo {pages} {6653--6674} (\bibinfo {year} {2021})}\BibitemShut {NoStop}%
\bibitem [{\citenamefont {Panteleev}\ and\ \citenamefont {Kalachev}(2022)}]{panteleev2022qldpc}%
  \BibitemOpen
  \bibfield  {author} {\bibinfo {author} {\bibfnamefont {Pavel}\ \bibnamefont {Panteleev}}\ and\ \bibinfo {author} {\bibfnamefont {Gleb}\ \bibnamefont {Kalachev}},\ }\bibfield  {title} {\enquote {\bibinfo {title} {Asymptotically good quantum and locally testable classical ldpc codes},}\ }in\ \href {\doibase 10.1145/3519935.3520017} {\emph {\bibinfo {booktitle} {Proceedings of the 54th Annual ACM SIGACT Symposium on Theory of Computing}}},\ \bibinfo {series and number} {STOC 2022}\ (\bibinfo  {publisher} {Association for Computing Machinery},\ \bibinfo {address} {New York, NY, USA},\ \bibinfo {year} {2022})\ p.\ \bibinfo {pages} {375–388}\BibitemShut {NoStop}%
\bibitem [{\citenamefont {Dinur}\ \emph {et~al.}(2023)\citenamefont {Dinur}, \citenamefont {Hsieh}, \citenamefont {Lin},\ and\ \citenamefont {Vidick}}]{dinur2023qldpc}%
  \BibitemOpen
  \bibfield  {author} {\bibinfo {author} {\bibfnamefont {Irit}\ \bibnamefont {Dinur}}, \bibinfo {author} {\bibfnamefont {Min-Hsiu}\ \bibnamefont {Hsieh}}, \bibinfo {author} {\bibfnamefont {Ting-Chun}\ \bibnamefont {Lin}}, \ and\ \bibinfo {author} {\bibfnamefont {Thomas}\ \bibnamefont {Vidick}},\ }\bibfield  {title} {\enquote {\bibinfo {title} {Good quantum ldpc codes with linear time decoders},}\ }in\ \href {\doibase 10.1145/3564246.3585101} {\emph {\bibinfo {booktitle} {Proceedings of the 55th Annual ACM Symposium on Theory of Computing}}},\ \bibinfo {series and number} {STOC 2023}\ (\bibinfo  {publisher} {Association for Computing Machinery},\ \bibinfo {address} {New York, NY, USA},\ \bibinfo {year} {2023})\ p.\ \bibinfo {pages} {905–918}\BibitemShut {NoStop}%
\bibitem [{\citenamefont {Franz}\ \emph {et~al.}(2001)\citenamefont {Franz}, \citenamefont {Mézard}, \citenamefont {Ricci-Tersenghi}, \citenamefont {Weigt},\ and\ \citenamefont {Zecchina}}]{franz2001ferromagnet}%
  \BibitemOpen
  \bibfield  {author} {\bibinfo {author} {\bibfnamefont {S.}~\bibnamefont {Franz}}, \bibinfo {author} {\bibfnamefont {M.}~\bibnamefont {Mézard}}, \bibinfo {author} {\bibfnamefont {F.}~\bibnamefont {Ricci-Tersenghi}}, \bibinfo {author} {\bibfnamefont {M.}~\bibnamefont {Weigt}}, \ and\ \bibinfo {author} {\bibfnamefont {R.}~\bibnamefont {Zecchina}},\ }\bibfield  {title} {\enquote {\bibinfo {title} {A ferromagnet with a glass transition},}\ }\href {\doibase 10.1209/epl/i2001-00438-4} {\bibfield  {journal} {\bibinfo  {journal} {Europhysics Letters}\ }\textbf {\bibinfo {volume} {55}},\ \bibinfo {pages} {465} (\bibinfo {year} {2001})}\BibitemShut {NoStop}%
\bibitem [{\citenamefont {Placke}\ \emph {et~al.}()\citenamefont {Placke}, \citenamefont {Rakovszky}, \citenamefont {Sommers}, \citenamefont {Breuckmann},\ and\ \citenamefont {Khemani}}]{placke2024classical_glass}%
  \BibitemOpen
  \bibfield  {author} {\bibinfo {author} {\bibfnamefont {Benedikt}\ \bibnamefont {Placke}}, \bibinfo {author} {\bibfnamefont {Tibor}\ \bibnamefont {Rakovszky}}, \bibinfo {author} {\bibfnamefont {Grace}\ \bibnamefont {Sommers}}, \bibinfo {author} {\bibfnamefont {Nikolas~P.}\ \bibnamefont {Breuckmann}}, \ and\ \bibinfo {author} {\bibfnamefont {Vedika}\ \bibnamefont {Khemani}},\ }\href@noop {} {\enquote {\bibinfo {title} {Spin glass order from expansion in classical ldpc codes},}\ }\bibinfo {note} {Forthcoming.}\BibitemShut {Stop}%
\bibitem [{\citenamefont {Freedman}\ and\ \citenamefont {Hastings}(2014)}]{freedman2014nlts}%
  \BibitemOpen
  \bibfield  {author} {\bibinfo {author} {\bibfnamefont {Michael~H.}\ \bibnamefont {Freedman}}\ and\ \bibinfo {author} {\bibfnamefont {Matthew~B.}\ \bibnamefont {Hastings}},\ }\bibfield  {title} {\enquote {\bibinfo {title} {Quantum systems on non-k-hyperfinite complexes: a generalization of classical statistical mechanics on expander graphs},}\ }\href {\doibase 10.26421/QIC14.1-2-9} {\bibfield  {journal} {\bibinfo  {journal} {Quantum Information and Computation}\ }\textbf {\bibinfo {volume} {14}} (\bibinfo {year} {2014}),\ 10.26421/QIC14.1-2-9}\BibitemShut {NoStop}%
\bibitem [{\citenamefont {Anshu}\ \emph {et~al.}(2022)\citenamefont {Anshu}, \citenamefont {Breuckmann},\ and\ \citenamefont {Nirkhe}}]{anshu2022nlts}%
  \BibitemOpen
  \bibfield  {author} {\bibinfo {author} {\bibfnamefont {Anurag}\ \bibnamefont {Anshu}}, \bibinfo {author} {\bibfnamefont {Nikolas~P.}\ \bibnamefont {Breuckmann}}, \ and\ \bibinfo {author} {\bibfnamefont {Chinmay}\ \bibnamefont {Nirkhe}},\ }\bibfield  {title} {\enquote {\bibinfo {title} {Nlts hamiltonians from good quantum codes},}\ }\href {\doibase 10.1145/3564246.3585114} {\  (\bibinfo {year} {2022}),\ 10.1145/3564246.3585114},\ \Eprint {http://arxiv.org/abs/arXiv:2206.13228} {arXiv:2206.13228} \BibitemShut {NoStop}%
\bibitem [{\citenamefont {Rakovszky}\ and\ \citenamefont {Khemani}(2023)}]{rakovszky2023physics}%
  \BibitemOpen
  \bibfield  {author} {\bibinfo {author} {\bibfnamefont {Tibor}\ \bibnamefont {Rakovszky}}\ and\ \bibinfo {author} {\bibfnamefont {Vedika}\ \bibnamefont {Khemani}},\ }\bibfield  {title} {\enquote {\bibinfo {title} {The physics of (good) ldpc codes i. gauging and dualities},}\ }\href@noop {} {\bibfield  {journal} {\bibinfo  {journal} {arXiv preprint arXiv:2310.16032}\ } (\bibinfo {year} {2023})}\BibitemShut {NoStop}%
\bibitem [{\citenamefont {Rakovszky}\ and\ \citenamefont {Khemani}(2024)}]{rakovszky2023physics2}%
  \BibitemOpen
  \bibfield  {author} {\bibinfo {author} {\bibfnamefont {Tibor}\ \bibnamefont {Rakovszky}}\ and\ \bibinfo {author} {\bibfnamefont {Vedika}\ \bibnamefont {Khemani}},\ }\href@noop {} {\enquote {\bibinfo {title} {The physics of (good) ldpc codes ii. product constructions},}\ } (\bibinfo {year} {2024}),\ \Eprint {http://arxiv.org/abs/arXiv:2402.16831} {arXiv:2402.16831} \BibitemShut {NoStop}%
\bibitem [{\citenamefont {Koll{\'a}r}\ \emph {et~al.}(2019)\citenamefont {Koll{\'a}r}, \citenamefont {Fitzpatrick},\ and\ \citenamefont {Houck}}]{kollar2019hyperbolic}%
  \BibitemOpen
  \bibfield  {author} {\bibinfo {author} {\bibfnamefont {Alicia~J.}\ \bibnamefont {Koll{\'a}r}}, \bibinfo {author} {\bibfnamefont {Mattias}\ \bibnamefont {Fitzpatrick}}, \ and\ \bibinfo {author} {\bibfnamefont {Andrew~A.}\ \bibnamefont {Houck}},\ }\bibfield  {title} {\enquote {\bibinfo {title} {Hyperbolic lattices in circuit quantum electrodynamics},}\ }\href {\doibase 10.1038/s41586-019-1348-3} {\bibfield  {journal} {\bibinfo  {journal} {Nature}\ }\textbf {\bibinfo {volume} {571}},\ \bibinfo {pages} {45--50} (\bibinfo {year} {2019})}\BibitemShut {NoStop}%
\bibitem [{\citenamefont {Periwal}\ \emph {et~al.}(2021)\citenamefont {Periwal}, \citenamefont {Cooper}, \citenamefont {Kunkel}, \citenamefont {Wienand}, \citenamefont {Davis},\ and\ \citenamefont {Schleier-Smith}}]{avikar2021programmable}%
  \BibitemOpen
  \bibfield  {author} {\bibinfo {author} {\bibfnamefont {Avikar}\ \bibnamefont {Periwal}}, \bibinfo {author} {\bibfnamefont {Eric~S.}\ \bibnamefont {Cooper}}, \bibinfo {author} {\bibfnamefont {Philipp}\ \bibnamefont {Kunkel}}, \bibinfo {author} {\bibfnamefont {Julian~F.}\ \bibnamefont {Wienand}}, \bibinfo {author} {\bibfnamefont {Emily~J.}\ \bibnamefont {Davis}}, \ and\ \bibinfo {author} {\bibfnamefont {Monika}\ \bibnamefont {Schleier-Smith}},\ }\bibfield  {title} {\enquote {\bibinfo {title} {Programmable interactions and emergent geometry in an array of atom clouds},}\ }\href {\doibase 10.1038/s41586-021-04156-0} {\bibfield  {journal} {\bibinfo  {journal} {Nature}\ }\textbf {\bibinfo {volume} {600}},\ \bibinfo {pages} {630--635} (\bibinfo {year} {2021})}\BibitemShut {NoStop}%
\bibitem [{\citenamefont {Bluvstein}\ \emph {et~al.}(2024)\citenamefont {Bluvstein}, \citenamefont {Evered}, \citenamefont {Geim}, \citenamefont {Li}, \citenamefont {Zhou}, \citenamefont {Manovitz}, \citenamefont {Ebadi}, \citenamefont {Cain}, \citenamefont {Kalinowski}, \citenamefont {Hangleiter}, \citenamefont {Bonilla~Ataides}, \citenamefont {Maskara}, \citenamefont {Cong}, \citenamefont {Gao}, \citenamefont {Sales~Rodriguez}, \citenamefont {Karolyshyn}, \citenamefont {Semeghini}, \citenamefont {Gullans}, \citenamefont {Greiner}, \citenamefont {Vuleti{\'c}},\ and\ \citenamefont {Lukin}}]{bluvstein2024logical}%
  \BibitemOpen
  \bibfield  {author} {\bibinfo {author} {\bibfnamefont {Dolev}\ \bibnamefont {Bluvstein}}, \bibinfo {author} {\bibfnamefont {Simon~J.}\ \bibnamefont {Evered}}, \bibinfo {author} {\bibfnamefont {Alexandra~A.}\ \bibnamefont {Geim}}, \bibinfo {author} {\bibfnamefont {Sophie~H.}\ \bibnamefont {Li}}, \bibinfo {author} {\bibfnamefont {Hengyun}\ \bibnamefont {Zhou}}, \bibinfo {author} {\bibfnamefont {Tom}\ \bibnamefont {Manovitz}}, \bibinfo {author} {\bibfnamefont {Sepehr}\ \bibnamefont {Ebadi}}, \bibinfo {author} {\bibfnamefont {Madelyn}\ \bibnamefont {Cain}}, \bibinfo {author} {\bibfnamefont {Marcin}\ \bibnamefont {Kalinowski}}, \bibinfo {author} {\bibfnamefont {Dominik}\ \bibnamefont {Hangleiter}}, \bibinfo {author} {\bibfnamefont {J.~Pablo}\ \bibnamefont {Bonilla~Ataides}}, \bibinfo {author} {\bibfnamefont {Nishad}\ \bibnamefont {Maskara}}, \bibinfo {author} {\bibfnamefont {Iris}\ \bibnamefont {Cong}}, \bibinfo {author} {\bibfnamefont {Xun}\ \bibnamefont {Gao}}, \bibinfo {author} {\bibfnamefont {Pedro}\ \bibnamefont {Sales~Rodriguez}}, \bibinfo {author} {\bibfnamefont {Thomas}\ \bibnamefont {Karolyshyn}}, \bibinfo {author} {\bibfnamefont {Giulia}\ \bibnamefont {Semeghini}}, \bibinfo {author} {\bibfnamefont {Michael~J.}\ \bibnamefont {Gullans}}, \bibinfo {author} {\bibfnamefont {Markus}\ \bibnamefont {Greiner}}, \bibinfo {author} {\bibfnamefont {Vladan}\ \bibnamefont {Vuleti{\'c}}}, \ and\ \bibinfo {author} {\bibfnamefont {Mikhail~D.}\ \bibnamefont {Lukin}},\ }\bibfield  {title} {\enquote {\bibinfo {title} {Logical quantum processor based on reconfigurable atom arrays},}\ }\href {\doibase 10.1038/s41586-023-06927-3} {\bibfield  {journal} {\bibinfo  {journal} {Nature}\ }\textbf {\bibinfo {volume} {626}},\ \bibinfo {pages} {58--65} (\bibinfo {year} {2024})}\BibitemShut {NoStop}%
\bibitem [{\citenamefont {Xu}\ \emph {et~al.}(2024)\citenamefont {Xu}, \citenamefont {Bonilla~Ataides}, \citenamefont {Pattison}, \citenamefont {Raveendran}, \citenamefont {Bluvstein}, \citenamefont {Wurtz}, \citenamefont {Vasi{\'c}}, \citenamefont {Lukin}, \citenamefont {Jiang},\ and\ \citenamefont {Zhou}}]{xu2024constant}%
  \BibitemOpen
  \bibfield  {author} {\bibinfo {author} {\bibfnamefont {Qian}\ \bibnamefont {Xu}}, \bibinfo {author} {\bibfnamefont {J.~Pablo}\ \bibnamefont {Bonilla~Ataides}}, \bibinfo {author} {\bibfnamefont {Christopher~A.}\ \bibnamefont {Pattison}}, \bibinfo {author} {\bibfnamefont {Nithin}\ \bibnamefont {Raveendran}}, \bibinfo {author} {\bibfnamefont {Dolev}\ \bibnamefont {Bluvstein}}, \bibinfo {author} {\bibfnamefont {Jonathan}\ \bibnamefont {Wurtz}}, \bibinfo {author} {\bibfnamefont {Bane}\ \bibnamefont {Vasi{\'c}}}, \bibinfo {author} {\bibfnamefont {Mikhail~D.}\ \bibnamefont {Lukin}}, \bibinfo {author} {\bibfnamefont {Liang}\ \bibnamefont {Jiang}}, \ and\ \bibinfo {author} {\bibfnamefont {Hengyun}\ \bibnamefont {Zhou}},\ }\bibfield  {title} {\enquote {\bibinfo {title} {Constant-overhead fault-tolerant quantum computation with reconfigurable atom arrays},}\ }\href {\doibase 10.1038/s41567-024-02479-z} {\bibfield  {journal} {\bibinfo  {journal} {Nature Physics}\ }\textbf {\bibinfo {volume} {20}},\ \bibinfo {pages} {1084--1090} (\bibinfo {year} {2024})}\BibitemShut {NoStop}%
\bibitem [{\citenamefont {Ramette}\ \emph {et~al.}(2022)\citenamefont {Ramette}, \citenamefont {Sinclair}, \citenamefont {Vendeiro}, \citenamefont {Rudelis}, \citenamefont {Cetina},\ and\ \citenamefont {Vuleti\ifmmode~\acute{c}\else \'{c}\fi{}}}]{ramette2022any_to_any}%
  \BibitemOpen
  \bibfield  {author} {\bibinfo {author} {\bibfnamefont {Joshua}\ \bibnamefont {Ramette}}, \bibinfo {author} {\bibfnamefont {Josiah}\ \bibnamefont {Sinclair}}, \bibinfo {author} {\bibfnamefont {Zachary}\ \bibnamefont {Vendeiro}}, \bibinfo {author} {\bibfnamefont {Alyssa}\ \bibnamefont {Rudelis}}, \bibinfo {author} {\bibfnamefont {Marko}\ \bibnamefont {Cetina}}, \ and\ \bibinfo {author} {\bibfnamefont {Vladan}\ \bibnamefont {Vuleti\ifmmode~\acute{c}\else \'{c}\fi{}}},\ }\bibfield  {title} {\enquote {\bibinfo {title} {Any-to-any connected cavity-mediated architecture for quantum computing with trapped ions or rydberg arrays},}\ }\href {\doibase 10.1103/PRXQuantum.3.010344} {\bibfield  {journal} {\bibinfo  {journal} {PRX Quantum}\ }\textbf {\bibinfo {volume} {3}},\ \bibinfo {pages} {010344} (\bibinfo {year} {2022})}\BibitemShut {NoStop}%
\bibitem [{\citenamefont {Bravyi}\ \emph {et~al.}(2024{\natexlab{a}})\citenamefont {Bravyi}, \citenamefont {Cross}, \citenamefont {Gambetta}, \citenamefont {Maslov}, \citenamefont {Rall},\ and\ \citenamefont {Yoder}}]{bravyi2024qldpc}%
  \BibitemOpen
  \bibfield  {author} {\bibinfo {author} {\bibfnamefont {Sergey}\ \bibnamefont {Bravyi}}, \bibinfo {author} {\bibfnamefont {Andrew~W.}\ \bibnamefont {Cross}}, \bibinfo {author} {\bibfnamefont {Jay~M.}\ \bibnamefont {Gambetta}}, \bibinfo {author} {\bibfnamefont {Dmitri}\ \bibnamefont {Maslov}}, \bibinfo {author} {\bibfnamefont {Patrick}\ \bibnamefont {Rall}}, \ and\ \bibinfo {author} {\bibfnamefont {Theodore~J.}\ \bibnamefont {Yoder}},\ }\bibfield  {title} {\enquote {\bibinfo {title} {High-threshold and low-overhead fault-tolerant quantum memory},}\ }\href {\doibase 10.1038/s41586-024-07107-7} {\bibfield  {journal} {\bibinfo  {journal} {Nature}\ }\textbf {\bibinfo {volume} {627}},\ \bibinfo {pages} {778--782} (\bibinfo {year} {2024}{\natexlab{a}})}\BibitemShut {NoStop}%
\bibitem [{\citenamefont {Franchini}(2023)}]{franchini2023replica}%
  \BibitemOpen
  \bibfield  {author} {\bibinfo {author} {\bibfnamefont {Simone}\ \bibnamefont {Franchini}},\ }\bibfield  {title} {\enquote {\bibinfo {title} {Replica symmetry breaking without replicas},}\ }\href {\doibase https://doi.org/10.1016/j.aop.2023.169220} {\bibfield  {journal} {\bibinfo  {journal} {Annals of Physics}\ }\textbf {\bibinfo {volume} {450}},\ \bibinfo {pages} {169220} (\bibinfo {year} {2023})}\BibitemShut {NoStop}%
\bibitem [{SM_()}]{SM_quantum}%
  \BibitemOpen
  \href@noop {} {}\bibinfo {note} {See Supplementary Material.}\BibitemShut {Stop}%
\bibitem [{\citenamefont {Rakovszky}\ \emph {et~al.}(2024)\citenamefont {Rakovszky}, \citenamefont {Placke}, \citenamefont {Breuckmann},\ and\ \citenamefont {Khemani}}]{rakovszky2024bottleneck}%
  \BibitemOpen
  \bibfield  {author} {\bibinfo {author} {\bibfnamefont {Tibor}\ \bibnamefont {Rakovszky}}, \bibinfo {author} {\bibfnamefont {Benedikt}\ \bibnamefont {Placke}}, \bibinfo {author} {\bibfnamefont {Nikolas~P.}\ \bibnamefont {Breuckmann}}, \ and\ \bibinfo {author} {\bibfnamefont {Vedika}\ \bibnamefont {Khemani}},\ }\href@noop {} {\enquote {\bibinfo {title} {Bottlenecks in quantum channels and finite temperature phases of matter},}\ } (\bibinfo {year} {2024}),\ \Eprint {http://arxiv.org/abs/arXiv:2412.09598} {arXiv:2412.09598} \BibitemShut {NoStop}%
\bibitem [{\citenamefont {Gamarnik}\ \emph {et~al.}(2024)\citenamefont {Gamarnik}, \citenamefont {Kiani},\ and\ \citenamefont {Zlokapa}}]{garmanik2024slow}%
  \BibitemOpen
  \bibfield  {author} {\bibinfo {author} {\bibfnamefont {David}\ \bibnamefont {Gamarnik}}, \bibinfo {author} {\bibfnamefont {Bobak~T.}\ \bibnamefont {Kiani}}, \ and\ \bibinfo {author} {\bibfnamefont {Alexander}\ \bibnamefont {Zlokapa}},\ }\href@noop {} {\enquote {\bibinfo {title} {Slow mixing of quantum gibbs samplers},}\ } (\bibinfo {year} {2024}),\ \Eprint {http://arxiv.org/abs/arXiv:2411.04300} {arXiv:2411.04300} \BibitemShut {NoStop}%
\bibitem [{\citenamefont {Huse}\ and\ \citenamefont {Fisher}(1987)}]{huse_fisher1987incongruent}%
  \BibitemOpen
  \bibfield  {author} {\bibinfo {author} {\bibfnamefont {D~A}\ \bibnamefont {Huse}}\ and\ \bibinfo {author} {\bibfnamefont {D~S}\ \bibnamefont {Fisher}},\ }\bibfield  {title} {\enquote {\bibinfo {title} {Pure states in spin glasses},}\ }\href {\doibase 10.1088/0305-4470/20/15/012} {\bibfield  {journal} {\bibinfo  {journal} {Journal of Physics A: Mathematical and General}\ }\textbf {\bibinfo {volume} {20}},\ \bibinfo {pages} {L997} (\bibinfo {year} {1987})}\BibitemShut {NoStop}%
\bibitem [{\citenamefont {Levin}\ and\ \citenamefont {Peres}(2017)}]{levin2017markov}%
  \BibitemOpen
  \bibfield  {author} {\bibinfo {author} {\bibfnamefont {David~A}\ \bibnamefont {Levin}}\ and\ \bibinfo {author} {\bibfnamefont {Yuval}\ \bibnamefont {Peres}},\ }\href@noop {} {\emph {\bibinfo {title} {Markov chains and mixing times}}},\ Vol.\ \bibinfo {volume} {107}\ (\bibinfo  {publisher} {American Mathematical Soc.},\ \bibinfo {year} {2017})\BibitemShut {NoStop}%
\bibitem [{\citenamefont {Frigerio}\ \emph {et~al.}(1977)\citenamefont {Frigerio}, \citenamefont {Gorini}, \citenamefont {Kossakowski},\ and\ \citenamefont {Verri}}]{frigerio1977quantum}%
  \BibitemOpen
  \bibfield  {author} {\bibinfo {author} {\bibfnamefont {Alberto}\ \bibnamefont {Frigerio}}, \bibinfo {author} {\bibfnamefont {Vittorio}\ \bibnamefont {Gorini}}, \bibinfo {author} {\bibfnamefont {Andrzej}\ \bibnamefont {Kossakowski}}, \ and\ \bibinfo {author} {\bibfnamefont {Maurizio}\ \bibnamefont {Verri}},\ }\bibfield  {title} {\enquote {\bibinfo {title} {Quantum detailed balance and kms condition},}\ }\href@noop {} {\  (\bibinfo {year} {1977})}\BibitemShut {NoStop}%
\bibitem [{\citenamefont {Fagnola}\ and\ \citenamefont {Umanita}(2007)}]{fagnola2007generators}%
  \BibitemOpen
  \bibfield  {author} {\bibinfo {author} {\bibfnamefont {Franco}\ \bibnamefont {Fagnola}}\ and\ \bibinfo {author} {\bibfnamefont {Veronica}\ \bibnamefont {Umanita}},\ }\bibfield  {title} {\enquote {\bibinfo {title} {Generators of detailed balance quantum markov semigroups},}\ }\href@noop {} {\bibfield  {journal} {\bibinfo  {journal} {Infinite Dimensional Analysis, Quantum Probability and Related Topics}\ }\textbf {\bibinfo {volume} {10}},\ \bibinfo {pages} {335--363} (\bibinfo {year} {2007})}\BibitemShut {NoStop}%
\bibitem [{\citenamefont {Temme}\ \emph {et~al.}(2010)\citenamefont {Temme}, \citenamefont {Kastoryano}, \citenamefont {Ruskai}, \citenamefont {Wolf},\ and\ \citenamefont {Verstraete}}]{temme2010chi}%
  \BibitemOpen
  \bibfield  {author} {\bibinfo {author} {\bibfnamefont {Kristan}\ \bibnamefont {Temme}}, \bibinfo {author} {\bibfnamefont {Michael~James}\ \bibnamefont {Kastoryano}}, \bibinfo {author} {\bibfnamefont {Mary~Beth}\ \bibnamefont {Ruskai}}, \bibinfo {author} {\bibfnamefont {Michael~Marc}\ \bibnamefont {Wolf}}, \ and\ \bibinfo {author} {\bibfnamefont {Frank}\ \bibnamefont {Verstraete}},\ }\bibfield  {title} {\enquote {\bibinfo {title} {The $\chi^2$-divergence and mixing times of quantum markov processes},}\ }\href@noop {} {\bibfield  {journal} {\bibinfo  {journal} {Journal of Mathematical Physics}\ }\textbf {\bibinfo {volume} {51}} (\bibinfo {year} {2010})}\BibitemShut {NoStop}%
\bibitem [{\citenamefont {Chen}\ \emph {et~al.}(2023)\citenamefont {Chen}, \citenamefont {Kastoryano},\ and\ \citenamefont {Gily{\'e}n}}]{chen2311efficient}%
  \BibitemOpen
  \bibfield  {author} {\bibinfo {author} {\bibfnamefont {CF}~\bibnamefont {Chen}}, \bibinfo {author} {\bibfnamefont {MJ}~\bibnamefont {Kastoryano}}, \ and\ \bibinfo {author} {\bibfnamefont {A}~\bibnamefont {Gily{\'e}n}},\ }\bibfield  {title} {\enquote {\bibinfo {title} {An efficient and exact noncommutative quantum gibbs sampler (2023)},}\ }\href@noop {} {\bibfield  {journal} {\bibinfo  {journal} {arXiv preprint arXiv:2311.09207}\ } (\bibinfo {year} {2023})}\BibitemShut {NoStop}%
\bibitem [{\citenamefont {Gily{\'e}n}\ \emph {et~al.}(2024)\citenamefont {Gily{\'e}n}, \citenamefont {Chen}, \citenamefont {Doriguello},\ and\ \citenamefont {Kastoryano}}]{gilyen2024quantum}%
  \BibitemOpen
  \bibfield  {author} {\bibinfo {author} {\bibfnamefont {Andr{\'a}s}\ \bibnamefont {Gily{\'e}n}}, \bibinfo {author} {\bibfnamefont {Chi-Fang}\ \bibnamefont {Chen}}, \bibinfo {author} {\bibfnamefont {Joao~F}\ \bibnamefont {Doriguello}}, \ and\ \bibinfo {author} {\bibfnamefont {Michael~J}\ \bibnamefont {Kastoryano}},\ }\bibfield  {title} {\enquote {\bibinfo {title} {Quantum generalizations of glauber and metropolis dynamics},}\ }\href@noop {} {\bibfield  {journal} {\bibinfo  {journal} {arXiv preprint arXiv:2405.20322}\ } (\bibinfo {year} {2024})}\BibitemShut {NoStop}%
\bibitem [{\citenamefont {Ding}\ \emph {et~al.}(2024)\citenamefont {Ding}, \citenamefont {Li},\ and\ \citenamefont {Lin}}]{ding2024efficient}%
  \BibitemOpen
  \bibfield  {author} {\bibinfo {author} {\bibfnamefont {Zhiyan}\ \bibnamefont {Ding}}, \bibinfo {author} {\bibfnamefont {Bowen}\ \bibnamefont {Li}}, \ and\ \bibinfo {author} {\bibfnamefont {Lin}\ \bibnamefont {Lin}},\ }\bibfield  {title} {\enquote {\bibinfo {title} {Efficient quantum gibbs samplers with kubo--martin--schwinger detailed balance condition},}\ }\href@noop {} {\bibfield  {journal} {\bibinfo  {journal} {arXiv preprint arXiv:2404.05998}\ } (\bibinfo {year} {2024})}\BibitemShut {NoStop}%
\bibitem [{\citenamefont {Blume-Kohout}\ \emph {et~al.}(2010)\citenamefont {Blume-Kohout}, \citenamefont {Ng}, \citenamefont {Poulin},\ and\ \citenamefont {Viola}}]{blume2010information}%
  \BibitemOpen
  \bibfield  {author} {\bibinfo {author} {\bibfnamefont {Robin}\ \bibnamefont {Blume-Kohout}}, \bibinfo {author} {\bibfnamefont {Hui~Khoon}\ \bibnamefont {Ng}}, \bibinfo {author} {\bibfnamefont {David}\ \bibnamefont {Poulin}}, \ and\ \bibinfo {author} {\bibfnamefont {Lorenza}\ \bibnamefont {Viola}},\ }\bibfield  {title} {\enquote {\bibinfo {title} {Information-preserving structures: A general framework for quantum zero-error information},}\ }\href@noop {} {\bibfield  {journal} {\bibinfo  {journal} {Physical Review A—Atomic, Molecular, and Optical Physics}\ }\textbf {\bibinfo {volume} {82}},\ \bibinfo {pages} {062306} (\bibinfo {year} {2010})}\BibitemShut {NoStop}%
\bibitem [{\citenamefont {Baumgartner}\ and\ \citenamefont {Narnhofer}(2012)}]{baumgartner2012structures}%
  \BibitemOpen
  \bibfield  {author} {\bibinfo {author} {\bibfnamefont {Bernhard}\ \bibnamefont {Baumgartner}}\ and\ \bibinfo {author} {\bibfnamefont {Heide}\ \bibnamefont {Narnhofer}},\ }\bibfield  {title} {\enquote {\bibinfo {title} {The structures of state space concerning quantum dynamical semigroups},}\ }\href@noop {} {\bibfield  {journal} {\bibinfo  {journal} {Reviews in Mathematical Physics}\ }\textbf {\bibinfo {volume} {24}},\ \bibinfo {pages} {1250001} (\bibinfo {year} {2012})}\BibitemShut {NoStop}%
\bibitem [{\citenamefont {Dauphinais}\ \emph {et~al.}(2024)\citenamefont {Dauphinais}, \citenamefont {Kribs},\ and\ \citenamefont {Vasmer}}]{dauphinais2024stabilizerformalism}%
  \BibitemOpen
  \bibfield  {author} {\bibinfo {author} {\bibfnamefont {Guillaume}\ \bibnamefont {Dauphinais}}, \bibinfo {author} {\bibfnamefont {David~W.}\ \bibnamefont {Kribs}}, \ and\ \bibinfo {author} {\bibfnamefont {Michael}\ \bibnamefont {Vasmer}},\ }\bibfield  {title} {\enquote {\bibinfo {title} {Stabilizer {F}ormalism for {O}perator {A}lgebra {Q}uantum {E}rror {C}orrection},}\ }\href {\doibase 10.22331/q-2024-02-21-1261} {\bibfield  {journal} {\bibinfo  {journal} {{Quantum}}\ }\textbf {\bibinfo {volume} {8}},\ \bibinfo {pages} {1261} (\bibinfo {year} {2024})}\BibitemShut {NoStop}%
\bibitem [{\citenamefont {Chen}\ and\ \citenamefont {Grover}(2024)}]{chen2024separability}%
  \BibitemOpen
  \bibfield  {author} {\bibinfo {author} {\bibfnamefont {Yu-Hsueh}\ \bibnamefont {Chen}}\ and\ \bibinfo {author} {\bibfnamefont {Tarun}\ \bibnamefont {Grover}},\ }\bibfield  {title} {\enquote {\bibinfo {title} {Symmetry-enforced many-body separability transitions},}\ }\href {\doibase 10.1103/PRXQuantum.5.030310} {\bibfield  {journal} {\bibinfo  {journal} {PRX Quantum}\ }\textbf {\bibinfo {volume} {5}},\ \bibinfo {pages} {030310} (\bibinfo {year} {2024})}\BibitemShut {NoStop}%
\bibitem [{\citenamefont {Bakshi}\ \emph {et~al.}(2024)\citenamefont {Bakshi}, \citenamefont {Liu}, \citenamefont {Moitra},\ and\ \citenamefont {Tang}}]{bakashi2024high}%
  \BibitemOpen
  \bibfield  {author} {\bibinfo {author} {\bibfnamefont {Ainesh}\ \bibnamefont {Bakshi}}, \bibinfo {author} {\bibfnamefont {Allen}\ \bibnamefont {Liu}}, \bibinfo {author} {\bibfnamefont {Ankur}\ \bibnamefont {Moitra}}, \ and\ \bibinfo {author} {\bibfnamefont {Ewin}\ \bibnamefont {Tang}},\ }\href@noop {} {\enquote {\bibinfo {title} {High-temperature gibbs states are unentangled and efficiently preparable},}\ } (\bibinfo {year} {2024}),\ \Eprint {http://arxiv.org/abs/arXiv:2403.16850} {arXiv:2403.16850} \BibitemShut {NoStop}%
\bibitem [{\citenamefont {Kitaev}(2003)}]{kitaev2003}%
  \BibitemOpen
  \bibfield  {author} {\bibinfo {author} {\bibfnamefont {A.Yu.}\ \bibnamefont {Kitaev}},\ }\bibfield  {title} {\enquote {\bibinfo {title} {Fault-tolerant quantum computation by anyons},}\ }\href {\doibase https://doi.org/10.1016/S0003-4916(02)00018-0} {\bibfield  {journal} {\bibinfo  {journal} {Annals of Physics}\ }\textbf {\bibinfo {volume} {303}},\ \bibinfo {pages} {2--30} (\bibinfo {year} {2003})}\BibitemShut {NoStop}%
\bibitem [{\citenamefont {Hoory}\ \emph {et~al.}(2006)\citenamefont {Hoory}, \citenamefont {Linial},\ and\ \citenamefont {Wigderson}}]{hoory2006expander}%
  \BibitemOpen
  \bibfield  {author} {\bibinfo {author} {\bibfnamefont {Shlomo}\ \bibnamefont {Hoory}}, \bibinfo {author} {\bibfnamefont {Nathan}\ \bibnamefont {Linial}}, \ and\ \bibinfo {author} {\bibfnamefont {Avi}\ \bibnamefont {Wigderson}},\ }\bibfield  {title} {\enquote {\bibinfo {title} {Expander graphs and their applications},}\ }\href@noop {} {\bibfield  {journal} {\bibinfo  {journal} {Bulletin of the American Mathematical Society}\ }\textbf {\bibinfo {volume} {43}},\ \bibinfo {pages} {439--561} (\bibinfo {year} {2006})}\BibitemShut {NoStop}%
\bibitem [{\citenamefont {Bravyi}\ \emph {et~al.}(2010{\natexlab{b}})\citenamefont {Bravyi}, \citenamefont {Poulin},\ and\ \citenamefont {Terhal}}]{bpt2010}%
  \BibitemOpen
  \bibfield  {author} {\bibinfo {author} {\bibfnamefont {Sergey}\ \bibnamefont {Bravyi}}, \bibinfo {author} {\bibfnamefont {David}\ \bibnamefont {Poulin}}, \ and\ \bibinfo {author} {\bibfnamefont {Barbara}\ \bibnamefont {Terhal}},\ }\bibfield  {title} {\enquote {\bibinfo {title} {Tradeoffs for reliable quantum information storage in 2d systems},}\ }\href {\doibase 10.1103/PhysRevLett.104.050503} {\bibfield  {journal} {\bibinfo  {journal} {Phys. Rev. Lett.}\ }\textbf {\bibinfo {volume} {104}},\ \bibinfo {pages} {050503} (\bibinfo {year} {2010}{\natexlab{b}})}\BibitemShut {NoStop}%
\bibitem [{\citenamefont {Eggarter}(1974)}]{eggarter1974cayley}%
  \BibitemOpen
  \bibfield  {author} {\bibinfo {author} {\bibfnamefont {T.~P.}\ \bibnamefont {Eggarter}},\ }\bibfield  {title} {\enquote {\bibinfo {title} {{Cayley trees, the Ising problem, and the thermodynamic limit}},}\ }\href {\doibase 10.1103/physrevb.9.2989} {\bibfield  {journal} {\bibinfo  {journal} {Physical Review B}\ }\textbf {\bibinfo {volume} {9}},\ \bibinfo {pages} {2989--2992} (\bibinfo {year} {1974})}\BibitemShut {NoStop}%
\bibitem [{\citenamefont {Weinstein}\ \emph {et~al.}(2019)\citenamefont {Weinstein}, \citenamefont {Ortiz},\ and\ \citenamefont {Nussinov}}]{weinstein2019universality}%
  \BibitemOpen
  \bibfield  {author} {\bibinfo {author} {\bibfnamefont {Zack}\ \bibnamefont {Weinstein}}, \bibinfo {author} {\bibfnamefont {Gerardo}\ \bibnamefont {Ortiz}}, \ and\ \bibinfo {author} {\bibfnamefont {Zohar}\ \bibnamefont {Nussinov}},\ }\bibfield  {title} {\enquote {\bibinfo {title} {Universality classes of stabilizer code hamiltonians},}\ }\href {\doibase 10.1103/PhysRevLett.123.230503} {\bibfield  {journal} {\bibinfo  {journal} {Phys. Rev. Lett.}\ }\textbf {\bibinfo {volume} {123}},\ \bibinfo {pages} {230503} (\bibinfo {year} {2019})}\BibitemShut {NoStop}%
\bibitem [{\citenamefont {Hong}\ \emph {et~al.}(2024)\citenamefont {Hong}, \citenamefont {Guo},\ and\ \citenamefont {Lucas}}]{Hong2024quantum_memory}%
  \BibitemOpen
  \bibfield  {author} {\bibinfo {author} {\bibfnamefont {Yifan}\ \bibnamefont {Hong}}, \bibinfo {author} {\bibfnamefont {Jinkang}\ \bibnamefont {Guo}}, \ and\ \bibinfo {author} {\bibfnamefont {Andrew}\ \bibnamefont {Lucas}},\ }\href {\doibase https://doi.org/10.48550/arXiv.2403.10599} {\enquote {\bibinfo {title} {Quantum memory at nonzero temperature in a thermodynamically trivial system},}\ } (\bibinfo {year} {2024}),\ \Eprint {http://arxiv.org/abs/arXiv:2403.10599} {arXiv:2403.10599} \BibitemShut {NoStop}%
\bibitem [{\citenamefont {Bravyi}\ and\ \citenamefont {Haah}(2013)}]{bravyi2013quantum_self_correction}%
  \BibitemOpen
  \bibfield  {author} {\bibinfo {author} {\bibfnamefont {Sergey}\ \bibnamefont {Bravyi}}\ and\ \bibinfo {author} {\bibfnamefont {Jeongwan}\ \bibnamefont {Haah}},\ }\bibfield  {title} {\enquote {\bibinfo {title} {Quantum self-correction in the 3d cubic code model},}\ }\href {\doibase 10.1103/PhysRevLett.111.200501} {\bibfield  {journal} {\bibinfo  {journal} {Phys. Rev. Lett.}\ }\textbf {\bibinfo {volume} {111}},\ \bibinfo {pages} {200501} (\bibinfo {year} {2013})}\BibitemShut {NoStop}%
\bibitem [{\citenamefont {Prem}\ \emph {et~al.}(2017)\citenamefont {Prem}, \citenamefont {Haah},\ and\ \citenamefont {Nandkishore}}]{prem2017glassy}%
  \BibitemOpen
  \bibfield  {author} {\bibinfo {author} {\bibfnamefont {Abhinav}\ \bibnamefont {Prem}}, \bibinfo {author} {\bibfnamefont {Jeongwan}\ \bibnamefont {Haah}}, \ and\ \bibinfo {author} {\bibfnamefont {Rahul}\ \bibnamefont {Nandkishore}},\ }\bibfield  {title} {\enquote {\bibinfo {title} {Glassy quantum dynamics in translation invariant fracton models},}\ }\href@noop {} {\bibfield  {journal} {\bibinfo  {journal} {Physical Review B}\ }\textbf {\bibinfo {volume} {95}},\ \bibinfo {pages} {155133} (\bibinfo {year} {2017})}\BibitemShut {NoStop}%
\bibitem [{\citenamefont {Tillich}\ and\ \citenamefont {Zemor}(2009)}]{tillich2009hgp}%
  \BibitemOpen
  \bibfield  {author} {\bibinfo {author} {\bibfnamefont {Jean-Pierre}\ \bibnamefont {Tillich}}\ and\ \bibinfo {author} {\bibfnamefont {Gilles}\ \bibnamefont {Zemor}},\ }\bibfield  {title} {\enquote {\bibinfo {title} {Quantum ldpc codes with positive rate and minimum distance proportional to n½},}\ }in\ \href {\doibase 10.1109/ISIT.2009.5205648} {\emph {\bibinfo {booktitle} {2009 IEEE International Symposium on Information Theory}}}\ (\bibinfo {year} {2009})\ pp.\ \bibinfo {pages} {799--803}\BibitemShut {NoStop}%
\bibitem [{\citenamefont {Fawzi}\ \emph {et~al.}(2018{\natexlab{a}})\citenamefont {Fawzi}, \citenamefont {Grospellier},\ and\ \citenamefont {Leverrier}}]{Fawzi2017efficient}%
  \BibitemOpen
  \bibfield  {author} {\bibinfo {author} {\bibfnamefont {Omar}\ \bibnamefont {Fawzi}}, \bibinfo {author} {\bibfnamefont {Antoine}\ \bibnamefont {Grospellier}}, \ and\ \bibinfo {author} {\bibfnamefont {Anthony}\ \bibnamefont {Leverrier}},\ }\bibfield  {title} {\enquote {\bibinfo {title} {Efficient decoding of random errors for quantum expander codes},}\ }in\ \href {\doibase 10.1145/3188745.3188886} {\emph {\bibinfo {booktitle} {Proceedings of the 50th Annual ACM SIGACT Symposium on Theory of Computing}}},\ \bibinfo {series and number} {STOC 2018}\ (\bibinfo  {publisher} {Association for Computing Machinery},\ \bibinfo {address} {New York, NY, USA},\ \bibinfo {year} {2018})\ p.\ \bibinfo {pages} {521–534}\BibitemShut {NoStop}%
\bibitem [{\citenamefont {Fawzi}\ \emph {et~al.}(2018{\natexlab{b}})\citenamefont {Fawzi}, \citenamefont {Grospellier},\ and\ \citenamefont {Leverrier}}]{Fawzi2018constant}%
  \BibitemOpen
  \bibfield  {author} {\bibinfo {author} {\bibfnamefont {Omar}\ \bibnamefont {Fawzi}}, \bibinfo {author} {\bibfnamefont {Antoine}\ \bibnamefont {Grospellier}}, \ and\ \bibinfo {author} {\bibfnamefont {Anthony}\ \bibnamefont {Leverrier}},\ }\bibfield  {title} {\enquote {\bibinfo {title} {Constant overhead quantum fault-tolerance with quantum expander codes},}\ }in\ \href {\doibase 10.1109/FOCS.2018.00076} {\emph {\bibinfo {booktitle} {2018 IEEE 59th Annual Symposium on Foundations of Computer Science (FOCS)}}}\ (\bibinfo {year} {2018})\ pp.\ \bibinfo {pages} {743--754}\BibitemShut {NoStop}%
\bibitem [{\citenamefont {Quintavalle}\ \emph {et~al.}(2021)\citenamefont {Quintavalle}, \citenamefont {Vasmer}, \citenamefont {Roffe},\ and\ \citenamefont {Campbell}}]{quintavalle2021single_shot}%
  \BibitemOpen
  \bibfield  {author} {\bibinfo {author} {\bibfnamefont {Armanda~O.}\ \bibnamefont {Quintavalle}}, \bibinfo {author} {\bibfnamefont {Michael}\ \bibnamefont {Vasmer}}, \bibinfo {author} {\bibfnamefont {Joschka}\ \bibnamefont {Roffe}}, \ and\ \bibinfo {author} {\bibfnamefont {Earl~T.}\ \bibnamefont {Campbell}},\ }\bibfield  {title} {\enquote {\bibinfo {title} {Single-shot error correction of three-dimensional homological product codes},}\ }\href {\doibase 10.1103/PRXQuantum.2.020340} {\bibfield  {journal} {\bibinfo  {journal} {PRX Quantum}\ }\textbf {\bibinfo {volume} {2}},\ \bibinfo {pages} {020340} (\bibinfo {year} {2021})}\BibitemShut {NoStop}%
\bibitem [{\citenamefont {Bravyi}\ \emph {et~al.}(2024{\natexlab{b}})\citenamefont {Bravyi}, \citenamefont {Lee}, \citenamefont {Li},\ and\ \citenamefont {Yoshida}}]{bravyi2024entanglement}%
  \BibitemOpen
  \bibfield  {author} {\bibinfo {author} {\bibfnamefont {Sergey}\ \bibnamefont {Bravyi}}, \bibinfo {author} {\bibfnamefont {Dongjin}\ \bibnamefont {Lee}}, \bibinfo {author} {\bibfnamefont {Zhi}\ \bibnamefont {Li}}, \ and\ \bibinfo {author} {\bibfnamefont {Beni}\ \bibnamefont {Yoshida}},\ }\href@noop {} {\enquote {\bibinfo {title} {How much entanglement is needed for quantum error correction?}}\ } (\bibinfo {year} {2024}{\natexlab{b}}),\ \Eprint {http://arxiv.org/abs/arXiv:2405.01332} {arXiv:2405.01332} \BibitemShut {NoStop}%
\bibitem [{\citenamefont {Schumacher}\ and\ \citenamefont {Nielsen}(1996)}]{schumacher1996quantum}%
  \BibitemOpen
  \bibfield  {author} {\bibinfo {author} {\bibfnamefont {Benjamin}\ \bibnamefont {Schumacher}}\ and\ \bibinfo {author} {\bibfnamefont {Michael~A}\ \bibnamefont {Nielsen}},\ }\bibfield  {title} {\enquote {\bibinfo {title} {Quantum data processing and error correction},}\ }\href@noop {} {\bibfield  {journal} {\bibinfo  {journal} {Physical Review A}\ }\textbf {\bibinfo {volume} {54}},\ \bibinfo {pages} {2629} (\bibinfo {year} {1996})}\BibitemShut {NoStop}%
\bibitem [{\citenamefont {Roeck}\ \emph {et~al.}(2024)\citenamefont {Roeck}, \citenamefont {Khemani}, \citenamefont {Li}, \citenamefont {O'Dea},\ and\ \citenamefont {Rakovszky}}]{de_roeck2024ldpc_stability}%
  \BibitemOpen
  \bibfield  {author} {\bibinfo {author} {\bibfnamefont {Wojciech~De}\ \bibnamefont {Roeck}}, \bibinfo {author} {\bibfnamefont {Vedika}\ \bibnamefont {Khemani}}, \bibinfo {author} {\bibfnamefont {Yaodong}\ \bibnamefont {Li}}, \bibinfo {author} {\bibfnamefont {Nicholas}\ \bibnamefont {O'Dea}}, \ and\ \bibinfo {author} {\bibfnamefont {Tibor}\ \bibnamefont {Rakovszky}},\ }\href@noop {} {\enquote {\bibinfo {title} {Ldpc stabilizer codes as gapped quantum phases: stability under graph-local perturbations},}\ } (\bibinfo {year} {2024}),\ \Eprint {http://arxiv.org/abs/arXiv:2411.02384} {arXiv:2411.02384} \BibitemShut {NoStop}%
\bibitem [{\citenamefont {Yin}\ and\ \citenamefont {Lucas}(2024)}]{yin2024ldpc_stability}%
  \BibitemOpen
  \bibfield  {author} {\bibinfo {author} {\bibfnamefont {Chao}\ \bibnamefont {Yin}}\ and\ \bibinfo {author} {\bibfnamefont {Andrew}\ \bibnamefont {Lucas}},\ }\href@noop {} {\enquote {\bibinfo {title} {Low-density parity-check codes as stable phases of quantum matter},}\ } (\bibinfo {year} {2024}),\ \Eprint {http://arxiv.org/abs/arXiv:2411.01002} {arXiv:2411.01002} \BibitemShut {NoStop}%
\bibitem [{\citenamefont {Anshu}\ and\ \citenamefont {Breuckmann}(2022)}]{anshu2022cnlts}%
  \BibitemOpen
  \bibfield  {author} {\bibinfo {author} {\bibfnamefont {Anurag}\ \bibnamefont {Anshu}}\ and\ \bibinfo {author} {\bibfnamefont {Nikolas~P.}\ \bibnamefont {Breuckmann}},\ }\bibfield  {title} {\enquote {\bibinfo {title} {A construction of combinatorial nlts},}\ }\href {\doibase 10.1063/5.0113731} {\bibfield  {journal} {\bibinfo  {journal} {Journal of Mathematical Physics}\ }\textbf {\bibinfo {volume} {63}},\ \bibinfo {pages} {122201} (\bibinfo {year} {2022})}\BibitemShut {NoStop}%
\bibitem [{\citenamefont {Aharonov}\ \emph {et~al.}(2013)\citenamefont {Aharonov}, \citenamefont {Arad},\ and\ \citenamefont {Vidick}}]{aharonov2021qpcp}%
  \BibitemOpen
  \bibfield  {author} {\bibinfo {author} {\bibfnamefont {Dorit}\ \bibnamefont {Aharonov}}, \bibinfo {author} {\bibfnamefont {Itai}\ \bibnamefont {Arad}}, \ and\ \bibinfo {author} {\bibfnamefont {Thomas}\ \bibnamefont {Vidick}},\ }\bibfield  {title} {\enquote {\bibinfo {title} {Guest column: the quantum pcp conjecture},}\ }\href {\doibase 10.1145/2491533.2491549} {\bibfield  {journal} {\bibinfo  {journal} {SIGACT News}\ }\textbf {\bibinfo {volume} {44}},\ \bibinfo {pages} {47–79} (\bibinfo {year} {2013})}\BibitemShut {NoStop}%
\bibitem [{\citenamefont {Mazumdar}\ and\ \citenamefont {Rawat}(2015)}]{arya2015associative}%
  \BibitemOpen
  \bibfield  {author} {\bibinfo {author} {\bibfnamefont {Arya}\ \bibnamefont {Mazumdar}}\ and\ \bibinfo {author} {\bibfnamefont {Ankit~Singh}\ \bibnamefont {Rawat}},\ }\bibfield  {title} {\enquote {\bibinfo {title} {Associative memory via a sparse recovery model},}\ }in\ \href {https://proceedings.neurips.cc/paper_files/paper/2015/file/020c8bfac8de160d4c5543b96d1fdede-Paper.pdf} {\emph {\bibinfo {booktitle} {Advances in Neural Information Processing Systems}}},\ Vol.~\bibinfo {volume} {28},\ \bibinfo {editor} {edited by\ \bibinfo {editor} {\bibfnamefont {C.}~\bibnamefont {Cortes}}, \bibinfo {editor} {\bibfnamefont {N.}~\bibnamefont {Lawrence}}, \bibinfo {editor} {\bibfnamefont {D.}~\bibnamefont {Lee}}, \bibinfo {editor} {\bibfnamefont {M.}~\bibnamefont {Sugiyama}}, \ and\ \bibinfo {editor} {\bibfnamefont {R.}~\bibnamefont {Garnett}}}\ (\bibinfo  {publisher} {Curran Associates, Inc.},\ \bibinfo {year} {2015})\BibitemShut {NoStop}%
\end{thebibliography}%


\begin{thebibliography}{31}%
\makeatletter
\providecommand \@ifxundefined [1]{%
 \@ifx{#1\undefined}
}%
\providecommand \@ifnum [1]{%
 \ifnum #1\expandafter \@firstoftwo
 \else \expandafter \@secondoftwo
 \fi
}%
\providecommand \@ifx [1]{%
 \ifx #1\expandafter \@firstoftwo
 \else \expandafter \@secondoftwo
 \fi
}%
\providecommand \natexlab [1]{#1}%
\providecommand \enquote  [1]{``#1''}%
\providecommand \bibnamefont  [1]{#1}%
\providecommand \bibfnamefont [1]{#1}%
\providecommand \citenamefont [1]{#1}%
\providecommand \href@noop [0]{\@secondoftwo}%
\providecommand \href [0]{\begingroup \@sanitize@url \@href}%
\providecommand \@href[1]{\@@startlink{#1}\@@href}%
\providecommand \@@href[1]{\endgroup#1\@@endlink}%
\providecommand \@sanitize@url [0]{\catcode `\\12\catcode `\$12\catcode `\&12\catcode `\#12\catcode `\^12\catcode `\_12\catcode `\%12\relax}%
\providecommand \@@startlink[1]{}%
\providecommand \@@endlink[0]{}%
\providecommand \url  [0]{\begingroup\@sanitize@url \@url }%
\providecommand \@url [1]{\endgroup\@href {#1}{\urlprefix }}%
\providecommand \urlprefix  [0]{URL }%
\providecommand \Eprint [0]{\href }%
\providecommand \doibase [0]{http://dx.doi.org/}%
\providecommand \selectlanguage [0]{\@gobble}%
\providecommand \bibinfo  [0]{\@secondoftwo}%
\providecommand \bibfield  [0]{\@secondoftwo}%
\providecommand \translation [1]{[#1]}%
\providecommand \BibitemOpen [0]{}%
\providecommand \bibitemStop [0]{}%
\providecommand \bibitemNoStop [0]{.\EOS\space}%
\providecommand \EOS [0]{\spacefactor3000\relax}%
\providecommand \BibitemShut  [1]{\csname bibitem#1\endcsname}%
\let\auto@bib@innerbib\@empty
\bibitem [{\citenamefont {Bravyi}\ \emph {et~al.}(2024)\citenamefont {Bravyi}, \citenamefont {Lee}, \citenamefont {Li},\ and\ \citenamefont {Yoshida}}]{bravyi2024entanglement}%
  \BibitemOpen
  \bibfield  {author} {\bibinfo {author} {\bibfnamefont {Sergey}\ \bibnamefont {Bravyi}}, \bibinfo {author} {\bibfnamefont {Dongjin}\ \bibnamefont {Lee}}, \bibinfo {author} {\bibfnamefont {Zhi}\ \bibnamefont {Li}}, \ and\ \bibinfo {author} {\bibfnamefont {Beni}\ \bibnamefont {Yoshida}},\ }\href@noop {} {\enquote {\bibinfo {title} {How much entanglement is needed for quantum error correction?}}\ } (\bibinfo {year} {2024}),\ \Eprint {http://arxiv.org/abs/arXiv:2405.01332} {arXiv:2405.01332} \BibitemShut {NoStop}%
\bibitem [{\citenamefont {Rakovszky}\ \emph {et~al.}(2024)\citenamefont {Rakovszky}, \citenamefont {Placke}, \citenamefont {Breuckmann},\ and\ \citenamefont {Khemani}}]{rakovszky2024bottleneck}%
  \BibitemOpen
  \bibfield  {author} {\bibinfo {author} {\bibfnamefont {Tibor}\ \bibnamefont {Rakovszky}}, \bibinfo {author} {\bibfnamefont {Benedikt}\ \bibnamefont {Placke}}, \bibinfo {author} {\bibfnamefont {Nikolas~P.}\ \bibnamefont {Breuckmann}}, \ and\ \bibinfo {author} {\bibfnamefont {Vedika}\ \bibnamefont {Khemani}},\ }\href@noop {} {\enquote {\bibinfo {title} {Bottlenecks in quantum channels and finite temperature phases of matter},}\ } (\bibinfo {year} {2024}),\ \Eprint {http://arxiv.org/abs/arXiv:2412.09598} {arXiv:2412.09598} \BibitemShut {NoStop}%
\bibitem [{\citenamefont {Baumgartner}\ and\ \citenamefont {Narnhofer}(2012)}]{baumgartner2012structures}%
  \BibitemOpen
  \bibfield  {author} {\bibinfo {author} {\bibfnamefont {Bernhard}\ \bibnamefont {Baumgartner}}\ and\ \bibinfo {author} {\bibfnamefont {Heide}\ \bibnamefont {Narnhofer}},\ }\bibfield  {title} {\enquote {\bibinfo {title} {The structures of state space concerning quantum dynamical semigroups},}\ }\href@noop {} {\bibfield  {journal} {\bibinfo  {journal} {Reviews in Mathematical Physics}\ }\textbf {\bibinfo {volume} {24}},\ \bibinfo {pages} {1250001} (\bibinfo {year} {2012})}\BibitemShut {NoStop}%
\bibitem [{\citenamefont {Blume-Kohout}\ \emph {et~al.}(2010)\citenamefont {Blume-Kohout}, \citenamefont {Ng}, \citenamefont {Poulin},\ and\ \citenamefont {Viola}}]{blume2010information}%
  \BibitemOpen
  \bibfield  {author} {\bibinfo {author} {\bibfnamefont {Robin}\ \bibnamefont {Blume-Kohout}}, \bibinfo {author} {\bibfnamefont {Hui~Khoon}\ \bibnamefont {Ng}}, \bibinfo {author} {\bibfnamefont {David}\ \bibnamefont {Poulin}}, \ and\ \bibinfo {author} {\bibfnamefont {Lorenza}\ \bibnamefont {Viola}},\ }\bibfield  {title} {\enquote {\bibinfo {title} {Information-preserving structures: A general framework for quantum zero-error information},}\ }\href@noop {} {\bibfield  {journal} {\bibinfo  {journal} {Physical Review A—Atomic, Molecular, and Optical Physics}\ }\textbf {\bibinfo {volume} {82}},\ \bibinfo {pages} {062306} (\bibinfo {year} {2010})}\BibitemShut {NoStop}%
\bibitem [{\citenamefont {Steane}(1996)}]{steane1996css}%
  \BibitemOpen
  \bibfield  {author} {\bibinfo {author} {\bibfnamefont {A.~M.}\ \bibnamefont {Steane}},\ }\bibfield  {title} {\enquote {\bibinfo {title} {Error correcting codes in quantum theory},}\ }\href {\doibase 10.1103/PhysRevLett.77.793} {\bibfield  {journal} {\bibinfo  {journal} {Phys. Rev. Lett.}\ }\textbf {\bibinfo {volume} {77}},\ \bibinfo {pages} {793--797} (\bibinfo {year} {1996})}\BibitemShut {NoStop}%
\bibitem [{\citenamefont {Calderbank}\ and\ \citenamefont {Shor}(1996)}]{calderbank1996css}%
  \BibitemOpen
  \bibfield  {author} {\bibinfo {author} {\bibfnamefont {A.~R.}\ \bibnamefont {Calderbank}}\ and\ \bibinfo {author} {\bibfnamefont {Peter~W.}\ \bibnamefont {Shor}},\ }\bibfield  {title} {\enquote {\bibinfo {title} {Good quantum error-correcting codes exist},}\ }\href {\doibase 10.1103/PhysRevA.54.1098} {\bibfield  {journal} {\bibinfo  {journal} {Phys. Rev. A}\ }\textbf {\bibinfo {volume} {54}},\ \bibinfo {pages} {1098--1105} (\bibinfo {year} {1996})}\BibitemShut {NoStop}%
\bibitem [{\citenamefont {Sipser}\ and\ \citenamefont {Spielman}(1996)}]{sipser_spielman1996}%
  \BibitemOpen
  \bibfield  {author} {\bibinfo {author} {\bibfnamefont {M.}~\bibnamefont {Sipser}}\ and\ \bibinfo {author} {\bibfnamefont {D.A.}\ \bibnamefont {Spielman}},\ }\bibfield  {title} {\enquote {\bibinfo {title} {Expander codes},}\ }\href {\doibase 10.1109/18.556667} {\bibfield  {journal} {\bibinfo  {journal} {IEEE Transactions on Information Theory}\ }\textbf {\bibinfo {volume} {42}},\ \bibinfo {pages} {1710--1722} (\bibinfo {year} {1996})}\BibitemShut {NoStop}%
\bibitem [{\citenamefont {Gromov}(2010)}]{gromov2010singularities}%
  \BibitemOpen
  \bibfield  {author} {\bibinfo {author} {\bibfnamefont {Mikhail}\ \bibnamefont {Gromov}},\ }\bibfield  {title} {\enquote {\bibinfo {title} {Singularities, expanders and topology of maps. part 2: From combinatorics to topology via algebraic isoperimetry},}\ }\href@noop {} {\bibfield  {journal} {\bibinfo  {journal} {Geometric and Functional Analysis}\ }\textbf {\bibinfo {volume} {20}},\ \bibinfo {pages} {416--526} (\bibinfo {year} {2010})}\BibitemShut {NoStop}%
\bibitem [{\citenamefont {Linial*}\ and\ \citenamefont {Meshulam*}(2006)}]{linial2006homological}%
  \BibitemOpen
  \bibfield  {author} {\bibinfo {author} {\bibfnamefont {Nathan}\ \bibnamefont {Linial*}}\ and\ \bibinfo {author} {\bibfnamefont {Roy}\ \bibnamefont {Meshulam*}},\ }\bibfield  {title} {\enquote {\bibinfo {title} {Homological connectivity of random 2-complexes},}\ }\href@noop {} {\bibfield  {journal} {\bibinfo  {journal} {Combinatorica}\ }\textbf {\bibinfo {volume} {26}},\ \bibinfo {pages} {475--487} (\bibinfo {year} {2006})}\BibitemShut {NoStop}%
\bibitem [{\citenamefont {Gotlib}\ and\ \citenamefont {Kaufman}(2023)}]{gotlib_kaufman_2023}%
  \BibitemOpen
  \bibfield  {author} {\bibinfo {author} {\bibfnamefont {Roy}\ \bibnamefont {Gotlib}}\ and\ \bibinfo {author} {\bibfnamefont {Tali}\ \bibnamefont {Kaufman}},\ }\bibfield  {title} {\enquote {\bibinfo {title} {Nowhere to go but high: A perspective on high-dimensional expanders},}\ }in\ \href {\doibase 10.4171/ICM2022/185} {\emph {\bibinfo {booktitle} {Lectures on the Structure of Algebraic Varieties}}}\ (\bibinfo  {publisher} {EMS Press},\ \bibinfo {year} {2023})\ pp.\ \bibinfo {pages} {4842--4871}\BibitemShut {NoStop}%
\bibitem [{\citenamefont {Breuckmann}\ and\ \citenamefont {Eberhardt}(2021{\natexlab{a}})}]{qldpc_review}%
  \BibitemOpen
  \bibfield  {author} {\bibinfo {author} {\bibfnamefont {Nikolas~P.}\ \bibnamefont {Breuckmann}}\ and\ \bibinfo {author} {\bibfnamefont {Jens~Niklas}\ \bibnamefont {Eberhardt}},\ }\bibfield  {title} {\enquote {\bibinfo {title} {Quantum low-density parity-check codes},}\ }\href@noop {} {\bibfield  {journal} {\bibinfo  {journal} {PRX Quantum}\ }\textbf {\bibinfo {volume} {2}},\ \bibinfo {pages} {040101} (\bibinfo {year} {2021}{\natexlab{a}})}\BibitemShut {NoStop}%
\bibitem [{\citenamefont {Gallager}(1960)}]{gallager1960low}%
  \BibitemOpen
  \bibfield  {author} {\bibinfo {author} {\bibfnamefont {Robert~G.}\ \bibnamefont {Gallager}},\ }\emph {\bibinfo {title} {Low density parity check codes}},\ \href@noop {} {\bibinfo {type} {Sc.d. thesis}},\ \bibinfo  {school} {Massachusetts Institute of Technology}, \bibinfo {address} {Cambridge, MA} (\bibinfo {year} {1960})\BibitemShut {NoStop}%
\bibitem [{\citenamefont {Gallager}(1962)}]{gallager1962low}%
  \BibitemOpen
  \bibfield  {author} {\bibinfo {author} {\bibfnamefont {Robert}\ \bibnamefont {Gallager}},\ }\bibfield  {title} {\enquote {\bibinfo {title} {Low-density parity-check codes},}\ }\href@noop {} {\bibfield  {journal} {\bibinfo  {journal} {IRE Transactions on information theory}\ }\textbf {\bibinfo {volume} {8}},\ \bibinfo {pages} {21--28} (\bibinfo {year} {1962})}\BibitemShut {NoStop}%
\bibitem [{\citenamefont {Richardson}\ and\ \citenamefont {Urbanke}(2008)}]{richardson2008modern}%
  \BibitemOpen
  \bibfield  {author} {\bibinfo {author} {\bibfnamefont {Tom}\ \bibnamefont {Richardson}}\ and\ \bibinfo {author} {\bibfnamefont {Ruediger}\ \bibnamefont {Urbanke}},\ }\href@noop {} {\emph {\bibinfo {title} {Modern coding theory}}}\ (\bibinfo  {publisher} {Cambridge university press},\ \bibinfo {year} {2008})\BibitemShut {NoStop}%
\bibitem [{\citenamefont {Guruswami}\ \emph {et~al.}(2023)\citenamefont {Guruswami}, \citenamefont {Rudra},\ and\ \citenamefont {Sudan}}]{guruswami2019essential}%
  \BibitemOpen
  \bibfield  {author} {\bibinfo {author} {\bibfnamefont {Venkatesan}\ \bibnamefont {Guruswami}}, \bibinfo {author} {\bibfnamefont {Atri}\ \bibnamefont {Rudra}}, \ and\ \bibinfo {author} {\bibfnamefont {Madhu}\ \bibnamefont {Sudan}},\ }\href@noop {} {\emph {\bibinfo {title} {Essential coding theory}}}\ (\bibinfo {year} {2023})\BibitemShut {NoStop}%
\bibitem [{\citenamefont {Tillich}\ and\ \citenamefont {Zemor}(2009)}]{tillich2009hgp}%
  \BibitemOpen
  \bibfield  {author} {\bibinfo {author} {\bibfnamefont {Jean-Pierre}\ \bibnamefont {Tillich}}\ and\ \bibinfo {author} {\bibfnamefont {Gilles}\ \bibnamefont {Zemor}},\ }\bibfield  {title} {\enquote {\bibinfo {title} {Quantum ldpc codes with positive rate and minimum distance proportional to n½},}\ }in\ \href {\doibase 10.1109/ISIT.2009.5205648} {\emph {\bibinfo {booktitle} {2009 IEEE International Symposium on Information Theory}}}\ (\bibinfo {year} {2009})\ pp.\ \bibinfo {pages} {799--803}\BibitemShut {NoStop}%
\bibitem [{\citenamefont {Leverrier}\ \emph {et~al.}(2015)\citenamefont {Leverrier}, \citenamefont {Tillich},\ and\ \citenamefont {Zémor}}]{leverrier2015quantum_expander}%
  \BibitemOpen
  \bibfield  {author} {\bibinfo {author} {\bibfnamefont {Anthony}\ \bibnamefont {Leverrier}}, \bibinfo {author} {\bibfnamefont {Jean-Pierre}\ \bibnamefont {Tillich}}, \ and\ \bibinfo {author} {\bibfnamefont {Gilles}\ \bibnamefont {Zémor}},\ }\bibfield  {title} {\enquote {\bibinfo {title} {Quantum expander codes},}\ }in\ \href {\doibase 10.1109/FOCS.2015.55} {\emph {\bibinfo {booktitle} {2015 IEEE 56th Annual Symposium on Foundations of Computer Science}}}\ (\bibinfo {year} {2015})\ pp.\ \bibinfo {pages} {810--824}\BibitemShut {NoStop}%
\bibitem [{\citenamefont {Fawzi}\ \emph {et~al.}(2018{\natexlab{a}})\citenamefont {Fawzi}, \citenamefont {Grospellier},\ and\ \citenamefont {Leverrier}}]{Fawzi2018constant}%
  \BibitemOpen
  \bibfield  {author} {\bibinfo {author} {\bibfnamefont {Omar}\ \bibnamefont {Fawzi}}, \bibinfo {author} {\bibfnamefont {Antoine}\ \bibnamefont {Grospellier}}, \ and\ \bibinfo {author} {\bibfnamefont {Anthony}\ \bibnamefont {Leverrier}},\ }\bibfield  {title} {\enquote {\bibinfo {title} {Constant overhead quantum fault-tolerance with quantum expander codes},}\ }in\ \href {\doibase 10.1109/FOCS.2018.00076} {\emph {\bibinfo {booktitle} {2018 IEEE 59th Annual Symposium on Foundations of Computer Science (FOCS)}}}\ (\bibinfo {year} {2018})\ pp.\ \bibinfo {pages} {743--754}\BibitemShut {NoStop}%
\bibitem [{\citenamefont {Breuckmann}\ and\ \citenamefont {Eberhardt}(2021{\natexlab{b}})}]{breuckmann2021balanced}%
  \BibitemOpen
  \bibfield  {author} {\bibinfo {author} {\bibfnamefont {Nikolas~P.}\ \bibnamefont {Breuckmann}}\ and\ \bibinfo {author} {\bibfnamefont {Jens~N.}\ \bibnamefont {Eberhardt}},\ }\bibfield  {title} {\enquote {\bibinfo {title} {Balanced product quantum codes},}\ }\href {\doibase 10.1109/TIT.2021.3097347} {\bibfield  {journal} {\bibinfo  {journal} {IEEE Transactions on Information Theory}\ }\textbf {\bibinfo {volume} {67}},\ \bibinfo {pages} {6653--6674} (\bibinfo {year} {2021}{\natexlab{b}})}\BibitemShut {NoStop}%
\bibitem [{\citenamefont {Tanner}(1981)}]{tanner1981recursive}%
  \BibitemOpen
  \bibfield  {author} {\bibinfo {author} {\bibfnamefont {R.}~\bibnamefont {Tanner}},\ }\bibfield  {title} {\enquote {\bibinfo {title} {A recursive approach to low complexity codes},}\ }\href {\doibase 10.1109/TIT.1981.1056404} {\bibfield  {journal} {\bibinfo  {journal} {IEEE Transactions on Information Theory}\ }\textbf {\bibinfo {volume} {27}},\ \bibinfo {pages} {533--547} (\bibinfo {year} {1981})}\BibitemShut {NoStop}%
\bibitem [{\citenamefont {Margulis}(1988)}]{margulis1988expanders}%
  \BibitemOpen
  \bibfield  {author} {\bibinfo {author} {\bibfnamefont {G.~A.}\ \bibnamefont {Margulis}},\ }\bibfield  {title} {\enquote {\bibinfo {title} {Explicit group-theoretical constructions of combinatorial schemes and their application to the design of expanders and concentrators},}\ }\href {https://www.mathnet.ru/php/archive.phtml?wshow=paper&jrnid=ppi&paperid=686&option_lang=eng} {\bibfield  {journal} {\bibinfo  {journal} {Problems Inform. Transmission}\ }\textbf {\bibinfo {volume} {24}},\ \bibinfo {pages} {39--46} (\bibinfo {year} {1988})},\ \bibinfo {note} {original in russian: Probl. Peredachi Inf. 24, Issue 1, pp 51--60 (1988)}\BibitemShut {NoStop}%
\bibitem [{\citenamefont {Lubotzky}\ \emph {et~al.}(1988)\citenamefont {Lubotzky}, \citenamefont {Phillips},\ and\ \citenamefont {Sarnak}}]{lps1988expanders}%
  \BibitemOpen
  \bibfield  {author} {\bibinfo {author} {\bibfnamefont {A.}~\bibnamefont {Lubotzky}}, \bibinfo {author} {\bibfnamefont {R.}~\bibnamefont {Phillips}}, \ and\ \bibinfo {author} {\bibfnamefont {P.}~\bibnamefont {Sarnak}},\ }\bibfield  {title} {\enquote {\bibinfo {title} {Ramanujan graphs},}\ }\href {\doibase 10.1007/BF02126799} {\bibfield  {journal} {\bibinfo  {journal} {Combinatorica}\ }\textbf {\bibinfo {volume} {8}},\ \bibinfo {pages} {261--277} (\bibinfo {year} {1988})}\BibitemShut {NoStop}%
\bibitem [{\citenamefont {Panteleev}\ and\ \citenamefont {Kalachev}(2022)}]{panteleev2022qldpc}%
  \BibitemOpen
  \bibfield  {author} {\bibinfo {author} {\bibfnamefont {Pavel}\ \bibnamefont {Panteleev}}\ and\ \bibinfo {author} {\bibfnamefont {Gleb}\ \bibnamefont {Kalachev}},\ }\bibfield  {title} {\enquote {\bibinfo {title} {Asymptotically good quantum and locally testable classical ldpc codes},}\ }in\ \href {\doibase 10.1145/3519935.3520017} {\emph {\bibinfo {booktitle} {Proceedings of the 54th Annual ACM SIGACT Symposium on Theory of Computing}}},\ \bibinfo {series and number} {STOC 2022}\ (\bibinfo  {publisher} {Association for Computing Machinery},\ \bibinfo {address} {New York, NY, USA},\ \bibinfo {year} {2022})\ p.\ \bibinfo {pages} {375–388}\BibitemShut {NoStop}%
\bibitem [{\citenamefont {Dinur}\ \emph {et~al.}(2023)\citenamefont {Dinur}, \citenamefont {Hsieh}, \citenamefont {Lin},\ and\ \citenamefont {Vidick}}]{dinur2023qldpc}%
  \BibitemOpen
  \bibfield  {author} {\bibinfo {author} {\bibfnamefont {Irit}\ \bibnamefont {Dinur}}, \bibinfo {author} {\bibfnamefont {Min-Hsiu}\ \bibnamefont {Hsieh}}, \bibinfo {author} {\bibfnamefont {Ting-Chun}\ \bibnamefont {Lin}}, \ and\ \bibinfo {author} {\bibfnamefont {Thomas}\ \bibnamefont {Vidick}},\ }\bibfield  {title} {\enquote {\bibinfo {title} {Good quantum ldpc codes with linear time decoders},}\ }in\ \href {\doibase 10.1145/3564246.3585101} {\emph {\bibinfo {booktitle} {Proceedings of the 55th Annual ACM Symposium on Theory of Computing}}},\ \bibinfo {series and number} {STOC 2023}\ (\bibinfo  {publisher} {Association for Computing Machinery},\ \bibinfo {address} {New York, NY, USA},\ \bibinfo {year} {2023})\ p.\ \bibinfo {pages} {905–918}\BibitemShut {NoStop}%
\bibitem [{\citenamefont {Bomb\'{\i}n}(2015)}]{bombin2015single_shot}%
  \BibitemOpen
  \bibfield  {author} {\bibinfo {author} {\bibfnamefont {H\'ector}\ \bibnamefont {Bomb\'{\i}n}},\ }\bibfield  {title} {\enquote {\bibinfo {title} {Single-shot fault-tolerant quantum error correction},}\ }\href {\doibase 10.1103/PhysRevX.5.031043} {\bibfield  {journal} {\bibinfo  {journal} {Phys. Rev. X}\ }\textbf {\bibinfo {volume} {5}},\ \bibinfo {pages} {031043} (\bibinfo {year} {2015})}\BibitemShut {NoStop}%
\bibitem [{\citenamefont {Fawzi}\ \emph {et~al.}(2018{\natexlab{b}})\citenamefont {Fawzi}, \citenamefont {Grospellier},\ and\ \citenamefont {Leverrier}}]{Fawzi2017efficient}%
  \BibitemOpen
  \bibfield  {author} {\bibinfo {author} {\bibfnamefont {Omar}\ \bibnamefont {Fawzi}}, \bibinfo {author} {\bibfnamefont {Antoine}\ \bibnamefont {Grospellier}}, \ and\ \bibinfo {author} {\bibfnamefont {Anthony}\ \bibnamefont {Leverrier}},\ }\bibfield  {title} {\enquote {\bibinfo {title} {Efficient decoding of random errors for quantum expander codes},}\ }in\ \href {\doibase 10.1145/3188745.3188886} {\emph {\bibinfo {booktitle} {Proceedings of the 50th Annual ACM SIGACT Symposium on Theory of Computing}}},\ \bibinfo {series and number} {STOC 2018}\ (\bibinfo  {publisher} {Association for Computing Machinery},\ \bibinfo {address} {New York, NY, USA},\ \bibinfo {year} {2018})\ p.\ \bibinfo {pages} {521–534}\BibitemShut {NoStop}%
\bibitem [{\citenamefont {Hong}\ \emph {et~al.}(2024)\citenamefont {Hong}, \citenamefont {Guo},\ and\ \citenamefont {Lucas}}]{Hong2024quantum_memory}%
  \BibitemOpen
  \bibfield  {author} {\bibinfo {author} {\bibfnamefont {Yifan}\ \bibnamefont {Hong}}, \bibinfo {author} {\bibfnamefont {Jinkang}\ \bibnamefont {Guo}}, \ and\ \bibinfo {author} {\bibfnamefont {Andrew}\ \bibnamefont {Lucas}},\ }\href {\doibase https://doi.org/10.48550/arXiv.2403.10599} {\enquote {\bibinfo {title} {Quantum memory at nonzero temperature in a thermodynamically trivial system},}\ } (\bibinfo {year} {2024}),\ \Eprint {http://arxiv.org/abs/arXiv:2403.10599} {arXiv:2403.10599} \BibitemShut {NoStop}%
\bibitem [{\citenamefont {Erd\"os}(1942)}]{erdos1942partition}%
  \BibitemOpen
  \bibfield  {author} {\bibinfo {author} {\bibfnamefont {P.}~\bibnamefont {Erd\"os}},\ }\bibfield  {title} {\enquote {\bibinfo {title} {On an elementary proof of some asymptotic formulas in the theory of partitions},}\ }\href {http://www.jstor.org/stable/1968802} {\bibfield  {journal} {\bibinfo  {journal} {Annals of Mathematics}\ }\textbf {\bibinfo {volume} {43}},\ \bibinfo {pages} {437--450} (\bibinfo {year} {1942})}\BibitemShut {NoStop}%
\bibitem [{\citenamefont {Worsch}(1994)}]{worsch1994_binomial_bounds}%
  \BibitemOpen
  \bibfield  {author} {\bibinfo {author} {\bibfnamefont {Thomas}\ \bibnamefont {Worsch}},\ }\href@noop {} {\emph {\bibinfo {title} {Lower and upper bounds for (sums of) binomial coefficients}}}\ (\bibinfo {year} {1994})\ \bibinfo {note} {karlsruhe 1994. (Interner Bericht. Fakultät für Informatik, Universität Karlsruhe. 1994,31.)}\BibitemShut {NoStop}%
\bibitem [{\citenamefont {Knill}\ and\ \citenamefont {Laflamme}(1997)}]{knill1997theory}%
  \BibitemOpen
  \bibfield  {author} {\bibinfo {author} {\bibfnamefont {Emanuel}\ \bibnamefont {Knill}}\ and\ \bibinfo {author} {\bibfnamefont {Raymond}\ \bibnamefont {Laflamme}},\ }\bibfield  {title} {\enquote {\bibinfo {title} {Theory of quantum error-correcting codes},}\ }\href@noop {} {\bibfield  {journal} {\bibinfo  {journal} {Physical Review A}\ }\textbf {\bibinfo {volume} {55}},\ \bibinfo {pages} {900} (\bibinfo {year} {1997})}\BibitemShut {NoStop}%
\bibitem [{\citenamefont {Nielsen}\ and\ \citenamefont {Chuang}(2010)}]{NielsenChuang2010}%
  \BibitemOpen
  \bibfield  {author} {\bibinfo {author} {\bibfnamefont {Michael~A.}\ \bibnamefont {Nielsen}}\ and\ \bibinfo {author} {\bibfnamefont {Isaac~L.}\ \bibnamefont {Chuang}},\ }\href@noop {} {\emph {\bibinfo {title} {Quantum Computation and Quantum Information: 10th Anniversary Edition}}}\ (\bibinfo  {publisher} {Cambridge University Press},\ \bibinfo {year} {2010})\BibitemShut {NoStop}%
\end{thebibliography}%

\end{document}


\title{Supplementary Material for:\\
Topological Quantum Spin Glass Order and its realization in qLDPC codes
} 

\author{Benedikt Placke}%
\affiliation{Rudolf Peierls Centre for Theoretical Physics, University of Oxford, Oxford OX1 3PU, United Kingdom}
\author{Tibor Rakovszky}
\affiliation{Department of Physics, Stanford University, Stanford, California 94305, USA},
\affiliation{Department of Theoretical Physics, Institute of Physics,
Budapest University of Technology and Economics, M\H{u}egyetem rkp. 3., H-1111 Budapest, Hungary}
\affiliation{HUN-REN-BME Quantum Error Correcting Codes and Non-equilibrium Phases Research Group,
Budapest University of Technology and Economics,
M\H{u}egyetem rkp. 3., H-1111 Budapest, Hungary}
\author{Nikolas P.\ Breuckmann}
\affiliation{School of Mathematics, University of Bristol, Bristol BS8 1UG, United Kingdom}
\author{Vedika Khemani}
\affiliation{Department of Physics, Stanford University, Stanford, California 94305, USA}

\maketitle

\tableofcontents

\section{A guide to the supplementary material}

In this supplementary material, we provide rigorous results establishing the existence of topological quantum spin glass order in qLDPC codes with sufficiently strong linear confinement. 

In \cref{suppl:gibbs_states}, we outline the Gibbs state decomposition, its relation to classical and quantum memory, and define the configurational entropy which captures one of the defining features of topological quantum spin glass order.

In \cref{suppl:expansion}, we provide a brief introduction to classical and quantum low-density parity check (LDPC) codes, fix related notation, and provide formal definitions of linear confinement in both classical and quantum LDPC codes. Finally, we introduce two known constructions of qLDPC codes with linear confinement and summarize their properties. 

The rest of the supplementary material (\autoref{suppl:ergo_breaking}-\ref{suppl:proof_profit}) is devoted to proving the following formal result

\begin{theorem}[Existence of Topological Quantum Spin Glasses]\label{thm:profit}
    There exists a family of qLDPC codes on $n$ qubits, with linear $(\delta, \gamma)$-confinement (\cref{def:expansion}), associated local Hamiltonian $\mathcal H$, and constants $\beta^*, c_1, c_2, c_3 \in \mathbb{R}_+$ such that for an eigenstate $\ket{\vec x, \vec z}$ chosen randomly from the Gibbs distribution $p(\vec x, \vec z) \propto e^{-\beta E(\vec x,\vec z)}$ with $\beta > \beta^*$, the following is true with high probability $p = 1 - \order{e^{- 
 c_1 \sqrt n}}$:
    \begin{enumerate}
        \item The state is surrounded by a bottleneck: there exists a subspace $\mathcal V(\vec x, \vec z)$ such that $\ket{\vec x, \vec z}\in\mathcal V(\vec x, \vec z)$ and
        \begin{equation}
            \Delta(n) := \frac{\tr P_{\partial \mathcal V}\rhoG}{\tr P_{\mathcal V}\rhoG} \leq e^{-c_2 \delta(n)}
        \end{equation}
        where $\partial \mathcal V$ is the $\delta(n)/4$-neighborhood of $\mathcal V$ (see \cref{def:hilbert_neighborhood}).
        \item The Gibbs state supported on $\mathcal V$ contains an exponentially small fraction of the weight:
        \begin{equation}\label{eq:profit2}
            \tr P_{\mathcal V}\rhoG \leq 2^{-[r + f(T)]n}
        \end{equation}
        where $r$ is the asymptotic rate of the code and $f(T)$, at low temperature $T$, is a positive and strictly increasing function of $T$.
        \item All states in the subspace $\mathcal V$ have circuit complexity $\Omega(\log \delta (n))$.
        \item The system serves as a passive memory for the code $\mathcal C' := \linspan\{\ket{\ell} := \hat L_{\ell} \ket{\vec x_0, \vec z_0}, \ell=1\dots 2^k\}$ with $\hat L_\ell$, $\ell=1,\dotsc, 2^k$ are the logical operators of the original code. 
        The code $\mathcal C'$ has the same number of logical qubits, $k$, as the original code, and distance $d' \geq c_3 \delta(n)$.
    \end{enumerate}
\end{theorem}

Let us briefly comment on how this formal result is related to the informal definition, that is Theorem II.1 as provided in the main text. 
The first ingredient to define a non-trivial phase of matter is to show that the bottleneck condition is fulfilled for some set of subspaces. 

Point (1) above guarantees that for a randomly chosen eigenstate, we can define a Gibbs state component supported entirely on a subspace $\mathcal V$ that ``surrounds'' the eigenstate. 

Point (2) implies both \emph{shattering} and \emph{incongruence}. 
In the stabilizer models we study, the subspaces $\mathcal{V}$ appearing in the Gibbs state decomposition are diagonal in the eigenbasis, in which case $\tr P_{\mathcal V}\rhoG$ is the weight of Gibbs state component supported on $\mathcal V$.
By \autoref{eq:profit2}, no single extremal component has more than an exponentially small fraction of the weight, and also this fraction is upper bounded by a decreasing function of temperature. In particular, this upper bound on the weight is a lower bound of the min-entropy of the extremal Gibbs states, which however lower bound the Shannon-entropy of the same ensemble, and hence the configurational entropy defined in the main text and \cref{suppl:gibbs_states}. The configurational entropy therefore is larger than the rate $r$ at any finite temperature (shattering) and, at low temperature, it is an increasing function of temperature (incongruence).

Point (3) implies long-range entanglement of the Gibbs state component supported on $\mathcal V$.

Finally, point (4), refers to the emergent quantum memories. It guarantees that the systems acts as a passive memory with respect to the code spanned by $\ket{\vec x_0, \vec z_0}$ and symmetry-related states. Due to shattering (point 3), there are exponentially many (in $n$) choices for such an encoding that will lead to distinct steady states under any local quantum channel that has the global Gibbs state as a steady state (cf. also the discussion in \cref{suppl:gibbs_states}).
The system thus passively encodes a mixture of classical (the choice of quantum code to use) and quantum information.

The proof of \cref{thm:profit} is split into multiple sections. \cref{suppl:ergo_breaking} shows how to construct $\mathcal{V}$ and $\partial\mathcal{V}$ and shows point (1). \cref{suppl:Sconf} uses a counting argument to bound the configurational entropy and prove point (2). \cref{suppl:topo_order} shows that arbitrary states in typical subspaces can be used form a constant code, which (using a result from~\cite{bravyi2024entanglement}) establishes Points (3) and (4) above. Finally, in \cref{suppl:proof_profit}, we use these general results to instantiate \cref{thm:profit} using the hypergraph product of random classical LDPC codes with appropriate parameters.

\section{Bottleneck theorem and the structure of equilibrium states\label{suppl:gibbs_states}}

Here we give a more detailed discussion of the definition and structure of Gibbs state components outlined in Section II of the main text. 

\subsection{Bottleneck theorem}

First, let us state the precise bottleneck condition introduced in Ref. \cite{rakovszky2024bottleneck} that ensures that the projected state $\rho_\mathcal{V} \propto P_\mathcal{V} \rho_G P_\mathcal{V}$ is an approximate steady state. To do so, we need to define the boundary of $\mathcal{V}$:

\begin{definition}\label{def:hilbert_neighborhood}
    The $r$-\textbf{Hilbert space neighborhood} of the subspace $\mathcal{V}$ is defined as 
\begin{equation}
    \mathcal{B}_r(\mathcal{V}) \equiv \text{Span}\{\hat{O}\ket{\psi}| \ket{\psi} \in \mathcal{V}, \hat{O} \text { acts on a most } r \text{ qubits} \}.
\end{equation}
Note that $\mathcal{B}_r(\mathcal{V})$ is itself a subspace containing $\mathcal{V}$, $\mathcal{V} \subseteq \mathcal{B}_r(\mathcal{V})$, by construction. The $r$-\textbf{boundary} of $\mathcal{V}$ is defined as the orthogonal complement of $\mathcal{V}$ within $\mathcal{B}_r(\mathcal{V})$. In terms of the corresponding projectors: $P_{\partial_r \mathcal{V}} = P_{\mathcal{B}_r(\mathcal{V})} (1\!\!1 - P_{\mathcal{V}})$.
\end{definition}

With this in hand, we can state the bottleneck theorem:

\begin{theorem}[Theorem 2 of \cite{rakovszky2024bottleneck}]
\label{thm:bottleneck}
    Let $\mathcal{M}$ be a quantum channel with steady state $\rho$ and assume that $\mathcal{M}$ has a Kraus representation where every Kraus operator is $r$-local (acts on at most $r$ qubits). Let $\mathcal{V}$ be a subspace and define  the projected state $\rho_{\mathcal{V}} = P_{\mathcal{V}} \rho P_{\mathcal{V}} / \text{Tr}(P_{\mathcal{V}}\rho)$. Then the following inequality holds
    \begin{equation}
        ||\mathcal{M}[\rho_{\mathcal{V}}] - \rho_{\mathcal{V}}||_1 \leq \frac{||P_{\partial_{2r}\mathcal{V}}\rho||_1}{\text{Tr}(P_{\mathcal{V}}\rho)} \equiv  10 \Delta.
    \end{equation}
\end{theorem}

We are interested in the case where $\rho = \rho_G \propto e^{-\beta \hamil}$ is the Gibbs state of Hamiltonian $\hamil$. Note that the bottleneck ratio defined here, $\Delta = \frac{||P_{\partial_{2r}\mathcal{V}}\rho||_1}{\text{Tr}(P_{\mathcal{V}}\rho)}$ is not quite the same expression as the one we used in the main text, phrased in terms of free energy barriers, which would correspond to  $\Delta = \frac{\text{Tr}(P_{\partial_{2r}\mathcal{V}}\rho)}{\text{Tr}(P_{\mathcal{V}}\rho)}$. The two expressions conincide if $[P_{\partial_{2r}\mathcal{V}},\rho] = 0$, which holds for stabilizer Hamiltonians such as the models we consider in this work. More generally, one can upper bound the bottleneck ratio as~\cite{rakovszky2024bottleneck} 
\begin{equation}
    \frac{||P_{\partial_{2r}\mathcal{V}}\rho||_1}{\text{Tr}(P_{\mathcal{V}}\rho)} \leq \frac{\sqrt{\text{Tr}(P_{\partial_{2r}\mathcal{V}}\rho)}}{\text{Tr}(P_{\mathcal{V}}\rho)}
\end{equation}
In Ref. \cite{rakovszky2024bottleneck}, it was shown that for codes with ``good'' expansion ($\delta(n) = \Omega(n)$) this condition is still satisfied when we weakly perturb away from the stabilizer point, at least at sufficiently low energy densities.

\subsection{Classical vs quantum memories and the structure of steady states}

The bottleneck theorem allows us to decompose the Gibbs state $\rho_G$ into approximate steady states, satisfying 
\begin{equation}\label{eq:approx_ss}
    ||\mathcal{M}\rho - \rho||_1 \leq \Delta(n)    
\end{equation}
for any local dynamics $\mathcal{M}$ that has $\rho_G$ as its steady state. In the cases we are interested in here, there will be multiple distinct approximate steady states, all satisfying \cref{eq:approx_ss} for some function $\Delta(n)$ that decays super-polynomially with system size $n$. These form a convex set, i.e. if two sates $\rho_{1,2}$ both satisfy this condition, then so does their convex combination $\lambda\rho_1 + (1-\lambda) \rho_2$. We want to understand the structure of this convex set. To do so, we rely on the structure theorem for \emph{exact} steady states (corresponding to $\Delta = 0)$, developed in Refs. \cite{baumgartner2012structures,blume2010information}. We will assume that the same overall structure carries over to the case of approximate steady states (which become increasingly close to exact ones with increasing $n$); justifying this assumption is an interesting open problem.

Exact steady states always correspond to subspaces $\mathcal{V}$ that are left invariant by the channel $\mathcal{M}$ (i.e., $\mathcal{M}^\dagger[P_\mathcal{V}] = P_{\mathcal{V}}$) and extremal points in the convex set of steady states are correspond to subspaces that cannot be decomposed into a direct sum of smaller invariant subspaces. Similarly, by \cref{thm:bottleneck}, we associate approximate steady states to subspaces that satisfy a bottleneck condition and the extremal states to those subspaces that cannot be further decomposed\footnote{We note that the situation is more complicated in this case, since the subspaces in question are only \emph{approximately} invariant, which means that we have some freedom in how we define them, e.g. what exactly we include in the boundary region $\partial\mathcal{V}$. However, this ambiguity should not affect the important features and different choices should still lead to the same steady state in the limit of $n\to\infty$.}.

Any approximate steady state can be written as a convex combination of extremal states. However, there is an important difference between the case when the space of steady states is fully classical, and the one where it includes coherent superpositions; the two cases distinguish between classical and quantum memories.  

First, let us consider the fully classical case. In that case, there is a unique set of extremal subspaces $\{\mathcal{V}_i\}_{i=1\dots\mathcal N}$ that are all orthogonal to each other. Correspondingly, any state can be uniquely decomposed into a classical mixture of extremal components, $\rho = \sum_i p_i \rho_i$. Thus, the information preserved by the channel $\mathcal{M}$ is fully specified by the classical probability distribution $\{p_i\}$ over the extremal components and the system acts as a passive classical memory with $\log_2{\mathcal{N}}$ encoded bits. 

The classical case should be contrasted with that of a passive quantum memory. Let us illustrate the idea on the simplest example, where the system preserves a single qubit's worth of quantum information. We then have two orthogonal extremal subspaces, $\mathcal{V}_{\pm}$, that have the same dimensions, and which we can associate (without loss of generality) with the $\overline{Z} = \pm 1$ eigenvalues of some `logical' Pauli $z$ operator $\overline{Z}$. Since the two subspaces have the same dimensions, we can write the combined subspace as $\mathcal{V}_+ \oplus \mathcal{V}_- = \mathbb{C}^2 \otimes \tilde{\mathcal{V}}$, decomposing it into an effective qubit degree of freedom $\mathbb{C}^2$, and a remaining part $\tilde{\mathcal{V}}$. There are two extremal Gibbs state components corresponding to $\mathcal{V}_\pm$, which take the form $\rho_{\pm} = \mathcal{P}_\pm \otimes \sigma$, where $\mathcal{P}_\pm \equiv (1\!\!1 \pm \overline{Z}) / 2$ and, importantly the density matrix $\sigma$, acting on $\tilde{\mathcal{V}}$, is the \emph{same} in both extremal states. The key point that distinguishes this case from a classical memory is that this set of extremal sates is not unique. Indeed, one can perform an arbitrary unitary rotation $\mathcal{U}$ on the qubit degree of freedom to obtain another extremal state of the form $\mathcal{U} \mathcal{P}_+ \mathcal{U}^\dagger \otimes \sigma$. Thus, the set of extremal Gibbs states is isomorphic to the Bloch sphere of a qubit, while the convex set of all Gibbs states, which includes mixture of extremal ones, is isomorphic to a Bloch ball, taking the form $\mu \otimes \sigma$ where $\mu$ is now an \emph{arbitrary} density matrix on a single qubit. 

This structure easily generalizes to the case when the system preserves a $D$-dimensional qudit (A special case of this is a passive quantum memory of $k$ qubits, wherein $D = 2^k$). In this case there are $D$ (approximately) invariant orthogonal subspaces $\{\mathcal{V}_i\}_{i=1\dots D}$, all of the same size, such that $\bigoplus_i \mathcal{V}_i = \mathbb{C}^D \otimes \tilde{\mathcal{V}}$. Other invariant subspaces can be obtained by performing rotations on the first component of this tensor product decomposition and a generic steady state again takes the form $\mu \otimes \sigma$ where $\mu$ is an arbitrary $D \times D$ density matrix while $\sigma$ is a fixed density matrix of size $\dim{\mathcal{V}_i}$. The full set of (approximate) steady states is then isomorphic to the set of all states of the $D$-dimensional qudit. Drawing an analogy with the classical case, we can think of $\mu$ as a density matrix over the set of extremal Gibbs states (similarly to how in the classical case we had a classical probability distribution $\{p_i\}$ over extremal components). 

In the most general case, relevant for the example we discuss below, the system might preserve a combination of classical and quantum information. Correspondingly, the most general decomposition for steady states is of the form 
\begin{equation}\label{eq:StructureTheorem}
    \rho_{\rm ss}(\{p_i\},\{\mu_i\}) = \sum_i p_i \mu_i \otimes \sigma_i.    
\end{equation}
For each $i$, $\mu_i$ is an arbitrary density matrix on a $D_i$-dimensional qudit, while $\sigma_i$ is fixed and the same for all steady states. The $p_i$ form a classical probability distribution. Thus, the set of all steady states is isomorphic to that of block-diagonal density matrices of the form $\bigoplus_i p_i \mu_i$ with the sizes of blocks given by $D_i$. Classical memories correspond to the case where all the $D_i = 1$, while the fully quantum case is when there is only a single block of dimension $D$. The extremal Gibbs state components correspond to the extremal points of the set of all steady states, parameterized in \autoref{eq:StructureTheorem}. Thus, they correspond to the choice $p_i = \delta_{i,i_*}$ and $\mu_{i_*} = \ket{\phi}\bra{\phi}_{i_*}$, for some $D_{i_*}$-dimensional vector $\ket{\phi}_{i_*}$.

\subsection{Configurational entropy of the Gibbs state}

We can also use the decomposition into extremal components to define entropic quantities that measure the number of relevant components that contribute to the global Gibbs state. As an element of the set of steady states, the Gibbs state, $\rho_G \propto e^{-\beta \hamil}$, corresponds to a particular choice of the probabilities $p_i$ and the encoded density matrices $\mu_i$ in the decomposition $\rho_G = \sum_i p_i \mu_i \otimes \sigma_i$. We are here interested in $\rho_G$ not as a state over the many-body Hilbert space $\hilbert$ but in terms of its decomposition into extremal components. This is characterized by the effective density matrix $\rho_\beta^{\rm eff} \equiv \bigoplus_i p_i \mu_i$ which is a block-diagonal density matrix composed of the different $\mu_i$'s weighted according to their probabilities $p_i$. We can now define a \emph{configurational entropy} 
\begin{equation}\label{eq:conf_ent_vN}
    S(\beta) \equiv -\text{Tr}(\rho_\beta^{\rm eff} \log_2{\rho_\beta^{\rm eff}}) = -\sum_i p_i \log_2{p_i} - \sum_i p_i \text{Tr}(\mu_i \log_2{\mu_i}),
\end{equation}
where in the second equality we used the block-diagonality of $\rho_\beta^{\rm eff}$ to write the entropy as the sum of the Shannon entropy of the classical distribution $p_i$ and the average von Neumann entropy of the density matrices $\mu_i$. 

$S(\beta)$ gives a measure of how many distinct components contribute to $\rho_G$ at inverse temperature $\beta$. In a spin glass, this can be a complicated function of temperature due to the competition between the decreasing probability $p_i$ of Gibbs state components associated with local minima at high energies, and the entropic factor corresponding to the potentially large number of such minima. More finely grained information on the full distribution of Gibbs state components by considering the full spectrum of R\'enyi entropies of $\rho_\beta^{\rm eff}$, given by
\begin{equation}
    S_\alpha(\beta) \equiv \frac{1}{1-\alpha} \log_2{\text{Tr}((\rho_\beta^{\rm eff})^\alpha)},
\end{equation}
which reproduces \cref{eq:conf_ent_vN} in the limit $\alpha \to 1$. The \emph{Hartley entropy}, $S_0(\beta)$ simply measures the total number of non-zero eigenvalues of $\rho_\beta^{\rm eff}$, which is the same as its total dimension $D = \sum_i D_i$ (assuming that all the $\mu_i$ are full rank, which is true in the Gibbs state with $0 < \beta < \infty)$. This quantity should only increase with $\beta$, as more and more local minima become stable. This is in contrast with the behavior of the von Neumann entropy \cref{eq:conf_ent_vN}, which can be non monotonic due to the fact that while the sheer number of components might increase with lowering temperature, fewer of them might actually contribute significantly to $\rho_G$. In the other extreme, the \emph{min-entropy} $S_\infty(\beta)$ measures the largest eigenvalue of $\rho_\beta^{\rm eff}$, which we roughly take to correspond to the size of the largest Gibbs state component. The R\'enyi entropies satisfy the inequality $S_\alpha(\beta) \geq S_{\alpha'}(\beta)$ if $\alpha \leq \alpha'$, so we can use $S_\infty(\beta)$ to lower bound \cref{eq:conf_ent_vN} for example, which is indeed the approach we will take below in \cref{suppl:Sconf}.

In principle, calculating the entropies requires calculating the eigenvalues $\lambda_{i,\alpha}$ of the matrices $\mu_i$. However, in the examples we consider, which are all commuting stabilizer code Hamiltonians, the rotations within the different blocks are all generated by exact symmetries of the Hamiltonian (the logical operators of the quantum code). As such, they do not change the energy and therefore the density matrices $\mu_i$ are all \emph{maximally mixed} states with eigenvalues $\lambda_{i,\alpha} = 1/D_i$ (we expect the same to be true more generally, even when we perturb away from the stabilizer limit, at least in some approximate sense). Thus, the spectrum of $\rho_\beta^{\rm eff}$ consists of eigenvalues of the form $p_i / D_i$, each with a degeneracy of $D_i$. In particular, the min-entropy is given by $S_\infty(\beta) = -\log_2\max_{i}\{p_i/D_i\}$. 

Finally, we can relate this to the Gibbs weight of extremal components. Consider an extremal Gibbs state component $\rhoGV$, associated to some subspace $\mathcal{V}$. As noted above, this takes the form $\ket{\phi}\bra{\phi}_{i_*} \otimes \sigma_{i_*}$ for some $i_*$ and $D_{i_*}$-dimensional pure state $\ket{\phi}_{i_*}$, which both depend on the choice of $\mathcal{V}$. Consequently, the weight of the extremal component is given by $w_\mathcal{V} = \tr(P_{\mathcal{V}}\rho_G) = p_{i_*} \langle \phi|\mu_{i_*}|\phi\rangle$. Now, making the assumption that $\mu_i = 1\!\!1 / D_i$, we get that $w_\mathcal{V} = p_{i_*}/D_{i_*}$. In other words, the eigenvalues of $\rho_\beta^{\rm eff}$ are precisely the Gibbs weights and the formula for the min-entropy becomes $S_\infty(\beta) = -\log_2\max_{\mathcal{V}}\{w_\mathcal{V}\}$. This is the formula we use to bound the configurational entropy below. 

%
%
%
%

\section{Quantum low density parity check codes with linear confinement\label{suppl:expansion}}

The models that we consider in this work as an instantiation of topological quantum spin glasses are derived directly from quantum error correcting codes. More precisely, they are based on quantum low density parity check codes with a property called linear confinement, or, in its most extreme incarnation, expansion.

In this section, for readers unfamiliar with error correcting codes but also as a way to fix notation, we providing a brief introduction of classical and quantum low density parity check codes.
We then provide a formal definition of linear confinement in classical and quantum error correcting codes, and comment on its relation to the boundary- and coboundary expansion of chain complexes. 
We end this section with a review of two constructions of quantum error correcting codes with linear confinement.

\subsection{Classical and quantum low density parity check codes}

Classical linear codes on $n$ bits are defined via a so called ``parity check matrix'' $H\in \mathbb{F}_2^{m\times n}$, where $\mathbb{F}_2$ is the the field with two elements (for simplicity, we restrict the discussion here to binary codes defined on bits). Each row of $H$ defines a single ``parity check'', which corresponds to a subset of bits (corresponding to the nonzero entries of the row) whose sum is enforced to be zero mod $2$ in the codewords of the code.
its easy to check that this means that the codewords of the code are exactly given by the elements of the kernel of $H$. This will indeed be the formal definition:

\begin{definition}[Classical linear code]
    Given a parity check matrix $H\in \mathbb{F}_2^{m\times n}$ with $\rank H < n$, the classical linear code corresponding to $H$ is defined as the subspace $\mathcal C = \ker H$.\\
    We also define the \emph{number of logical bits} $k$, as the dimension of the code space, that is $k := \dim \mathcal C = n - \rank H$.\\
    We further define the \emph{distance} $d$ of the code as the smallest Hamming weight of any non-zero element of $\mathcal C$.
\end{definition}

By the above definition, codewords are vectors $\vec z \in \mathbb{F}_2^n$ such that $H \vec z = 0$. 
Flipping some bits amounts to adding to $\vec z$ another vector $\vec e\in \mathbb{F}_2^n$ (note that addition in $\mathbb{F}_2$ is mod $2$). 
Because of linearity, $\Hcheck(\vec z + \vec e) = \Hcheck \vec e =: \vec s$, and $\vec s$ is called the \emph{syndrome} of the error $\vec e$.
Because the codewords form a linear subspace, the definition of $d$ as the weight of the smallest codeword means that the distance is also the smallest difference between any two codewords. 
It is thus possible, at least in principle, to recover a corrupted codeword faithfully as long as $\abs{\vec e} < d/2$, by identifying the solution to $H \vec e = \vec s$ with minimal Hamming weight.
The triplet number of bits $n$, the number of logical bits $k$, and the distance $d$ are often denoted as a triplet and we say in short that $\mathcal C$ is an $[n, k, d]$ code.

For practical purposes we often want to enforce the condition that the parity check matrix is sparse.

\begin{definition}[Classical low density parity check (LDPC) code]\label{def:classical_ldpc}
    A family of classical linear codes define by a family parity check matrices $H_n$ is called a \emph{low density parity check (LDPC) code}, if there exist two $n$-independent integers $w_{\rm check}$ and $w_{\rm bit}$, called the check- and bit-degree respectively, such that $\sum_j(H_n)_{ij} \leq \wcheck$ and ${\sum_i(H_n)_{ij} \leq \wbit}$ for all~$n$.
\end{definition}

In other words, each parity check only contains a finite number of bits (at most $w_{\rm check}$) and each bit participates in a finite number of checks (at most $w_{\rm bit}$).

In the following, we will not explicitly denote the $n$-dependence of the parity check matrix $H$, but whenever we talk about LDPC codes implicitly consider an infinite family of codes with bit- and check- degrees independent of $n$.

Fixing some further notation, the ratio $r = k/n$, quantifies the overhead associated to the encoding and is called the \emph{rate} of the code. 
Intuitively, the best codes should maximize the rate (i.e.\ minimize the overhead) as well as the distance $d$ (i.e.\ the robustness against noise). Codes which have the optimal scaling of both parameters, that is $[n, \Theta(n), \Theta(n)]$-codes, are called \emph{good}.

Quantum error correcting codes correspond to a subspace $\mathcal C$ of the $2^n$-dimensional complex Hilbert space of $n$ qubits. 
Here, we will here consider so called Calderbank-Shor-Steane codes \cite{steane1996css, calderbank1996css}, where the subspace is defined as the common $+1$ eigenspace of a set of commuting Pauli operators. 
The Pauli operators are specified by a pair of classical parity check matrices, $\Hcheck_X\in\mathbb{F}_2^{m_X\times n}$ and $\Hcheck_Z\in\mathbb{F}_2^{m_Z\times n}$ and define two different kind of check operators called stabilizers, corresponding to products of Pauli-$X$ and Pauli-$Z$ operators, respectively:
\begin{flalign}
    &&
    S_{X, i} = \prod_{j; (\Hcheck_X)_{ij} = 1} \hat X_j, 
    &&
    S_{Z, i} = \prod_{j; (\Hcheck_Z)_{ij} = 1} \hat Z_j.
    &&
    \label{eq:checks_Hs}
\end{flalign}

The condition that all the $S_{X, i}$ and $S_{Z, i}$ commute then corresponds to the condition that $H_{X}\cdot H_{Z}^T = \vec 0^{m_X \times m_Z}$.
The codespace $\mathcal C$ is then defined as the common $+1$ eigenspace of all $g_{X, i}$ and $g_{Z, i}$, and has dimension $k = n - \rank \Hcheck_X - \rank \Hcheck_Z$.
As before, this leads us naturally to the formal definition

\begin{definition}[Calderbank-Shor-Steane (CSS) code]
    Consider a pair of classical parity check matrices ${\Hcheck_X\in\mathbb{F}_2^{m_X\times n}}$ and $\Hcheck_Z\in\mathbb{F}_2^{m_Z\times n}$ such that $H_{X}\cdot H_{Z}^T = \vec 0^{m_X \times m_Z}$.
    These define two sets of commuting check operators $S_{X, i}$, $S_{Z, i}$ (see~\cref{eq:checks_Hs}), and the corresponding quantum code is defined as the mutual $+1$ eigenstates of all check operators.\\
    We also define the number of logical qubits as $k = \dim \mathcal C = n - \rank \Hcheck_X - \rank\Hcheck_Z$.\\
    We further define the logical operators of the code as those Pauli-strings that commute with all checks. The distance of the code is then defined as the size of the smallest non-trivial logical operator, i.e. the smallest operator that commuted with all checks but is not a product of checks (and hence not the logical identity).
\end{definition}

We also define the LDPC condition for quantum codes

\begin{definition}[Quantum low density parity check (qLDPC) code]
    A family of CSS codes defined by a family of parity check matrices $\Hcheck_{X, n}$, $\Hcheck_{Z, n}$ is called a quantum low density parity check (qLDPC) code, if both $\Hcheck_{X, n}$ and $\Hcheck_{Z, n}$ are LDPC (see \cref{def:classical_ldpc}).
\end{definition}

The above definition are almost analogous to the classical case, apart from the distinction of the code states and logical operators. Since the parity check matrices in the CSS case are used to define operators on the Hilbert space, the codewords of the corresponding classical codes correspond to the logical operators of the quantum code. The condition $H_{X}\cdot H_{Z}^T = \vec 0^{m_X \times m_Z}$ means that these codewords fall into two classes: (1) check operators $g_{X/Z}$, which are representatives of the logical identity (recall that the code space is the mutual $+1$ eigenstate of all checks) and (2) logical operators that act non-trivial on the code space. The logical operators are again divided into two sets: logical $X$-operators are given by $\overline X = \hat X^{\overline{x}}$ where $ \overline{x}\in \ker \Hcheck_Z / \Im \Hcheck_X^T$ while logical $Z$-operators are given by $\overline Z = \hat Z^{\overline{z}}$ where $\overline{z} \in \ker \Hcheck_X / \Im \Hcheck_Z^T$. Here and below, we use the notation $\hat O^{\vec a} = \prod_j O_j^{a_j}$ and use $\Im M$ to denote the image of the matrix $M$.
The distance of the quantum code in terms of the parity check matrices is hence the smallest Hamming weight of any element of the set $(\ker \Hcheck_Z - \Im \Hcheck_X^T) \cup (\ker \Hcheck_X - \Im \Hcheck_Z^T)$.

As in the classical case, $r=k/n$ is called the rate of the code, and a code is called good if both $k$ and $d$ are proportional to $n$.

\subsection{Linear confinement and expansion}

We begin by defining linear confinement for classical codes. In coding theory, the same property is also sometimes called \emph{robustness}.

\begin{definition}[Linear confinement]
\label{def:expansion_classical}
    We say that a classical code with parity check matrix $H\in\mathbb{F}_2^{m\times n}$ has linear ($\delta$, $\gamma$)-confinement if 
\begin{equation}
	\abs{\vec x} \leq \delta(n) \Rightarrow \abs{H \vec x} \geq \gamma \abs{\vec x}
\end{equation}
for some monotonically increasing function $\delta(n)$ which diverges super-logarithmically with $n$ (i.e. $\log n/\delta(n)\to0$ as $n\to\infty$), and some $\gamma >0$.
\end{definition}

Intuitively, linear confinement implies that ``large errors have large syndromes'' --- that is until the error becomes \emph{too} large, which is quantified by the function $\delta(n)$.
Note that the reverse is not true : there may be large errors $\abs{\vec x} > \delta(n)$ with small syndromes, but which are far away from any logical operator. A guarantee of this second kind, ``small syndromes correspond to small errors'' is called \emph{soundness}, and an interesting topic in itself. In this paper however, we focus on the physical implications of confinement. 

\begin{remark}\label{rem:expander_code}
    In the case where the upper bound $\delta(n)$ scales optimally, that is $\delta(n) = \delta\cdot n$ for some $\delta > 0$, our definition of linear confinement coincides with that of (code) \emph{expansion} by Sipser and Spielman \cite{sipser_spielman1996}. Classical codes with this scaling are also called \emph{expander codes}.
\end{remark}

\begin{remark}
    The parameter $\gamma$, which sets the number of checks violated per error, is upper bounded by the bit-degree $d_{\rm bit}$. Families of expander codes that approach this optimal value are called \emph{lossless expanders}.
\end{remark}

For quantum codes CSS, there is an equivalent notion of ``large errors have large syndromes'', however we will have to adjust the measure by which to measure the size of the error. 
Recall that a quantum CSS code is defined by two parity check matrices $\Hcheck_X$ and $\Hcheck_Z$ that fulfill the condition $\Hcheck_X^T\Hcheck_Z = 0$, which guarantees that $X$- and $Z$-checks commute.
This means that when viewed a classical code, neither $H_X$ nor $H_Z$ can be expanding in the sense of \cref{def:expansion_classical} since they have many small codewords corresponding to the stabilizers of the other kind. The stabilizer however act trivially on the codespace, and hence should not count towards the size of the error. This is formalized by defining the ``reduced weight'' of an error.

\begin{definition}[Reduced weight]
\label{def:reduced_weight}
    Given a pair of parity check matrices $H_X\in \mathbb{F}_2^{m\times n}$ and $H_Z \in \mathbb{F}_2^{m\times n}$ with $H_X^T H_Z = 0$, we define the reduced weight with respect to $H_X$ and $H_Z$ as the following norms on $\mathbb{F}_2^n$
    \begin{align}
        \norm{\vec x}_{X} :=& \dist[\vec x, \Im(\Hcheck_X^T)], \\
        \norm{\vec z}_{Z} :=& \dist[\vec z, \Im(\Hcheck_Z^T)]. 
    \end{align}
    Above, the distance of a vector $\vec x \in \mathbb{F}_2^n$ to a subspace $A$ is defined as ${\rm dist(\vec x, A)} := \min_{\vec a \in A}\abs{\vec x + \vec a}$.
\end{definition}

\begin{remark}
   Given Pauli operators $E_X = \hat X^{\vec x}$ and $ E_Z = \hat Z^{\vec z}$, and a code state $\ket{\psi} \in \mathcal C$, then $\norm{\vec x}_X$ and $\norm{\vec z}_z$ are the weights of the smallest Pauli operators with the same action as $E_X$ and $E_Z$ on $\ket{\psi}$, respectively. 
\end{remark}

We are now ready to define linear boundary- and coboundary confinement.

\begin{definition}[Linear boundary and coboundary confinement]
\label{def:expansion}

Consider a quantum CSS code defined by two parity check matrices $H_{\rm X}$ and $H_{\rm Z}$. For $\delta(n)$ some monotonic function of $n$ that diverges super-logarithmically with $n$ (i.e. $\log n/\delta(n)\to0$ as $n\to\infty$), and $\gamma >0$, we say that the code has \emph{($\delta$, $\gamma$)-boundary confinement} if
\begin{equation}
	\norm{\vec x}_{\rm X} \leq \delta(n) \Rightarrow \abs{H_{\rm Z} \vec x} \geq \gamma \norm{\vec x}_{\rm X}.
\end{equation}
We say that the quantum code has \emph{($\delta$, $\gamma$)-coboundary confinement} if
\begin{equation}
	\norm{\vec z}_{\rm Z} \leq \delta(n) \Rightarrow \abs{H_{\rm X} \vec z} \geq \gamma \norm{\vec z}_{\rm Z}.
\end{equation}
We say that a quantum code has \emph{($\delta$, $\gamma$)-confinement} if it has both ($\delta$, $\gamma$)-boundary and ($\delta$, $\gamma$)-coboundary confinement.
\end{definition}

Although, for the ease of readers more familiar with the perspective from coding theory, we have formulated the above definitions in terms of classical and CSS codes, we still borrow the names linear boundary and coboundary confinement from homology.
This is because a classical code can be identified with a two-term chain complex
\begin{equation}
    C_1
    \xrightarrow[\partial_1 = H^T]{} C_0,
\end{equation}
where the basis elements of the $\mathbb{F}_2$-vector spaces $C_0$ and $C_1$ correspond to bits and checks, respectively.
Any CSS quantum code can be identified with a three-term chain complex
\begin{equation}
    C_2
    \xrightarrow[\partial_2 = H_{\rm Z}^T]{} C_1
    \xrightarrow[\partial_1 = H_{\rm X}]{} C_0,
\end{equation}
where the basis elements of the vector spaces $C_0$, $C_1$, and $C_2$ correspond to $X$-checks, qubits, and $Z$-checks, respectively, and the $Z$ and $X$-parity check matrices correspond to the boundary operators. In this language, errors correspond to elements of $C_1$. The $X$($Z$)-syndrome is given by the (co-)boundary of this chain, and the reduced weight is simply the distance of a chain to the set of (co-)boundaries. Linear (co-)boundary confinement with optimal scaling of $\delta(n) = \delta\,n$ for some $\delta >0$ is therefore equivalent to (co-)boundary expansion of the corresponding chain complex \cite{gromov2010singularities, linial2006homological, gotlib_kaufman_2023}.
See \cite[Section~II.B]{qldpc_review} for more details on this connection between (quantum) codes and homological algebra.

\subsection{Two examples of quantum codes with linear confinement}

In the main text, we mention two explicit examples of codes with linear confinement. Here, we provide their formal definition as well as a summary of their properties that are relevant for this work. 
For a more detailed overview, we refer the interested reader to Ref. \onlinecite{qldpc_review}.

\subsubsection{Hypergraph Products of Gallager Codes\label{suppl:hgp}}

Our first example will be constructed by taking a certain homological product of the historically first example of good LDPC codes, called Gallager codes. 
Gallager introduced the idea of low-density parity check codes in his PhD thesis \cite{gallager1960low,gallager1962low}. 
He also provided the first example of a family of \emph{good} LDPC codes, by considering random codes.

In particular, the (n, \wbit, \wcheck)-ensemble is defined by considering the following ensemble of $\Hcheck\in \mathbb{F}_2^{m\times n}$ matrices with $m = \tfrac{\wbit}{\wcheck}n$. Partition $\Hcheck$ vertically into $\wbit$ equal-size blocks of size $\tfrac{n}{\wcheck}\times n$. Each block will have one nonzero entry per column. The first block contains all its ones in descending order such that each row contains $\wcheck$ ones. The remaining blocks are obtained from the first by random permutations of the columns, where each permutations is considered with equal probability.

Note that for $\wcheck > \wbit$ codes from the ensemble have linear rate since $k = n - \rank \Hcheck \geq n - m = n(1-\tfrac{\wbit}{\wcheck})$.
In fact, it is known that this bound is tight:

\begin{lemma}[Gallager Codes have no redundancies, Lemma 3.27 in \cite{richardson2008modern}]\label{lem:gallager_no_redundancies}
    Consider the $(n, \wbit, \wcheck)$ ensemble of classical LDPC codes, and let
    \begin{equation}
        k_{\rm des} = 
        \begin{cases}
            n\,\left( 1 - \tfrac{\wbit}{\wcheck} \right) - 1 & \text{if}~\wbit~\text{even} \\
            n\,\left( 1 - \tfrac{\wbit}{\wcheck} \right) \phantom{-1} & \text{if}~\wbit~\text{odd} \\
        \end{cases}.
    \end{equation}
    
    Then, for $\wcheck > \wbit \geq 2$, we have for the actual number of logical bits $k := n-\rank H$
    \begin{equation}
        {\rm Prob}(k = k_{\rm des}) \xrightarrow[n\to\infty]{} 1
    \end{equation}
\end{lemma}

The above Lemma in particular means that a parity check matrix chosen from the LDPC ensemble has no redundancies with high probability for sufficiently large $n$.

Sipser and Spielman later showed that Gallager codes are expander codes (see \cref{rem:expander_code}) with high probability \cite{sipser_spielman1996} and one can further show that they achieve confinement with arbitrarily large coefficient $\gamma$, that is $\gamma > \wbit -2$ with high probability (see \cite[Theorem~8.7]{richardson2008modern} and \cite[Lemma~11.3.4]{guruswami2019essential}).

Tillich and Zemor \cite{tillich2009hgp} considered the so-called \emph{hypergraph product} of a Gallager code with itself, which the quantum code defined by the parity check matrices
\begin{subequations}\label{eq:hgp_def}
\begin{align}
    \Hcheck_X &= \left(H \otimes \id_n, \hspace{3mm} \id_m \otimes H^T \right) \\
    \Hcheck_Z &= \left(\id_n \otimes H, \hspace{3mm} H^T \otimes \id_m \right)
\end{align}
\end{subequations}
where $H$ is the parity check matrix of the classical Gallager code. They showed that this provides a quantum code with linear rate $k = \Theta(n),$ and $d = \Omega(\sqrt n)$ \cite{tillich2009hgp}. This construction was the first to achieve linear rate and super-logarithmic distance.

In the following, we refer to the hypergraph product of a Gallager codes with itself simply as ``hypergraph product code'', omitting the classical input code for brevity. For our purposes their most important property is the fact that they have linear confinement, which was first shown in Ref.~\onlinecite{leverrier2015quantum_expander}. Here, we quote a revised version from Ref.~\onlinecite{Fawzi2018constant}, which provides an improved lower bound on the value of the coefficient $\gamma$.

\begin{lemma}[Lemma 15 in \cite{Fawzi2018constant}]\label{lem:hgp_confinement}
    Let $\epsilon$ be a positive constant. For $\wbit, \wcheck > \epsilon^{-1}$, the hypergraph product of a classical Gallager code chosen from the $(n, \wbit, \wcheck)$ ensemble with itself has, with high probability, linear $(\delta, \gamma)$-confinement in the sense of \cref{def:expansion}, where $\delta(n) = \delta\cdot\sqrt n$ with $\delta = \wbit(\wbit^2 + \wcheck^2)^{-1}\epsilon$ and $\gamma = \tfrac{1}{2}(1 - 8\epsilon)\wbit$.
\end{lemma}

In particular, this means that while the scaling of the upper bound $\delta(n)$ is sub-optimal, hypergraph products of Gallager codes realize arbitrarily large coefficients $\gamma$.

\subsubsection{Balanced Products of Tanner codes on Symmetric Expanders}

The upper bound  on the distance $d \leq \Omega(\sqrt{n})$ of hypergraph product codes is due to the fact that code words of the classical input codes lift to logical operators of the quantum code, whereas the number of qubits $n$ of the quantum code scales quadratically in the number of bits of the input code.

This problem can be overcome using a different type of product of classical codes called \emph{balanced product}.
Assuming that we have a group $G$ acting on the parity check matrix $\Hcheck$ then, instead of just performing the usual tensor product, we can factor out the action of $G$ on both factors.
More generally, for vector spaces $V$ and $W$ with a group $G$ acting on both, we define the product
$$V\otimes_G W := V\otimes W/\langle v\cdot g\otimes w-v\otimes g\cdot w\mid v\in V,w\in W, g\in G \rangle.$$
We refer to \cite{breuckmann2021balanced,qldpc_review} regarding the details of the construction.
However, it should be clear that we can not use Gallager codes as inputs, as they are sampled from a random ensemble and will not have sufficiently many symmetries. 
We therefore require a highly symmetric construction of good LDPC codes that we can use as inputs.
Such a construction was developed by Sipser--Spielman \cite{sipser_spielman1996} who showed, based on earlier ideas by Tanner \cite{tanner1981recursive}, that one can obtain good codes not only from random graphs, but also from graphs which have certain expansion properties.

Let us describe how to obtain suitable input codes. Consider a group~$G$ with two symmetric generating sets $A,B \subset G$ (meaning $A = A^{-1}$ and $B = B^{-1}$) of equal cardinality $\Delta = |A| = |B|$. From these, we construct a right Cayley graph $\operatorname{Cay}^r(G,A)$ and a left Cayley graph $\operatorname{Cay}^\ell(G,B)$.
We then consider their respective double-covers, which we denote by $X_2^r = \operatorname{Cay}^r_2(G,A)$ and $X_2^\ell = \operatorname{Cay}^\ell_2(G,B)$. When these double-covers exhibit strong spectral expansion properties, we can construct Sipser-Spielman expander codes~\cite{sipser_spielman1996} using local codes~$K$ and~$L$. 
The resulting codes are denoted by $C(X_2^r,L)$ and $C(X_2^\ell,K)$.
We will require that the Cayley graphs are Ramanujan graphs, which means the spectrum of the adjacency matrix satisfies $\lambda = \max \{ |\lambda_2|, |\lambda_n| \} \leq 2 \sqrt{\Delta - 1}$. Such graphs can be constructed following~\cite{margulis1988expanders,lps1988expanders}.
These codes were constructed in \cite{breuckmann2021balanced} and conjecture to be \emph{good} qLDPC codes , i.e.\ with $k\sim n$ and $d\sim n$, which was later proven in~\cite{panteleev2022qldpc}. 
In~\cite{dinur2023qldpc} proved a slightly stronger statement, namely that balanced products of Sipser-Spielman codes are (co-)boundary expanders, which we will use here.

Consider balanced product code $$C(X_2^r,K) \otimes_G C(X_2^\ell,L) = C(X_2^r\times_G X_2^\ell, K\otimes L).$$ 
The right-hand side of this equation reflects that we can assign a local tensor code $K\otimes L$ directly to the balanced product complex $X_2^r\times_G X_2^\ell$, as detailed in \cite[Section~IV-B]{breuckmann2021balanced}.

The following theorem, which combines Theorem~3.8 and Corollary~3.10 from Ref.~\cite{dinur2023qldpc}, characterizes the expansion properties of the balanced product code.

\begin{theorem}\label{thm:CCexpansion}
    The parity-check matrices of the balanced product code $$C(X_2^r,L)\otimes_G C(X_2^\ell,K) = C(X_2^r\times_G X_2^\ell, K\otimes L)$$ have linear $(\delta,\gamma)$-confinement with constant $\delta >0$.
\end{theorem}

\section{Typical low-temperature states are surrounded by bottlenecks\label{suppl:ergo_breaking}}

In this section, we prove that the Gibbs state in Hamiltonians derived from quantum codes with linear confinement can generally be decomposed into components supported on subspaces which are surrounded by a bottleneck in the sense of \cref{thm:bottleneck}.
In particular, we show that we can define a subspace $\mathcal V$ around any \emph{typical} (with respect to the Gibbs distribution) eigenstate of the Hamiltonian which is surrounded by a bottleneck subspace $\partial \mathcal V$.

\subsection{Setup of the proof and notation\label{sec:suppl_ergo_breaking_setup}}

Consider an eigenstate $\ket{\vec x_0,\vec z_0}$ of a CSS stabilizer code Hamiltonian with linear confinement, labeled by the supports of the Pauli operators used to produce it from one of the ground states (that is $\ket{\vec x_0,\vec z_0} = \hat X^{\vec x_0} \hat Z^{\vec z_0} \ket{\psi}$ for an arbitrary reference ground state $\ket{\psi}$, see also main text). The idea is to choose the subspace $\mathcal V$ in a way such that in includes typical thermal fluctuations around $\ket{\vec x_0, \vec z_0}$. That is, we are interested in eigenstates $\ket{\vec x_0 + \vec x, \vec z_0 + \vec z}$, where the ``fluctuations'' $\vec x, \vec z$ are in some sense typical according to the Gibbs distribution. 

In graph-local models (i.e.\ models with bounded-degree interaction graphs), percolation theory tells us that typical fluctuations at any finite temperature have extensive operator weight, but at low temperature their support will have no large connected components. 
Here, we will demand an even stronger condition, that is that fluctuation should have no large connected components, and that the connected components that are present should be sufficiently far apart. This is captured by the notion of an $\alpha$-subset:

\begin{figure}
    \centering
    \includegraphics{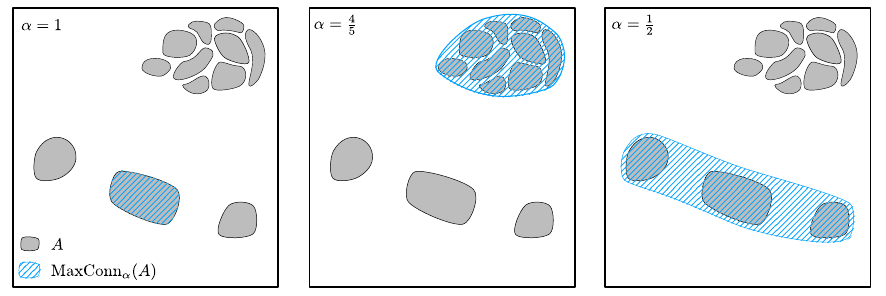}
    \caption{The maximal connected $\alpha$-subset, $\maxConn_{\alpha}(A)$ of a set $A$ (gray) indicated or decreasing values of $\alpha$ in hatched blue. 
    For $\alpha=1$, the notion of an $\alpha$-subset is identical with that of a usual subset, and the maximal connected $\alpha$ subset is just the largest connected component of $A$.
    As $\alpha$ is decreasing, we have increasingly more freedom to ``join'' connected components of $A$ through empty space, while maintaining a large intersection.}
    \label{fig:alpha_subset}
\end{figure}

\begin{definition}[$\alpha$-subset]
	Given a graph $G = (V, E)$, a subset of vertices $A\subset V$ and $\alpha > 0$, we call $B\subset V$ an $\alpha$-subset of~$A$ if $\abs{A\cap B} \geq \alpha \abs{B}$.
    \label{def:alpha_subset}
\end{definition}

\begin{remark}
    For $\alpha=1$ the definition coincides with that of a subset $B \subset A$.
\end{remark}

Of particular importance below will be the idea of the largest connected $\alpha$-subset, $\maxConn_{\alpha}(A)$. 
We illustrate this idea in \cref{fig:alpha_subset}.
While for $\alpha=1$, $\maxConn_{\alpha}(A)$ is just the largest connected component of $A$ (left panel), for $\alpha<1$ it is the largest connected superset of $A$ with relative overlap at least $\alpha$. 
This can be seen as giving freedom to ``join'' different connected components, as illustrated in the center and right panel. 
This means in particular that any set \emph{without} any large $\alpha$-subset must consist of connected components that are small \emph{and} far apart from each other.
We will make this last idea quantitative in \cref{lem:maxConn_alpha_maxConn} below, but for now proceed by using it to define our bottleneck set.

To apply the notion of $\alpha$-subset to our problem, we will abuse notation slightly and identify the binary vectors $\vec x$ with their support, which is also the support of the Pauli operators $\hat{X}^{\vec x}$ labeled by $\vec x$. The support of $\vec x$ is therefore a set of qubits and we can then talk about the connected components of $\vec x$, where connectivity is defined by the interaction graph of the code Hamiltonian, i.e. two qubits are connected if they are part of the same check. One reason why connected components of $\vec x$ (and $\vec z$) are important is that, by definition, different connected components have independent syndromes that cannot cancel each other. This means, among other things, that the expansion properties in \cref{def:expansion} apply to each connected component separately. 

With this preparation, we define our subspace $\mathcal V(\vec x_0, \vec z_0)$ as
\begin{subequations}\label{eq:gibbs_subspaces}
\begin{align}
    \mathcal V(\vec x_0, \vec z_0) :=& \linspan \left\{ 
        X^{\vec x_0 + \vec x} Z^{\vec z_0 + \vec z} \ket{\psi} ~\vert~
        \vec x_0 + \vec x \in \Omega(\vec x_0)
        ~\text{and}~
        \vec z_0 + \vec z \in \Omega(\vec z_0)
    \right\}, \\
    \Omega(\vec x_0) :=& \left\{ 
        \vec x_0 + \vec x ~\vert~ \maxConn_{\alpha = 1/2}(\vec x) \leq \tfrac{1}{2}\delta(n)
    \right\},
\end{align}
\end{subequations}
where $\delta(n)$ is the upper bound on the size of confined errors in the code (see \cref{def:expansion}), and $\maxConn_{\alpha}(\vec x)$ denotes the size of the largest connected $\alpha$-subset of the support of $\vec x$. We define the subspace in terms of a set of of classical binary vectors $\Omega$, since, as we shall see below, to establish the desired bottleneck condition for the space $\mathcal V$ and its boundary $\partial \mathcal V$ it will be sufficient to show an equivalent condition on the subset $\Omega$ and its boundary, with respect to a suitably defined \emph{classical} Gibbs measure.

We can now consider the neighborhood of $\mathcal{V}(\vec x_0,\vec z_0)$ as defined in \cref{def:hilbert_neighborhood}. In particular, we consider a boundary of width $2r = \tfrac{1}{4}\delta(n)$ (note that it is this $2r$-boundary that appears in \cref{thm:bottleneck}), which is defined by
\begin{subequations}\label{eq:gibbs_subspaces_boundary}
\begin{align}
    \partial \mathcal V(\vec x_0, \vec z_0) :=& \linspan\left\{ 
        X^{\vec x_0 + \vec x} Z^{\vec z_0 + \vec z} \ket{\psi} ~\vert~
        \vec x_0 + \vec x \in \partial\Omega(\vec x_0)
        ~\text{and}~
        \vec z_0 + \vec z \in \partial\Omega(\vec z_0)
    \right\}\\
    \partial \Omega(\vec x_0) :=& \left\{
        \vec x \notin \Omega(\vec x_0) \mid {\rm dist}(\vec x, \Omega(\vec x_0)) <\tfrac{1}{4} \delta(n)
    \right\}
\end{align}
\end{subequations}

In the following, we will sometimes drop the arguments and refer to the two subspaces simply as $\mathcal{V}$ and $\partial\mathcal{V}$.  Note that from the way we defined these subspaces, they are well-separated from their image under any logical operation; that is, if we define the completion $\overline{\mathcal V} := \mathcal V \cup \partial \mathcal V$ we have $\overline{\mathcal V} \cap \hat L\overline{\mathcal V} = \emptyset$ for all logical operators $\hat L$ of the code. 
This follows directly from linear confinement, which implies $\maxConn(\hat L) > \delta(n)$, and \cref{lem:maxConn_alpha_maxConn} below. 

Let us give some intuition for the definitions \autoref{eq:gibbs_subspaces} and \autoref{eq:gibbs_subspaces_boundary} above. As already mentioned, the central idea that given a state $\ket{\vec x_0, \vec z_0}$, the space $\mathcal V(\vec x_0, \vec z_0)$ should already include typical thermal fluctuations $\vec x$, $\vec z$ around it, such that the bottleneck $\partial V(\vec x_0, \vec z_0)$ only has small weight relative to $\mathcal V$.
As will show below (\cref{lem:maxConn_alpha_maxConn}), the way we defined $\mathcal V$ and its boundary $\partial \mathcal V$ guarantees that no fluctuation in either $\mathcal V$ or $\partial \mathcal V$ has a large connected component. Since each connected component has an independent syndrome, this will allow us to use linear confinement to argue that all fluctuations in $\partial \mathcal V$ have high energy cost relative to those in $\mathcal V$. In particular, each ``large'' connected $\alpha$-component (those not allowed in $\mathcal{V}$ but in $\partial \mathcal V$) comes with an energy that is proportional to its size. Since the entropic contribution to the free energy is also at most extensive in the size, this will allow us to prove that the bottleneck ratio defined in \cref{thm:botteleneck2} is exponentially small in $n$ at low temperature.

The global Gibbs state is given by
\begin{equation}
    \rhoG = Z^{-1} \sum_{\vec x, \vec z} e^{-\beta (\abs{H_Z \vec x} + \abs{H_X \vec z})} \ketbra{\vec x, \vec z}{\vec x, \vec z}
\end{equation}
where the sum goes over all binary vectors $\vec x$, $\vec z$, and we have chosen an arbitrary ground state as a reference. Note that in this notation, two states $\ket{\vec x, \vec z}$ and $\ket{\vec x', \vec z'}$ are identical iff $\vec x + \vec x' \in \Im(H_{X}^T)$ and $\vec z + \vec z' \in \Im(H_{Z}^T)$. This means that in the sum above, every eigenstate appears $2^{m}$ times. However, since this is an overall factor, it cancels out with the normalization $Z^{-1}$.

We can now write the bottleneck ratio defined in \cref{thm:bottleneck} as 
\begin{subequations}\label{eq:bottleneck_ratio_classical}
\begin{align}
    \Delta(n) := \frac{\tr P_{\partial \mathcal V}\rhoG}{\tr P_{\mathcal V}\rhoG}
    &= \frac{
    \sum_{\vec x \in \partial\Omega(\vec x_0), \vec z \in \partial\Omega(\vec z_0)} 
    e^{-\beta (\abs{H_Z \vec x} + \abs{H_X \vec z})}
    }{
    \sum_{\vec x \in \Omega(\vec x_0), \vec z \in \Omega(\vec z_0)}
    e^{-\beta (\abs{H_Z \vec x} + \abs{H_X \vec z})}
    } \\
    &= \frac{
    \sum_{\vec x \in \partial\Omega(\vec x_0)} 
    e^{-\beta \abs{H_Z \vec x}}
    }{
    \sum_{\vec x \in \Omega(\vec x_0)} 
    e^{-\beta \abs{H_Z \vec x}}
    }
    \times
    \frac{
    \sum_{\vec z \in \partial\Omega(\vec z_0)} 
    e^{-\beta \abs{H_X \vec z}}
    }{
    \sum_{\vec z \in \Omega(\vec z_0)} 
    e^{-\beta \abs{H_X \vec z}}
    } \\
    &= \frac{\muG^{(X)}[\partial\Omega(\vec x_0)]}{\muG^{(X)}[\Omega(\vec x_0)]}
    \times
    \frac{\muG^{(Z)}[\partial\Omega(\vec z_0)]}{\muG^{(Z)}[\Omega(\vec z_0)]},
\end{align}
\end{subequations}
where 
\begin{equation}
    \muG^{(X/Z)}(\vec y) \propto e^{-\beta \abs{H_{Z/X} \vec y}}
\end{equation}
is the global Gibbs measure of the \emph{classical} model defined by $\Hcheck_{X/Z}$.
It hence suffices to show a bottleneck condition for the set of classical configurations $\Omega$ and its boundary $\partial \Omega$, for both  $\muG^{(X)}$ and $\muG^{(Z)}$.
In the classical model, configurations related by adding checks are distinct, but have the same energy. To show the bottleneck, we will have to use that the classical energy functional $E(\vec y) = \abs{\Hcheck_{X/Z} \vec y}$ has linear (co-)boundary confinement (\cref{def:expansion}), which is defined with respect to the the reduced weight (\cref{def:reduced_weight}).
We now introduce the notation $\vec x\reduced$ to denote the smallest representative of the equivalence class $[\vec x] \in \mathbb{F}_2^n / \Im(H_{X}^T)$. In this notation, the reduced weight is simply $\norm{\vec x} = \abs{\vec x\reduced}$ where $\abs{\bullet}$ is the Hamming weight.
The sets $\Omega$ and $\partial \Omega$ then have to be defined with respect to the reduced representative $\vec x\reduced$. That is, we (re-)define $\maxConn_{\alpha}(\vec x) \to \maxConn_{\alpha}(\vec x\reduced)$, and whenever we say that a vector $\vec x$ is \emph{connected} we mean that the support $\vec x\reduced$ is connected in the sense defined above.

In the following, we also will drop the superscript in $\muG$ and the subscript in the parity check matrix $H$, since the proof for both the $X$ and the $Z$ case proceeds completely in parallel.

\subsection{$\alpha$-percolation}

To bound the bottleneck ratio, we will use the notion of $\alpha$-percolation. To the best of our knowledge, this was first considered in Bombin's proof of single-shot error correction for subsystem codes \cite[Lemma~10]{bombin2015single_shot}, but under a different name. 
\cref{thm:alpha_percolation} below in this form is adapted from Ref.\ \cite[Theorem~17]{Fawzi2017efficient}, which proved single-shot error correction for hypergraph product codes of classical expander codes.

Percolation theory asks about the size of the largest connected component in a randomly chosen subgraph of a graph $G$. We will here focus on \emph{site percolation}, where we choose  the subgraph as induced by a random subset of vertices. 
We call a subset of vertices connected if its induced subgraph is connected. 
Site percolation hence studies the size of the largest connected subset of a randomly chosen set of vertices of $G$
The idea of $\alpha$-percolation is to study instead the size of the largest connected $\alpha$-subset of a randomly chosen set of vertices of $G$.
In the following, we denote by $\maxConn(A)$ the size of the largest connected subset, and by $\maxConn_\alpha(A)$ the size of the largest connected $\alpha$-subset of $A \subset V$.

\begin{theorem}[$\alpha$-percolation, Theorem 17 in \cite{Fawzi2017efficient}, i.i.d.\ version]

	Let $G = (V, E)$ be a graph with degree upper bounded by $w$, and let $\alpha \in (0, 1]$ and $t\geq 1$ an integer.
	Then, for a random subset $\Sigma_0 \subset V$, with vertices chosen independently with probability $p\leq \frac{\alpha}{d-1}$ we have
	\begin{equation}
		{\rm Prob}\left[ \maxConn_{\alpha}(\Sigma_0) \geq t \right] \leq \abs{V} \left(\frac{w-1}{w-2}\right)^2 \frac{q^t}{1-q}
	\end{equation} 
	where $q = (1-p)^{w-1-\alpha}\, p^\alpha \,2^{h(\alpha)}\Phi$ with $\Phi=(w-1)(1+\tfrac{1}{w-2})^{w-2}$ and $h(\alpha)$ is the binary entropy function.
\label{thm:alpha_percolation}
\end{theorem}

\begin{remark}
    For any finite graph degree $w$, and any value of $0 < \alpha \leq 1$,  there exists a threshold $p_c > 0$, such that $q < 1$ for $p < p_c$.
\end{remark}

\begin{remark}
    For $q < 1$ (i.e.\ $p < p_c$, see previous remark), the probability of the random subset $\Sigma_0$ having a connected $\alpha$ subset of size larger than $t$ decays exponentially in $t$. In this case, the typical size of the largest connected $\alpha$-subset of $\Sigma_0$ hence scales as $\order{\log \abs{V}}$
\end{remark}

Note that we state here the version of the theorem where the subgraph is chosen at each vertex independently, more generally we only need to demand a local stochasticity condition, that is $\mu(\Sigma_0) \leq p^{\abs{\Sigma_0}}$ for some $0 < p < 1$. 
Below, we will use $\alpha$-percolation on the interaction graph of the Hamiltonian $\hamil$ to guarantee that at low temperature excitations are distributed sparsely within the graph. The i.i.d.\ condition in this case corresponds to the condition that the parity check matrices $\Hcheck_X$ and $\Hcheck_Z$ have no redundant (linearly dependent) rows, i.e. $\rank \Hcheck_X = m_X$ and $\rank \Hcheck_Z = m_Z$.
Local stochasticity however is still fulfilled in the presence of redundancies and hence the results presented also hold in this case. This is because the presence of redundancies simply implies that certain syndromes are forbidden (e.g. in the 2D Ising model, only closed loops of violated bonds are ``allowed''), but we still have $\mu(\Sigma_0) \leq p^{\abs{\Sigma_0}}$ simply because the energy is, by definition, proportional to the size of the syndrome.

Note that throughout this section, we use $\alpha$-subsets on two distinct graphs. In the definitions of the subspaces $\mathcal V$ in \autoref{eq:gibbs_subspaces}, we constrain the size of fluctuations to have no large $\alpha$-subsets, and fluctuations are subsets of qubits that are connected if they share a check. In the following, we use $\alpha$-percolation to ensure that the syndrome $\vec\Sigma_0$ of typical states at low temperature has no large $\alpha$-subset. The syndrome, however, is a set of checks, and a set of checks is connected when they are of the same type $X$ or $Z$, and they share a bit.

Below, we show a few simple properties of sets with bounded $\maxConn_\alpha$. These will prove to be useful in upper bounding the bottleneck ratio.

We first make more precise the intuition, already given above in \cref{sec:suppl_ergo_breaking_setup} and in particular in \cref{fig:alpha_subset}, that demanding the absence of large connected $\alpha$-subsets in a set means that connected components of the set have to be small \emph{and} far apart. To this end, we prove the following

\begin{lemma}
    \label{lem:maxConn_alpha_maxConn}
    Let $G = (G, E)$ be a graph, $\alpha \in (0, 1]$ and $L$ an integer. 
    Also let $A\subset V$ be a subset of vertices with $\maxConn_{\alpha}(A) \leq L$, and $B \subset V$ a  subset of vertices with $\abs{B} \leq L(1-\alpha)$. Then
    \begin{equation}
       \maxConn(A \cup B) \leq L
    \end{equation}
\end{lemma}

\begin{proof}\label{lem:maxconn_alpha1}
    We argue by contradiction. 
    Assume $\maxConn(A \cup B) > L$, and denote the largest connected component of $A\cup B$ by $S$. Note that by removing $B$ from $A\cup B$ and $S$, we only reduce the weight of $S$ by at most $\abs{B} \leq (1-\alpha)L$. Hence $\abs{A \cap S} 
        \geq \abs{S} - \abs{B} 
        \geq \abs{S} - L(1-\alpha)
        \geq \alpha\abs{S}$.
    However, since $\abs{S} > L$, this is in contradiction with the assumption that $\maxConn_\alpha(A) \leq L$.
\end{proof}

We will also need the following lemma, which, informally stated, tells us that the maximal connected $\alpha$ subset of a set $A$ is a collection of independent connected components of $A$, and hence separated from the rest of $A$.
\begin{lemma}\label{lem:maxconn_alpha2}
     Let $G = (G, E)$ be a graph, $\alpha = (0, 1]$, and $A\subset V$ a subset of vertices. Denote by $S$ the maximal connected $\alpha$-subset of $A$, and by $S^c = V / S$ its complement. Then no edge in $E$ connects $S\cap A$ and $S^{c}\cap A$.
\end{lemma}

\begin{proof}
    We argue by contradiction. Assume that there exists $(u, v)\in E$ such that $u \in S\cap A$ and $v\in S^{c}\cap A$. Then consider $S' = S \cup \{v\}$, which is connected, $\abs{S'} = \abs{S} + 1 > \abs{S}$, and
    \begin{equation}
        \abs{A \cap S'} = \abs{A\cap S} + 1 
        \geq \alpha \abs{S} + 1 
        \geq \alpha(\abs{S} + 1)
        \geq \alpha \abs{S'}.
    \end{equation}
    This is contradiction with the fact that $S$ is the \emph{maximal} $\alpha$ subset of $A$.
\end{proof}

\subsection{Typical states are surrounded by energy barriers}

In this section, we characterize the landscape around typical low temperature eigenstates of qLDPC codes with linear confinement This will provide some intuition that is helpful for the proof of the bottleneck condition in \cref{thm:botteleneck2}. More generally, the characterization of the energy landscape is also of independent interest in itself.

Below, we discuss properties at the level of a single parity check matrix $H\in\mathbb{F}_2^{m\times n}$, rather than the full qLDPC code; nevertheless, we do assume that the check matrix fulfills the linear confinement property with respect to the reduced weight $\norm{\bullet}$ defined in \cref{def:reduced_weight}. Formally, this means we are considering one side of a three-term chain complex with linear (co-)boundary confinement as defined in \cref{def:expansion}. Also note that whenever below we talk about connected (sub)-sets of bits or checks, we define the connectivity with respect to the bit- and check-graph respectively. Two bits are connected when they share a check, and two checks are connected when they share a bit.

We start our characterization of the landscape by noting that low-temperature states with high probability have an ``effective linear confinement'' property for fluctuations that are both sufficiently large and sufficiently small.

\begin{proposition}[Barriers around typical low-temperature states]
\label{prop:x0_expansion}
    Consider $H\in \mathbb{F}_2^{m \times n}$ of full rank, LDPC, and with linear $(\delta, \gamma)$-confinement.
    Let $0 < \eta(n) < \delta(n)$ be a function that diverges with $n$, and let $\muG(\vec y) = Z^{-1}\exp(-\beta \abs{H \vec y})$ be the Gibbs distribution of $H$. 
    Then, $\forall \epsilon > 0$, $\exists \beta^*, c > 0$, such that $\forall \beta > \beta^*$, for a state $\vec x_0$ chosen randomly from $\muG$,  with probability
    \begin{equation}
        p_{\rm conf} = 1 - \order{n\, e^{-c \eta(n)}}
    \end{equation}
   we have that for all perturbations $\vec x$ such $\vec x\reduced$ is connected and $\eta(n) \leq \norm{\vec x} \leq \delta(n)$
    \begin{equation}
        E(\vec x_0 + \vec x) - E(\vec x_0) \geq (\gamma - \epsilon) \norm{\vec x}.
    \end{equation}
\end{proposition}

\begin{figure}
    \centering
    \includegraphics{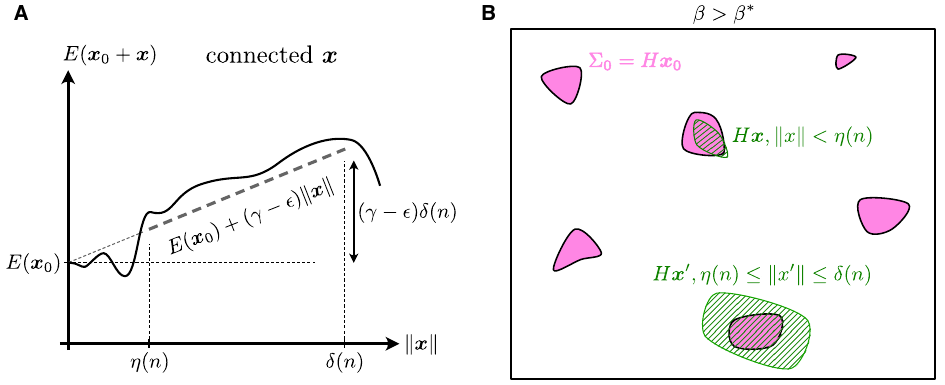}
    \caption{Illustration of the ``modified confinement'' around typical low-temperature states. (a) Around typical low energy states, \cref{prop:x0_expansion} still guarantees a barrier for \emph{connected} perturbations, but only if their size is in some given range $[\eta(n),\delta(n)]$. 
    (b) For typical $\vec x_0$, picked randomly from a low-temperature Gibbs distribution, the syndrome $\vec \Sigma_0 = H \vec x_0$ (drawn in pink) consists of small disconnected components. Small fluctuations (e.g. $\norm{\vec x} < \eta(n)$) may cancel out some of this syndrome, effectively lowering the energy. However, if the perturbation grows into a large connected set, it gets larger than the typical components of $\vec \Sigma_0$ and has to eventually increase the energy.}
    \label{fig:x0_expansion}
\end{figure}

Before proving this proposition, let us unpack the statement, which is also illustrated in \cref{fig:x0_expansion}. The idea is to draw a reference state $\vec x_0$ from the Gibbs measure at low, but finite temperature. The statement of the above proposition is that below a critical temperature, with probability going to one in the thermodynamic limit, this state fulfills a modified linear confinement property around it. In particular, unlike \cref{def:expansion} which applies to any error with size below $\delta(n)$, this modified confinement around $\vec x_0$ only applies to a particular class of relative errors $\vec x$, namely those that are (a) connected and (b) larger than above some $n$-dependent threshold $\eta(n)$. This threshold can be chosen arbitrarily, as long as it diverges super-logarithmically with $n$; the larger the threshold, the closer the probability $p_{\rm conf}$ will be to $1$.  

In particular, we note the following
\begin{remark}
    \cref{prop:x0_expansion} implies that $\forall \epsilon >0$, $\exists\beta^* >0$ such that for $\beta > \beta^*$ a typical state $\vec x_0$ is separated from $\vec x_0 + \vec z$, where $\vec z$ is a logical operator, by a barrier of size larger than $(\gamma-\epsilon)\delta(n)$.
\end{remark}
This is because just as expansion around codewords, the result of \cref{prop:x0_expansion} applies to each connected component of the perturbation $\vec x\reduced$ separately, and also due to expansion, logical operators have to have at least one connected component larger than $2\delta(n)$.

The resulting landscape around a typical low-temperature state is hence comprised of a ``plateau'' of super-logarithmic size, which is surrounded by subextensive barriers. Bounding the size of this plateau more quantitatively will be the topic of \cref{suppl:Sconf}.

The intuition behind this proposition is illustrated on the right of \cref{fig:x0_expansion}. While $\vec x_0$ violates a finite fraction of the checks, at low temperature its syndrome $\vec\Sigma_0 = H\vec x_0$ is sparse, i.e., it has no large connected component, or connected $\alpha$-subset, as indicated by \cref{thm:alpha_percolation} above. While a sufficiently small perturbation ($\vec x$ in the figure) can simply un-trigger some of the violated checks --- which decreases, rather than increases the energy --- a large connected perturbation eventually grows larger than the size of the largest connected component of $\vec \Sigma_0$ and has to increase the energy ($\vec x'$ in the figure).

\begin{proof}[Proof of \cref{prop:x0_expansion}]
    We assume w.l.o.g. that $\vec x=\vec x\reduced$ and hence $\vec x$ is connected. Note that the syndrome of $\vec x$, which we denote by $\vec\Sigma = H\vec x$,  may not be connected. However, we can consider the neighborhood $\Gamma(\vec x)$, that is the set of all checks incident to at least one bit in $\vec x$, and this is a connected superset of $\vec\Sigma$.
    We will use that
    \begin{equation}
        \abs{\vec \Sigma} \leq \abs{\Gamma(\vec x)} \leq w_{\rm bit} \norm{\vec x}
        \label{eq:sigma_gamma_bound}
    \end{equation}
    where $w_{\rm bit}$ is (an upper bound on) the bit-degree of $H$.

    The state $\vec x_0$ is chosen randomly from the Gibbs distribution.
    This is equivalent to choosing a \emph{syndrome} $\vec \Sigma_0 = H \vec x_0$ and then inverting the linear equation to obtain $\vec x_0$. Since the matrix $H$ is full rank, the syndrome $\vec \Sigma_0$ can be sampled by choosing its components as i.i.d random variables 
    \begin{equation}\label{eq:syndrome_iid}
        \Sigma_{0, j} = \begin{cases}
            1 &\text{with probability}~\frac{e^{-\beta}}{1 + e^{-\beta}} \\
            0 &\text{with probability}~\frac{1}{1+e^{-\beta}}
        \end{cases}.
    \end{equation}
    We can hence apply\cref{thm:alpha_percolation} on the check-graph (each check is a vertex and two vertices are connected if the corresponding checks share a bit) with $p = e^{-\beta} / (1 + e^{-\beta})$. Note that the degree of the check graph is bounded from above by $w_{\rm bit} w_{\rm check}$ .
    This implies in particular that for all $\zeta > 0$, $\exists \beta^*$ such that for $\beta > \beta^*$
    \begin{equation}
        \muG\left[\maxConn_{\zeta}(\vec \Sigma_0) \geq \eta \right] < ({\rm const})\times n \, e^{-c \eta}
        \label{eq:alpha_percolation_sigma0}
    \end{equation}
    where the constants $\beta^*$ and $c$ are fixed in \cref{thm:alpha_percolation}.

    Now note that in general
    \begin{align}
        E(\vec x_0 + \vec x) - E(\vec x_0)
            &= \abs{H(\vec x_0 + \vec x)} - \abs{H\vec x_0} \\
            &= \abs{H \vec x_0} + \abs{H\vec x} - 2 \abs{H\vec x_0 \wedge H\vec x} - \abs{H\vec x_0} \\
            &= \abs{H\vec x} - 2 \abs{\vec \Sigma_0 \wedge \vec \Sigma}
    \end{align}
    where $\wedge$ denotes bit-wise AND. 
    Using expansion of $H$ and $\cref{eq:sigma_gamma_bound}$ yields
    \begin{align}
        E(\vec x_0 + \vec x) - E(\vec x_0)
            &= \abs{H\vec x} - 2 \abs{\vec \Sigma_0 \wedge \vec \Sigma} \\
            &\geq \gamma \norm{\vec x} - 2 \abs{\vec\Sigma_0\wedge\vec\Gamma(\vec x)}.
    \end{align}
    Finally, we now use $\alpha$-percolation as set up in \cref{eq:alpha_percolation_sigma0} to conclude that with probability
    \begin{equation}
        P_{\rm exp} := \muG[\maxConn_{\zeta}(\vec\Sigma_0) < \eta] = 1 - \muG[\maxConn_{\zeta}(\vec\Sigma_0) \geq \eta] \geq 1 - \order{n\, e^{-\eta}}
    \end{equation}
    we have that
    \begin{align}
    E(\vec x_0 + \vec x) - E(\vec x)
            &\geq \gamma \norm{\vec x} - 2 \abs{\vec\Sigma_0\wedge\vec\Gamma(\vec x)} \\
            &\geq \gamma \norm{\vec x} - 2 \zeta \abs{\Gamma(\vec x)} \\
            &\geq (\gamma - 2 \zeta w_{\rm bit}) \norm{\vec x}
    \end{align}
    choosing $\zeta = \epsilon / (2 w_{\rm bit})$ then yields the desired inequality.
\end{proof}

\subsection{Typical states are surrounded by a bottleneck}\label{suppl:bottleneck}

We are now in a position to prove the desired bottleneck condition for the set $\Omega(\vec x_0)$ and its boundary $\partial \Omega(\vec x_0)$ defined in \cref{sec:suppl_ergo_breaking_setup}, for the case of a typical $\vec x_0$ drawn from the Gibbs distribution at low temperatures.

\begin{theorem}[Bottleneck around typical low-temperature states]
\label{thm:botteleneck2}
Consider $H\in \mathbb{F}_2^{m\times n}$ of full rank, LDPC, and with $(\delta, \gamma)$-expansion. Let $\vec x_0$ be a state chosen from the Gibbs distribution $\muG(\vec x) = Z^{-1} \exp(-\beta \abs{H \vec x})$ and define $\Omega(\vec x_0)$ and $\partial\Omega(\vec x_0)$ as in \autoref{eq:gibbs_subspaces} and \autoref{eq:gibbs_subspaces_boundary}. Then there exist $\beta^*, c_1 > 0$, such that for any $\beta > \beta^*$, it holds with probability
\begin{equation}
	p_{\rm bottleneck} = 1 - \order{e^{-c_1\,\delta(n)}}
	\label{eq:typical_state}
\end{equation}
that the bottleneck ratio obeys
\begin{equation}
	\frac{\muG[\partial\Omega (\vec x_0)]}{\muG[\Omega(\vec x_0)]} \leq e^{-c\delta(n)} \xrightarrow[n\to \infty]{} 0.
\end{equation}
\end{theorem}

Our proof strategy is similar to that used in Ref. \onlinecite{Hong2024quantum_memory}, to show self-correction in the ground states of hypergraph product codes, but our approach is different in two important ways.
First, we leverage the energy barriers around typical low-temperature states derived in the last section, to show a bottleneck condition around typical finite-temperature states, which in principle can be far from any ground state.
Second, we show a bottleneck condition for a boundary of diverging (with $n$) width. This allows us, using \cref{thm:bottleneck}, to bound the mixing time of \emph{any} local channel that has $\rhoG$ as a steady state.

\begin{proof}

The proof has two steps. As sketched in \cref{fig:y_to_z_mapping}, we first show that each state in $\partial \Omega$ can be related to a state in $\Omega$ by setting all bits on a set $S$ to zero in such a way that this decreases the energy by an amount proportional to $\abs{S}$.
In the second step, we bound the relative degeneracy of this mapping to show the desired bottleneck ratio.

\begin{figure}
    \centering
    \includegraphics{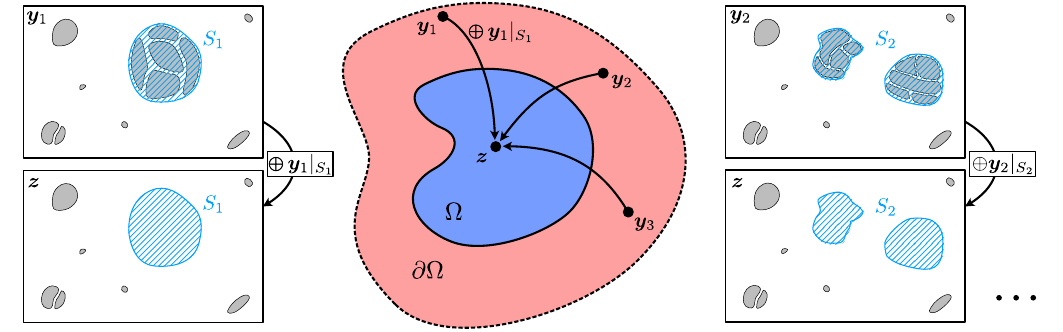}
    \caption{Mapping from states $\vec y$ in the bottleneck region, $\partial \Omega$, to states $\vec z(\vec y) = \vec y + \vec y\vert_S$ in $\Omega$, as used in the proof of \cref{thm:botteleneck2}.
    The mapping is generally many-to-one.
    The state $\vec z$ is obtained from $\vec y_1$, $\vec y_2$, $\dots$, by an iterated ``removal'' of the largest connected $\alpha=1/2$ subset, until no large such subset remains.
    We show that the state $\vec z (\vec y)$ has lower energy than $\vec y$ by an amount proportional to $\abs{S}$ [\cref{eq:deltaE-y}].
    }
    \label{fig:y_to_z_mapping}.
\end{figure}

\vspace{1em}
{\bf Bounding the relative energy of states in $\partial \Omega$ and $\Omega$}
\\[0.5em]
By construction, for $\vec x_0 + \vec y \in \partial\Omega$ we have $\maxConn_{\alpha=1/2}(\vec y\reduced) > \tfrac{1}{2}\delta(n)$.
Because of this, we can now map states $\vec x_0 + \vec y \in \partial \Omega$ to states $\vec x_0 + \vec z(\vec y) \in \Omega$ with lower energy.
Denote by $S$ the largest connected $\alpha=1/2$-subset of $\vec y\reduced$, and set all bits in $S$ to zero to obtain $\vec y'$, that is $\vec y' := \vec y + \vec y\vert_{S}$, where $\vec y\vert_{S}$ is the restriction of the vector $\vec y$ onto $S$, and addition is mod 2. Note that by assumption $\abs{S} > \tfrac{1}{2}\delta(n)$

Denoting the syndrome of $\vec x_0$ by $\vec\Sigma_0 := H \vec x_0$, we further know by \cref{thm:alpha_percolation} that for all $\epsilon > 0$, there exists $\beta^*$, such that for $\beta > \beta^*$ with high probability 
\begin{equation}\label{eq:alpha-epsilon}
    \maxConn_{\zeta=\epsilon/4 w_{\rm bit}}(\vec \Sigma_0) < \tfrac{1}{2}\delta(n).
\end{equation}
which implies in particular that
\begin{equation}
    \zeta\abs{\Gamma(S)} = \frac{\epsilon}{4 w_{\rm bit}} \abs{\Gamma(S)}\geq \abs{\vec \Sigma_0 \wedge \vec\Gamma(S)}
    \label{eq:zeta_percolation_bottleneck}
\end{equation}
where $w_{\rm bit}$ is (an upper bound on) the bit-degree, $\Gamma(S)$ is the neighborhood of $S$, i.e. the set of checks incident on at least one bit in $S$, and by construction $\abs{\Gamma(S)} \geq \abs{S} > \tfrac{1}{2}\delta(n)$. In an abuse of notation, we here and below sometimes identify sets of qubits with vectors in $\mathbb{F}_2^n$.

Together, this allows us to lower bound the energy gained by setting all bits on $S$ to zero.
First, note that by \cref{lem:maxconn_alpha2} the syndromes of $\vec y\vert_S$ and $\vec y\vert_{S^c}$ are independent, and thus
\begin{align}\label{eq:deltaE-y}
    E(\vec x_0 + \vec y) - E(\vec x_0 + \vec y')
        &= \abs{H\vec y} - \abs{H\vec y'} - 2 \left(\abs{\vec \Sigma_0 \wedge H \vec y} - \abs{\vec \Sigma_0 \wedge H \vec y'} \right)\\
        &= \abs{H(\vec y\vert_S)} - 2 \abs{\vec \Sigma_0 \wedge H(\vec y\vert_S)} \\
        &\geq \gamma \norm{(\vec y \vert_S)} - 
        2 \abs{\vec \Sigma_0 \wedge H(\vec y\vert_S)},
\end{align}
where in the last step we used that, by \cref{lem:maxConn_alpha_maxConn}, $\maxConn(\vec y\vert_S\reduced) < \delta(n)$ and hence we can use expansion to lower bound the energy.
Now, we are finally able to conclude that 
\begin{align}
    E(\vec x_0 + \vec y) - E(\vec x_0 + \vec y')
        &\geq \gamma \norm{(\vec y \vert_S)} - 
        2 \abs{\vec \Sigma_0 \wedge H(\vec y\vert_S)} \\
        &\geq \frac{\gamma}{2} \abs{S} - 2 \abs{\vec \Sigma_0 \wedge \Gamma(S)} \\
        &\geq \frac{\gamma}{2} \abs{S} - \frac{\epsilon}{2 d_{\rm bit}}\abs{\Gamma(S)}\\
    &\geq \tfrac{1}{2}(\gamma - \epsilon) \abs{S}.
\end{align}
going from the first to the second line uses that $S$ is an $\alpha=1/2$ subset of $\vec y\reduced$, as well as the fact that $\vec y\vert_S$ cannot trigger more than all checks adjacent to $S$.
The third line uses \cref{eq:zeta_percolation_bottleneck}, and the final step uses that $\abs{\Gamma(S)} \leq d_{\rm bit} \abs{S}$.

We have shown that for any state $\vec x_0 + \vec y \in \partial \Omega$, we can ``remove'' all flips of $\vec y$ in its largest connected $\alpha=1/2$-subset  (i.e. set $\vec y|_S = 0$) and lower the energy by at least a finite fraction of $\abs{S}$.

We can now iterate this procedure. Let us relabel $S \to S_1$. If $\vec x_0 + \vec y' \in \Omega$, we stop. Otherwise, we remove the largest connected $\alpha=1/2$-subset of $\vec y'$, denoted by $S_2$, to get a new vector $\vec y''$. We can repeat this procedure until we arrive at a vector $\vec x_0 + \vec z(\vec y) \in\Omega$, where $\vec z(\vec y)$ is the vector obtained from $\vec y$ by removing all the connected $\alpha=1/2$-subsets $S_i$ that are bigger than $\tfrac{1}{2}\delta(n)$.
Since in every step $i$, we lower the energy by a finite fraction $\tfrac{\gamma-\epsilon}{2}$ of $\abs{S_i}$, the total energy difference fulfills
 \begin{align}
 \label{eq:deltaE_reduction}
 	E(\vec x_0 + \vec y) - E(\vec x_0 + \vec z(\vec y)) 
 		&\geq \tfrac{1}{2}(\gamma - \epsilon)\abs{S}
 \end{align}
 where $S = \uplus_i S_i$ is the (disjoint) union of all $S_i$. 

This concludes the first part of the proof. Note that we achieved a \emph{linear} energy barrier in the sense of showing that each state in $\partial\Omega$ has a state in $\Omega$ from which it differs on at most $\abs{S}$ sites, while its energy is higher by an amount proportional to $\abs{S}$.
Since the entropy, that is the number of states in $\partial \Omega$ that map to the same $\Omega$, can at most be linear in $\abs{S}$ as well, we can at this point already be sure that $\partial \Omega$ will indeed be a bottleneck at low temperature. 
We derive a concrete upper bound on the relative degeneracy in the following, before finally bounding the bottleneck ratio.

\vspace{1em}
{\bf Bounding the relative entropy of states in $\partial \Omega$ and $\Omega$}
\\[0.5em]
Naturally, the above mapping $\vec z(\vec y)$ is many-to-one: many $\vec y$ are mapped to the same $\vec z$ (see \autoref{fig:y_to_z_mapping}). Below, we derive an upper bound for this relative degeneracy, fixing only the total size of the removed clusters, $\abs{S}$.

If the iteration procedure needs $R$ steps to reach $\vec z(y)$ from $\vec y$, and at each step we remove a set of size $\abs{S_i}$ then the number of $\vec y$ corresponding to the same $\vec z$ is at most
\begin{align}
 	\mathcal N(\{\abs{S_i}\})&\leq n \Phi^{\abs{S_1}} \cdot (n - \abs{S_1})\Phi^{\abs{S_2}}\cdot \dots \\
 		&\leq \Phi^{\abs{S}} n^R \\
 		&\leq \Phi^{\abs{S}} n^{2\abs{S} / \delta(n)}.
\end{align}
In the first line, we used Corollary 28 of Ref. \onlinecite{Fawzi2017efficient}, which upper bounds the total number of connected clusters of size $L$ in a graph with $n$ vertices and degree at most $w$ by $n\Phi^{L}$, with $\Phi=(w-1)(1+\tfrac{1}{w-2})^{w-2}$.
In the last line we have used that by definition $\abs{S_i} > \delta(n)/2$ and hence $R < 2 \abs{S} / \delta(n)$.

For a given $\abs{S}$, we can further upper bound the total number of sequences of sizes $\abs{S_i}$ by the number of integer partitions of $\abs{S}$, which is in turn upper bounded by $\Pi(\abs{S}) <  e^{\pi \sqrt{2\abs{S}/3}}$ \cite{erdos1942partition}.
The total relative degeneracy is then upper bounded by
\begin{align}
	\mathcal N(\abs{S}) &\leq \exp(
	\pi\sqrt{\tfrac{2}{3}\abs{S}}
	+\abs{S}\log \Phi
	+2\abs{S} \tfrac{\log n}{\delta(n)}) \\
	&\leq \exp[\left(\pi + \log\Phi + 2\tfrac{\log n}{\delta(n)}\right)\abs{S}]
    \label{eq:rel_degeneracy}
\end{align}
which depends only on the total size of all clusters removed in the reduction procedure, $\abs{S}$.

\vspace{1em}
{\bf Bounding the bottleneck ratio}
\\[0.5em]
Putting the lower bound on the energy barriers [\cref{eq:deltaE_reduction}] and the upper bound on the relative degeneracy [\cref{eq:rel_degeneracy}] together, we can now conclude that
\begin{align}
	\frac{\muG[\partial\Omega]}{\muG[\Omega]} &= 
		\frac{\sum_{\vec x_0 + \vec y\in \partial\Omega} e^{-\beta E(\vec x_0 + \vec y)}}
			{\muG[\Omega]} \\
			&\leq \frac{1}{\muG[\Omega]}\sum_{\vec x_0 + \vec y \in \partial \Omega} e^{-\beta [E(\vec x_0 + \vec z(\vec y)) + \tfrac{1}{2}(\gamma - \epsilon)\abs{S(\vec y)}]} \\ 
			&\leq \frac{\sum_{\vec w\in\Omega} e^{-\beta E(\vec w)}\sum_{\abs{S}=\delta(n)/2}^n\mathcal N(\abs{S})  e^{-\beta \frac{\gamma - \epsilon}{2}\abs{S}}}{\sum_{\vec w\in\Omega} e^{-\beta E(\vec w)}}  \\
			&\leq \sum_{\abs{S}=\delta(n)/2}^\infty
				e^{[\pi + \log \Phi+2\log n/\delta(n)-\beta (\gamma - \epsilon)/2]\abs{S}} \\
			&\leq e^{[1 + \log \Phi+2\log n/\delta(n)-\beta (\gamma - \epsilon)/2]\delta(n)/2}
\end{align}
In the second line, we have used the lower bound on the relative energy difference between matched states $\vec z(\vec y)$ and $\vec y$ [\cref{eq:deltaE_reduction}].
In the third and fourth line, we have used the upper bound on the number of states $\vec y$ reduced to the the same $\vec z$, summing over all possible reduction sizes $\abs{S}$.
In the last line, we used that at sufficiently low temperature, the sum over $\abs{S}$ is a convergent geometric series (starting from $\delta(n)/2$).
In particular, for sufficiently large $\beta$, since $\log n/\delta(n) \to 0$ by assumption, the last line vanishes exponentially in $\delta(n)$ as $n\to\infty$ and the result follows.
\end{proof}

\section{The configurational entropy of the Gibbs state is finite and grows with temperature}\label{suppl:Sconf}

In \cref{suppl:ergo_breaking}, we have established the existence of Gibbs state components, satisfying a bottleneck condition, around typical low-energy eigenstates $\ket{\vec x_0, \vec z_0}$. 
By virtue of the bottleneck theorem \cref{thm:bottleneck} and the discussion in \cref{suppl:gibbs_states}, each of these subspaces hosts at least one stable ergodic component of the Gibbs state.
In this section, we want to upper bound the total weight carried by any such component, which in turn provides a lower bound on the configuration entropy of the global Gibbs state as defined in \cref{suppl:gibbs_states}.

\subsection{Lower bounding the configurational entropy using the weight of typical components}

As discussed in \autoref{suppl:gibbs_states}, for the stabilizer codes we consider we can relate the weight of extremal Gibbs state components to the decomposition in \autoref{eq:StructureTheorem} as $w_i \equiv \tr(P_{\mathcal{V}_i}\rhoG) = 2^{-k}p_i$. The configurational entropy, defined in \autoref{eq:conf_ent_vN}, then takes the form
\begin{equation}
    \Sconf(\beta) = -\sum_i p_i \log p_i - \sum_i p_i \tr(\mu_i \log \mu_i)
    = 2^k \sum_i -w_i \log w_{i}.
\end{equation}

Further, the extremal-state decomposition is in terms of spans of eigenvectors, i.e. the projectors $\mathcal V$ are diagonal in the eigenbasis of $\hamil$. In this case, \cref{thm:botteleneck2} guarantees that for qLDPC codes with linear confinement 
\begin{equation}
    2^k\sum_{\text{typical}~i} \tr P_{\mathcal V_i} \rhoG 
    = 2^k\sum_{\text{typical}~i} w_{i}
    \geq 1 - e^{-\Omega(\delta(n))}
    \label{eq:sum_typical_pi}
\end{equation}
where the sum goes over ``typical'' extremal components with associated weights $w_{i}$ within one symmetry sector, which are supported on subspaces of the form defined in \autoref{eq:gibbs_subspaces}. By \autoref{eq:sum_typical_pi}, we then know that 
\begin{align}
    \Sconf(\beta) &= -2^k\sum_i w_i \log w_i \\
        &= -2^k\sum_{\text{rare}~j} w_j \log w_j
        - 2^k\sum_{\text{typical}~i} w_i \log w_i \\
        &\geq - 2^k\sum_{\text{typical}~i} w_i \log w_i \\
        &\geq \left(2^k\sum_{\text{typical}~i} w_i  \right) \log(\frac{1}{w_{\rm max}^{(\rm typ)}}) \\
        &\geq \left( 1 - e^{-\Omega(\delta(n))} \right) \log(\frac{1}{w_{\rm max}^{(\rm typ)}})
\end{align}
where $w_{\rm max}^{(\rm typ)}$ is the maximum weight of any ``typical'' component.
The above implies in particular, that $\forall \kappa >0$, $\exists n^* < \infty$ such that $\forall n > n^*$ :
\begin{equation}\label{eq:lower_bound_sconf_pmax}
    \Sconf > (1-\kappa) \log(\frac{1}{w_{\rm max}^{(\rm typ)}})
\end{equation}

\subsection{Upper bounding the weight of typical components}

Following the discussion above, to lower bound the asymptotic scaling of $\Sconf$, it thus suffices to \emph{upper} bound the weight of typical components 
\begin{equation}
    w_{\rm max}^{(\rm typ)} = \tr \mathcal V_{\rm typ} \rhoG 
        = Z^{-1} \tr \mathcal P_{\mathcal V_{\rm typ}} e^{-\beta \mathcal H}
    \label{eq:p_max_typ}
\end{equation}
where $Z = \tr e^{-\beta \hamil}$ is the partition function, and $\mathcal V_{\rm typ}$ is a subspace as defined in \autoref{eq:gibbs_subspaces}.

Assuming the absence of redundancies, we can derive such an upper bound by only using linear confinement. In particular, we show the following

\begin{theorem}\label{thm:upper_bound_typical_pi}
    Consider a qLDPC code with $(\delta, \gamma)$-confinement, defined by two full rank parity check matrices $H_X$ and $H_Z$, and associated Hamiltonian~$\hamil$. 
    There exists $\beta^* \in \mathbb{R}_+$, such that for all $\beta > \beta^*$, 
    for a eigenstate $\ket{\vec x_0, \vec z_0}$ chosen randomly form the Gibbs ensemble $\rhoG = Z^{-1} e^{-\beta \hamil}$, with high probability $p \geq 1 - e^{-\Omega(\delta(n))}$, the Gibbs state supported on the subspace $\mathcal V(\vec x_0, \vec z_0)$, defined in \autoref{eq:gibbs_subspaces}, contains only an exponentially small fraction of the weight: $\forall \kappa > 0$, $\exists n^* < \infty$ such that $\forall n > n^*$
    \begin{equation}
        \log\tr P_{\mathcal V}\rhoG \leq -(1+\kappa)[r + f(T)]\,n
    \end{equation}
    where $r = k/n$ is the rate of the code, and
    \begin{equation}
        f(T) = 
        (1-r)\left[
        \log(1 + e^{-1/T}) 
        + \frac{1}{T(1+e^{1/T})} 
        \right]
        - \log\Upsilon\left(\tfrac{2 (1-r)}{\gamma(1+e^{1/T})}\right)
    \end{equation}
    with $\Upsilon(\rho) = \rho^{-\rho} (1-\rho)^{\rho-1}$.
\end{theorem}
\begin{proof}

Let us first sketch the proof strategy. 
We want to upper bound $\tr \mathcal V_{\rm typ} \rhoG = Z^{-1} \tr \mathcal P_{\mathcal V_{\rm typ}} e^{-\beta \mathcal H}$.
First, the denominator $Z^{-1}$ can be computed exactly.
For the numerator, we first use Hoefferding's inequality to relate the quantity $\tr \mathcal V_{\rm typ} e^{-\beta \hamil}$ to the number of states in $\mathcal V_{\rm typ}$ below a certain energy cutoff, and then use linear confinement to upper bound that number using a simple counting argument.

\vspace{1em}
{\bf Partition function}
\\[0.5em]
We can write the partition function in closed form as
    \begin{equation}
        Z(\beta) = \sum_{\ell} e^{\beta E_{\ell}}
        = 2^k \sum_{\vec s} e^{-\beta \abs{\vec s}}
        = 2^k \left(1+e^{-\beta}\right)^m
        = 2^{rn} \left(1+e^{-\beta}\right)^{(1-r)n}
    \label{eq:partition_function_no_redundancies}
    \end{equation}
    where $E_{\ell}$ are the eigenenergies of $\hamil$, which are given by the Hamming weight of the corresponding syndrome $\vec s\in\mathbb{F}_ 2^m$, and because there are no redundancies the partition function reduces to a sum over all binary vectors of length $m$, where $m$ is the number of stabilizers of the code. As above, $r=k/n$ denotes the rate of the code.

\vspace{1em}
{\bf Upper bounding the projecton of the Boltzmann factor onto typical subspaces}
\\[0.5em]
Consider choosing a random state $\ket{\vec x_0, \vec z_0}$ from the Gibbs ensemble $\rhoG = Z^{-1} e^{-\beta \hamil}$.
By \cref{thm:botteleneck2}, we can choose $\beta^*$ such that for $\beta > \beta^*$ the Gibbs states supported on $\mathcal V(\vec x_0, \vec z_0)$ are stable.

In the absence of redundancies, the energy of a random eigenstate is the sum of i.i.d variables $s_j \in \{0, 1\}$. In this case Hoefferding's inequality states that
\begin{equation}\label{eq:hoefferding}
    {\rm Prob}\left(\abs{E - \expval{E}} \geq t \right) \leq 2 e^{-2t^2}.
\end{equation}
and hence we can further choose $c_1$ such that with probability $p = 1 - e^{-c_1 \delta(n)}$ we have also $\ket{\vec x_0, \vec z_0} \in \Xi$ where $\Xi$ is defined as an energy shell of with $\sqrt{\xi n}$ around the mean $\expval{E} := \varepsilon n$
\begin{equation}\label{eq:energy_shell}
    \Xi := \linspan\left\{
    \ket{\ell} \in \hilbert, \hamil\ket{\ell} = E_{\ell} \ket{\ell} ~\vert~ 
    \expval{E}_{\beta} - \sqrt{\xi n} \leq 
    E_{\ell}
    \leq \expval{E}_{\beta} + \sqrt{\xi n}
    \right\}.
\end{equation}

We can further conclude that 
\begin{align}
    \tr \mathcal V_{\rm typ} e^{-\beta \hamil} 
    = \sum_{\ket{\ell}\in\mathcal V} e^{-\beta E_{\ell}} 
    \leq \sum_{\ket{\ell} \in \mathcal V \cap \Xi} e^{-\beta E_{\ell}} + 2 e^{-2\xi n} 
    \leq \abs{\mathcal V \cap \Xi}\,e^{-\beta \expval{E}_{\beta} + \beta \sqrt{\xi n}} + 2 e^{-2\xi n}
\end{align}
where we have dropped the explicit dependence of $\mathcal V(\vec x_0, \vec z_0)$ on $\vec x_0$, $\vec z_0$ for brevity, and $\abs{\mathcal V \cap \Xi}$ denotes the dimension of the subspace $\mathcal V \cap \Xi$.

Taking the logarithm on both sides yields
\begin{align}
    \log \tr \mathcal V e^{-\beta \hamil} 
    &\leq
    \log \abs{\mathcal V \cap \Xi} - \beta \varepsilon n + \beta \sqrt{\xi n} + \log(1 + \frac{2e^{-\xi n}}{\abs{\mathcal V \cap \Xi}e^{-\beta \expval{E} + \beta\sqrt{\xi n}}}) \\
    &\leq
    \log \abs{\mathcal V \cap \Xi} - \beta \varepsilon n + \beta \sqrt{\xi n} + \frac{2e^{-2\xi n}}{\abs{\mathcal V \cap \Xi}e^{-\beta \expval{E} + \beta\sqrt{\xi n}}} \\
    &\leq
    \log \abs{\mathcal V \cap \Xi} - \beta \varepsilon n + \beta \sqrt{\xi n} + 2 e^{-\beta(2\xi-\varepsilon)n}
    \label{eq:log_weight_typ_subspace}
\end{align}
where the last error term $e^{-\beta (2\xi - \varepsilon)n}$ is exponentially small in $n$ for $\xi > \varepsilon / 2$.

We now need to upper bound the number of states in $\mathcal V \cap \Xi$. To this end, we define
\begin{equation}
   \Lambda := \linspan\left\{
    \ket{\ell} \in \hilbert, \hamil\ket{\ell} = E_{\ell} \ket{\ell} ~\vert~ 
    E_{\ell}
    \leq \expval{E}_{\beta} + \sqrt{\xi n}
    \right\}
\end{equation}
where naturally $\abs{\mathcal V \cap \Xi} \leq \abs{\mathcal V \cap \Lambda}$.

Consider the energy of $\ket{\ell} := \ket{\vec x_0 + \vec x, \vec z_0 + \vec z} \in \mathcal V$ relative to that of $\ket{\ell_{\rm ref}} = \ket{\vec x_0, \vec z_0}$, where $\ket{\ell_{\rm ref}}$ is the eigenstate with respect to which $\mathcal V$ is defined (cf. \autoref{eq:gibbs_subspaces})

Since $\maxConn(\vec x\reduced) \leq \maxConn_{\alpha=1/2}(\vec x\reduced) \leq \tfrac{1}{2}\delta(n)$, we know that $\abs{H_Z\vec x} > \gamma \norm{\vec x}$, and the equivalent statement is also true for $\vec z$.
This implies in particular that $\abs{H_z \vec x} + \abs{H_x \vec z} > 2 \varepsilon n$ if $\norm{\vec x} + \norm{\vec z}> 2\varepsilon n/\gamma$ and hence 
\begin{equation}
    E_{\ell} - E_{\rm ref} = \abs{H_Z\vec x} + \abs{H_X\vec z} - 2\abs{H_Z \vec x_0 \wedge H_Z \vec x} - 2\abs{H_X \vec z_0 \wedge H_X \vec z} > 0.
\end{equation}
Since $\ket{\ell_{\rm ref}} \in \Xi$, and hence $E_{\rm ref}/ n \to \varepsilon$ as $n\to\infty$, we can upper bound $\abs{\mathcal V \cap \Lambda}$ in terms of the volume of a Hamming Ball:
\begin{equation}
    S_n(\rho) = \sum_{l=0}^{\rho n} {n \choose l}
\end{equation}
such that 
\begin{equation}
    \abs{\mathcal V \cap \Lambda} 
    \leq S_n\left(\frac{2 \varepsilon}{\gamma}\right).
\end{equation}
We will use the following upper bound \cite[Lemma 3.3]{worsch1994_binomial_bounds} for $S_n(\rho)$: $\forall \kappa >0$, $\exists n^*$, such that for $\rho n > n^*$
    \begin{equation}
        S_n(\rho) 
        \leq
        (1+\kappa)\frac{1}{\sqrt{2\pi}}\sqrt n \sqrt{\frac{\rho}{1-\rho}} \Upsilon(\rho)^n
    \end{equation}
with $\Upsilon(\rho) = \rho^{-\rho} (1-\rho)^{\rho-1}$.

Plugging this into \autoref{eq:log_weight_typ_subspace}, we obtain that $\forall \kappa > 0$, $\exists n^* < \infty$ such that $\forall n > n^*$
\begin{align}
    \log \tr \mathcal V e^{-\beta \hamil} 
    &\leq
    (1 + \kappa)\,n \left( \log\Upsilon\left(\tfrac{2 \varepsilon}{\gamma}\right)  - \beta \varepsilon\right)
    \label{eq:upper_bound_expH_projected}
\end{align}

\vspace{1em}
{\bf Lower-bounding the configurational entropy}
\\[0.5em]
Combining \autoref{eq:partition_function_no_redundancies} and \autoref{eq:upper_bound_expH_projected} then yields the fact that $\forall \kappa > 0$, $\exists n^* < \infty$ such that $\forall n > n^*$
    \begin{equation}
    \log \tr P_{V(\vec x_0, \vec z_0)} \rhoG \leq -(1+\kappa)n \biggl(
        r + \underbrace{
        (1-r)\log(1+e^{-\beta}) 
        + \beta \varepsilon
        -\log\Upsilon\left(\tfrac{2 \varepsilon}{\gamma}\right)
        }_{:= f(\beta)}
    \biggr)
    \end{equation}
where $\Upsilon(\rho) = \rho^{-\rho} (1-\rho)^{\rho-1}$.

To write $f(T)$ explicitly as a function of only the temperature, we compute the expectation value of the energy explicitly
\begin{equation}
    \varepsilon = \frac{1}{n}\expval{E}_\beta
        =\frac{1-r}{1+e^{\beta}}
\end{equation}
so that with $T=1/\beta$
\begin{equation}
    f(T) = 
        (1-r)\left[
        \log(1 + e^{-1/T}) 
        + \frac{1}{T(1+e^{1/T})} 
        \right]
        - \log\Upsilon\left(\tfrac{2 (1-r)}{\gamma(1+e^{1/T})}\right).
\end{equation}
    
\end{proof}

Note that using \autoref{eq:lower_bound_sconf_pmax}, the above in particular implies that $\forall \kappa > 0$, $\exists n^* < \infty$ such that $\forall n > n^*$
\begin{align}
    \Sconf &\geq 
    (1-\kappa)\,n \left( 
    r + (1-r)\log(1+e^{-\beta}) 
    + \beta \varepsilon -\log\Upsilon\left(\tfrac{2 \varepsilon}{\gamma}\right) 
    \right).
\end{align}
which we can also again write explicitly as a function of only the temperature
\begin{align}\label{eq:sconf_lower_bound}
    \Sconf
    &\geq (1-\kappa)\,n \left( 
    r + (1-r)\left(\log(1+e^{-\beta}) 
    + \frac{\beta}{1+e^{\beta}}
    \right) 
    -\log\Upsilon\left(\tfrac{2 (1-r)}{\gamma(1+e^{\beta})}\right) 
    \right).
\end{align}
with as before $\Upsilon(\rho) = \rho^{-\rho} (1-\rho)^{\rho-1}$.
We show this lower bound on the configurational entropy in \autoref{fig:sconf_lower_bound}, both as a function of expansion parameter $\gamma$ for a fixed temperature and a range of code rates $r$, and as a function of temperature, at fixed code rate and a range of expansion parameters $\gamma$.
While the bound is only nontrivial at sufficiently large $\gamma$, the values necessary to ensure a nontrivial lower bound are not very large, and can easily be achieved using known constructions (see \autoref{suppl:proof_profit}).

\begin{figure}
    \centering
    \includegraphics{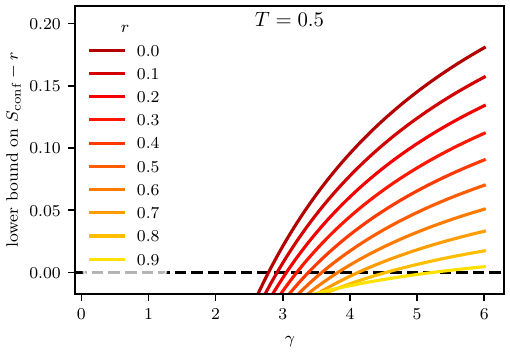}
    \includegraphics{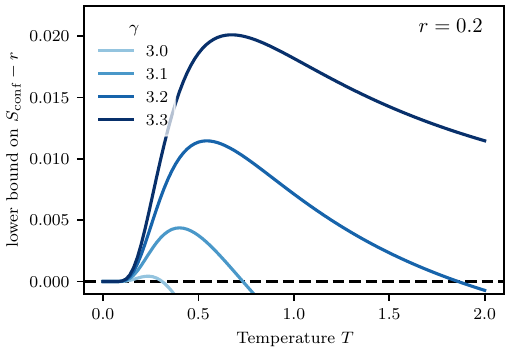}
    \caption{Lower bounds on the the configurational entropy \emph{per symmetry sector}, $\Sconf$, as derived from expansion in \autoref{eq:sconf_lower_bound}
    Left: lower bound for fixed temperature $T$ as a function of expansion parameter $\gamma$ and different code rates $r=k/n$. 
    Right: lower bound for fixed code and expansion parameter, as a function of temperature $T$.
    }
    \label{fig:sconf_lower_bound}
\end{figure}

%
%
%
%

\section{Typical extremal Gibbs states are non-separable\label{suppl:topo_order}}

In this section, we show that for any finite-rate qLDPC code Hamiltonian with linear $(\delta, \gamma)$-confinement, the Gibbs state components $\rhoGV$ around typical eigenstates are long-range entangled (LRE). In fact, we will prove an even stronger statement, that is that the subspace $\mathcal V$ associated to a typical component does not contain short-ranged entangled (SRE) states.

The distinction between long- and short-ranged entanglement here is based on circuit complexity. We call a pure state $\ket{\psi}$ short range entangled (SRE) if it can be prepared using a finite-depth circuit circuit of local (i.e. acting on a finite number of sites) unitary gates, and we call the state long ranged entangled (LRE) otherwise. Similarly, a mixed state $\hat\rho$ is called short-range entangled if it can be approximated by a classical mixture of SRE pure states, i.e. when $\hat\rho\approx\sum_j p_j \ketbra{\phi_j}{\phi_j}$ for a set of SRE states $\{\phi_j\}$, and we call it long range entangled otherwise.

The proof strategy, which was already sketched in the man text but is repeated here for completeness, is as follows. By construction, the space $\mathcal V$ associated to a typical component as defined in \cref{eq:gibbs_subspaces} is well-separated by its image under a logical operator. In particular, define $\mathcal V_{ \ell}:=\hat L_{\ell}\mathcal V$, where $\hat L_{\ell}$ $\ell=1\dots 2^k$ are the logical operators of the code including logical identity. 
Note that by the definition of $\mathcal V$ in \cref{eq:gibbs_subspaces}, and \cref{lem:maxConn_alpha_maxConn}, and the fact that $\maxConn(\hat L) > \delta(n)$, we know that $\bra{\psi_{l}} \hat{\mathcal O}\ket{\psi_{m}} = \delta_{lm}$ for all Pauli operators with $\norm*{\hat{\mathcal O}} \leq \tfrac{1}{4}\delta(n)$. 
This suffices to show that for any state $\ket{\psi_1} \in \mathcal V$, we can choose $2^k-1$ suitable states $\ket{\psi_{\ell}} \in \mathcal{V}_\ell$, $\ell=2,\dots,2^k$ such that
\begin{equation}
    \mathcal D := \linspan\{ \ket{\psi_{\ell}}, \ell=1,\dots,2^k\}
\end{equation}
defines a code with $k$ logical qubits and distance $d > \tfrac{1}{4}\delta(n)$.

We can then use the following recent result on the circuit complexity of quantum error correcting codes.
\begin{theorem}[Theorem 3 (i) in Ref.\ \cite{bravyi2024entanglement}]\label{thm:codes_lre}
    For any code $\mathcal D$ on $n$ qubits, with constant rate $\dim(\mathcal D)/n \xrightarrow[n\to\infty]{} r >0$ and distance $d$, then any logical state $\ket{\psi}\in \mathcal D$, $\ket{\psi}$ cannot be prepared by a unitary circuit of depth less than $\order{\log_2 d}$. 
\end{theorem}
When applied to the code $\mathcal D$ defined above, it implies that $\ket{\psi_1}$ cannot be prepared by a unitary circuit of depth less than $\order{\delta(n)}$.
Since the choice of $\ket{\psi}\in \mathcal V$ here was arbitrary, however, this already implies that $\mathcal V$ contains \emph{no} SRE states at all.

To construct the code $\mathcal D$, we will use the following Lemma.
\begin{lemma}
    Consider an error correcting code $\mathcal C$, and two sets of operators $\{E_\alpha\}$, $\{F_{\sigma}\}$, such that $E_{\alpha}F_{\sigma}$ is a correctable error for all $\alpha, \sigma$.
    Let $\ket{\ell}$, $\ell=1\dots 2^k$ be an orthogonal basis of the codespace $\mathcal C$.
    Then, the set $\{F_{\sigma}\}$ is correctable in the code $\mathcal D := \linspan\left\{ \ket{\psi_{\ell}} := \sum_{\alpha} c_\alpha E_{\alpha}\ket{\ell} \right\}$ for arbitrary coefficients $c_\alpha$ such that $\sum_\alpha \abs{c_{\alpha}}^2=1$
    \label{lem:superposition_codes}
\end{lemma}

Note that the code $\mathcal D$ is not necessarily a CSS, LDPC, or even stabilizer code.
Then, recall that for a general code $\mathcal C$ correctable errors $\{E_{\alpha}\}$ are defined as those fulfilling the Knill-Laflamme conditions \cite{knill1997theory} (see also \cite[Chapter~10]{NielsenChuang2010}):
\begin{equation}
    \bra{\psi_l} E_{\alpha}^\dagger E_{\beta} \ket{\psi_m}
    = \delta_{lm} \gamma_{\alpha\beta}
\end{equation}
where $\ket{\phi_l}, \ket{\psi_m}\in \mathcal C$ are two code words, $\delta_{lm}$ is the Kronecker-delta and $\gamma_{\alpha\beta}$ is a Hermitian matrix. 
Intuitively, the above states that pairs of correctable errors are not allowed to have matrix elements between different code words, but also the code words must be \emph{indistinguishable} by the $E_{\alpha}$: the expectation value of any pair must be identical with respect to all codewords.

\begin{proof}
We explicitly check that the Knill-Laflamme conditions are fulfilled. 
\begin{align}
    \bra{\psi_{l}}F^{\dagger}_\sigma F_{\omega}\ket{\psi_{m}} 
        &= \sum_{\alpha\beta} c_\alpha^*c_\beta \bra{l} E_{\alpha}^\dagger F^{\dagger}_\sigma F_{\omega} E_{\beta} \ket{m} \\
        &= \delta_{lm} \sum_{\alpha\beta} c_\alpha^*c_\beta \gamma_{(\alpha\sigma)(\omega\beta)}\\
        &= \delta_{lm} \Gamma_{(\alpha\sigma)(\omega\beta)}
\end{align}
where in the second line, we used the Knill-Laflamme conditions of the original code $\mathcal C$.
\end{proof}

We are now ready to prove the main result of this section, which will be the following theorem

\begin{theorem}\label{thm:local_minimum_nlts}
    Consider a qLDPC code with finite rate, linear $(\delta, \gamma)$-confinement, and with associated hamiltonian~$\hamil$. 
    Let $\ket{\vec x_0, \vec z_0}$ be an arbitrary eigenstate of $\hamil$ and $\mathcal V(\vec x_0, \vec z_0)$ the subspace defined in \cref{eq:gibbs_subspaces}. 
    Then any state $\ket{\psi}\in\mathcal V(\vec x_0, \vec z_0)$ cannot be prepared by a unitary circuit of depth less than $\order{\log_2 \delta(n)}$
\end{theorem}

\begin{proof}

Recall that the subspace $\mathcal V$ in \cref{eq:gibbs_subspaces} is defined as 
\begin{equation}
    \mathcal V = \linspan \left\{ 
        X^{\vec x} Z^{\vec z} \ket{\vec x_0, \vec z_0} ~\vert~
        \maxConn_{\alpha=1/2}(\vec x\reduced) < \tfrac{1}{2}\delta(n)
        ~\text{and}~
        \maxConn_{\alpha=1/2}(\vec z\reduced) < \tfrac{1}{2}\delta(n)
    \right\}
\end{equation}
Since the code that we started with is a stabilizer code, any set of eigenstates with the same syndrome forms a basis for a code with identical properties. 
In particular, defining a code $\mathcal C' = \linspan\{\ket{\ell} := \hat L_{\ell} \ket{\vec x_0, \vec z_0}, \ell=1\dots 2^k\}$, and two sets of errors
\begin{align}
    \{E_\alpha\} &:= \left\{
        P= X^{\vec x} Z^{\vec z} ~\vert~
        \maxConn_{\alpha=1/2}(\vec x\reduced) < \tfrac{1}{2}\delta(n)
        ~\text{and}~
        \maxConn_{\alpha=1/2}(\vec z\reduced) < \tfrac{1}{2}\delta(n)
        \right\} \\
    \{F_\sigma\} &:= \left\{ 
        P= X^{\vec x} Z^{\vec z} ~\vert~
        \abs{P}\leq \tfrac{1}{4}\delta(n)
        \right\}
\end{align}
where $\abs{P} := \abs{\vec x} + \abs{\vec z}$
and by construction $E_{\alpha}F_{\sigma}$ is correctable for all $\alpha, \sigma$.

Now, since $\mathcal V = \linspan\{E_\alpha \ket{\vec x_0, \vec z_0}\}$, for any state $\ket{\psi}\in\mathcal V$ we can define a code $\mathcal D$ as in \cref{lem:superposition_codes} such that $\ket{\psi}\in \mathcal D$ and the set $F_{\sigma}$ is correctable in this code. This means the code $\mathcal D$ has distance $d_{\mathcal D} > \tfrac{1}{4}\delta(n)$, and the result follows from \cref{thm:codes_lre}.
\end{proof}

The no-trivial state topological order of extremal Gibbs states around typical eigenstates is then a simple corollary. 

\begin{corollary}[Typical extremal Gibbs states are non-separable]\label{col:nts-to}
    Consider a qLDPC code with constant rate, linear $(\delta, \gamma)$-confinement, and with associated hamiltonian $\hamil$.
    Then at sufficiently low temperatures, extremal Gibbs states around typical eigenstates are long range entangled. In particular they are supported on a subspace, $\mathcal V$, which contains not states of trivial circuit complexity.
\end{corollary}

By ``extremal Gibbs states around typical eigenstates'' we here mean that choosing a random eigenstate $\ket{\vec x, \vec z}$ at sufficiently low temperature, then with high probability [\cref{eq:typical_state}], the subspace $\mathcal V(\vec x, \vec z)$ contains at least one extremal Gibbs state by \cref{thm:botteleneck2}. All extremal Gibbs states contained in $\mathcal V(\vec x, \vec z)$ are non-separable since they are supported entirely on a subspace that contains no trivial states (\cref{thm:local_minimum_nlts}).

%
%
%
%

\section{Proof of the main result}\label{suppl:proof_profit}

\begin{proof}[Proof of \cref{thm:profit}]
    Consider first the properties of Gallager codes and their hypergraph products. 
    \cref{lem:gallager_no_redundancies}, states that Gallager codes for $\wcheck > \wbit \geq 2$ have no redundancies, that is $\rank H = m$, with high probability.
    For the $(n_0, 14, 15)$-ensemble this implies asymptotically a constant rate $r_0\to 1/15$
    Further, by \cref{lem:hgp_confinement} the hypergraph product of two codes from this ensemble is a qLDPC code that has rate $r = 1/421$, as well as boundary and co-boundary confinement with $\gamma > \tfrac{1}{2}(\wbit - 8) \geq 3$ and $\delta(n) > \sqrt{n}/421$.

    Property (1) is then fulfilled by the proof in \cref{suppl:ergo_breaking}, in particular we use the subspace defined in \cref{eq:gibbs_subspaces}, in which case the bottlneck ratio above reduces to a classical ratio \cref{eq:bottleneck_ratio_classical}, which in turn is bounded in \cref{thm:botteleneck2}, which just requires linear confinement with super-logarithmic $\delta(n)$. The constants $c_1$ and $c_2$ are identical to $c_1$ and $c_2$ in \cref{thm:botteleneck2}.

    \begin{figure}
        \centering
        \includegraphics{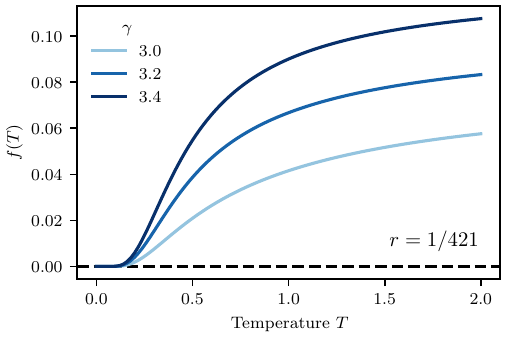}
        \caption{The function $f(T)$ as defined in the upper bound on the total weight of typical Gibbs states components in \cref{thm:upper_bound_typical_pi}, for the parameters relevant for the explicit instantiation of \cref{thm:profit} in \autoref{suppl:proof_profit}.}
        \label{fig:f_T_r_1_15}
    \end{figure}

    Property (2) follows directly from using \cref{thm:upper_bound_typical_pi}, which can be used because if two codes do not have redundancies, then their hypergraph product has no redundancies as well (cf. \cref{eq:hgp_def}).
    It is further easy to see that the function $f(T)$, as defined in \cref{thm:upper_bound_typical_pi}, is both positive and strictly increasing as a function of temperature $T$ for the parameters guaranteed by the construction, which we show explicitly in \autoref{fig:f_T_r_1_15}. Note that $f(T)$ at fixed $T$ increases as $\gamma$ is increased.
    
    Property (3) follows directly from the definition of $\mathcal V$ in \cref{eq:gibbs_subspaces} and \cref{thm:local_minimum_nlts}.

    Property (4) also follows directly from the definition of $\mathcal V$ in \cref{eq:gibbs_subspaces}, which fixes $c_3 = 1/4$, since $\mathcal V$ only includes correctable errors, and any local dynamics is confined to $\mathcal V$ for a time decaying as a (stretched) exponential in system size. 
\end{proof}

Note that while formally, we have shown the relevant properties only for the hypergraph product of Gallager codes with relatively large check weight, we expect topological quantum spin glass order to be realized much more generally in qLDPC codes with linear confinement. To this end, we note in particular that the bounds on linear confinement are expected to be quite loose, and that we have used other properties than linear confinement only for property (2), that is the lower bound on the configurational entropy which we also expect to be quite loose.

\bibliography{references.bib}